\documentclass[11pt]{llncs}

\pdfoutput=1

\newtheorem{observation}{Observation}
\newcommand{\LF}{\ensuremath{\mathit{LF}}}
\newcommand{\BWT}{\ensuremath{\mathit{BWT}}}
\newcommand{\DISA}{\ensuremath{\mathit{DISA}}}
\newcommand{\ISA}{\ensuremath{\mathit{ISA}}}
\newcommand{\SA}{\ensuremath{\mathit{SA}}}
\newcommand{\DSA}{\ensuremath{\mathit{DSA}}}
\newcommand{\LCP}{\ensuremath{\mathit{LCP}}}
\newcommand{\DLCP}{\ensuremath{\mathit{DLCP}}}
\newcommand{\PLCP}{\ensuremath{\mathit{PLCP}}}
\newcommand{\TDE}{\ensuremath{\mathit{TDE}}}

\newcommand{\PTDE}{\ensuremath{\mathit{PTDE}}}
\newcommand{\RMQ}{\ensuremath{\mathrm{RMQ}}}
\newcommand{\PSV}{\ensuremath{\mathrm{PSV}}}
\newcommand{\NSV}{\ensuremath{\mathrm{NSV}}}
\newcommand{\SV}{\ensuremath{\mathrm{SV}}}

\newcommand{\rank}{\ensuremath{\mathrm{rank}}}

\usepackage{fullpage}
\usepackage{graphicx}
\usepackage{color}
\usepackage{amsmath}
\usepackage{footmisc}

\begin{document}

\title{Fully-Functional Suffix Trees and Optimal Text Searching \\ in BWT-runs Bounded Space
\thanks{Partially funded by Basal Funds FB0001, Conicyt, by Fondecyt Grants
1-171058 and 1-170048, Chile, and by the Danish Research Council DFF-4005-00267.
An early partial version of this article appeared in {\em Proc. SODA 2018} \cite{GNP18}.}}
\author{Travis Gagie\inst{1,2} \and Gonzalo Navarro\inst{2,3} \and Nicola Prezza\inst{4}}
\institute{EIT, Diego Portales University, Chile 
\and 
Center for Biotechnology and Bioengineering (CeBiB), Chile
\and 
Department of Computer Science, University of Chile, Chile 
\and 
Department of Computer Science, University of Pisa, Italy}

\maketitle

\begin{abstract}
Indexing highly repetitive texts --- such as genomic databases, software 
repositories and versioned text collections --- has become an important problem
since the turn of the millennium. A relevant compressibility measure for 
repetitive texts is $r$, the number of runs in their Burrows-Wheeler Transforms
(BWTs). One of the earliest indexes for repetitive collections, the Run-Length
FM-index, used $O(r)$ space and was able to efficiently count the number of
occurrences of a pattern of length $m$ in the text (in loglogarithmic time per 
pattern symbol, with current techniques). However, it was unable to locate the
positions of those occurrences efficiently within a space bounded in terms of
$r$. Since then, a number of other indexes with space bounded by other measures
of repetitiveness --- the number of phrases in the Lempel-Ziv parse, the size
of the smallest grammar generating (only) the text, the size of the smallest automaton
recognizing the text factors --- have been proposed for efficiently locating,
but not directly counting, the occurrences of a pattern. In this paper we close
this long-standing problem, showing how to extend the Run-Length FM-index
so that it can locate the $occ$ occurrences efficiently within $O(r)$ space (in
loglogarithmic time each), and reaching optimal time, $O(m+occ)$, within 
$O(r\log\log_w(\sigma+n/r))$ space, for a text of length $n$ over an alphabet
of size $\sigma$ on a RAM machine with words of $w=\Omega(\log n)$ bits. 
Within that space, our index can also count in optimal time, $O(m)$. 
Multiplying the space by $O(w/\log\sigma)$, we support count and locate in 
$O(\lceil m\log(\sigma)/w\rceil)$ and $O(\lceil m\log(\sigma)/w\rceil+occ)$ 
time, which is optimal in the packed setting and had not been obtained before 
in compressed space. We also 
describe a structure using $O(r\log(n/r))$ space that replaces the text and 
extracts any text substring of length $\ell$ in almost-optimal time
$O(\log(n/r)+\ell\log(\sigma)/w)$. Within that space, we similarly provide 
direct access to suffix array, inverse suffix array, and longest common prefix 
array cells, and extend these capabilities to full suffix tree functionality, 
typically in $O(\log(n/r))$ time per operation. 
Our experiments show that our $O(r)$-space index outperforms the
space-competitive alternatives by 1--2 orders of magnitude.
\end{abstract}


\section{Introduction}

The data deluge has become a pervasive problem in most organizations that aim
to collect and process data. 
We are concerned about string (or text, or sequence) data, formed by collections
of symbol sequences. This includes natural language text collections, DNA and
protein sequences, source code repositories, semistructured text, and many 
others. The rate at which those sequence collections are growing is daunting, 
in some cases outpacing Moore's Law by a significant margin \cite{Plos15}. 
%
A key to handle this growth is the fact that the amount of {\em unique} material
does not grow at the same pace of the sequences. Indeed, the fastest-growing
string collections are in many cases {\em highly repetitive}, that is, most 
of the strings can be obtained from others with a few modifications.
For example, most genome sequence collections store many genomes from the
same species, which in the case of, say, humans differ by 0.1\% \cite{PHDR00}
(there is some discussion about the exact percentage). The 1000-genomes 
project\footnote{{\tt http://www.internationalgenome.org}} uses a 
Lempel-Ziv-like compression mechanism that reports compression ratios around 1\%
\cite{FLCB11} (i.e., the compressed space is about two orders of magnitude less
than the uncompressed space). Versioned document collections and software 
repositories are another natural source of repetitiveness. For example, 
Wikipedia reports that, by June 2015, there were over 20 revisions (i.e., 
versions) per article in its 10 TB content, and that {\tt p7zip}\footnote{\tt
http://p7zip.sourceforge.net} compressed it to about 1\%. They also report 
that what grows the fastest today are the revisions rather than the new 
articles, which increases repetitiveness.%
\footnote{{\tt https://en.wikipedia.org/wiki/Wikipedia:Size\_of\_Wikipedia}}
A study of GitHub (which surpassed 20 TB in 2016)%
\footnote{{\tt https://blog.sourced.tech/post/tab\_vs\_spaces}}
reports a ratio of {\em commit} (new versions) over {\em create} (brand new
projects) around 20.%
\footnote{{\tt http://blog.coderstats.net/github/2013/event-types}, see the
ratios of {\em push}/{\em create} and {\em commit}.{\em push}.}

Version management systems offer a good solution to the problem of
providing efficient {\em access} to the documents of a versioned collection, at
least when the versioning structure is known. They factor out repetitiveness 
by storing the first version of a document in plain form and then the edits of 
each version of it. It is much more challenging, however, to provide more
advanced functionalities, such as {\em counting} or {\em locating} the positions
where a string pattern occurs across the collection. 

An application field where this need is most pressing is bioinformatics.
The {\em FM-index}~\cite{FM05,FMMN07} was extremely successful in reducing the
size of classical data structures for pattern searching, such as suffix 
trees \cite{Wei73} or suffix arrays \cite{MM93}, to the {\em statistical} 
entropy of the sequence while emulating a significant part of their 
functionality. The FM-index has had a surprising impact far beyond the
boundaries of theoretical computer science: if someone now sends his or her
genome to be analyzed, it will almost certainly be sequenced on a machine
built by Illumina\footnote{{\tt https://www.illumina.com}. More than 94\% of the
human genomes in SRA \cite{KSL12} were sequenced by Illumina.}, which will 
produce a huge collection of quite short substrings of that genome, called 
reads. Those reads' closest matches will then be sought in a reference genome, 
to determine where they most likely came from in the newly-sequenced target 
genome, and finally a list of the likely differences between the target and 
the reference genomes will be reported. The searches in the reference genome 
will be done almost certainly using software such as {\tt Bowtie}\footnote{\tt
http://bowtie-bio.sourceforge.net}, {\tt BWA}\footnote{\tt
http://bio-bwa.sourceforge.net}, or {\tt Soap2}\footnote{\tt 
http://soap.genomics.org.cn}, all of them based on the FM-index.%
\footnote{Ben Langmead, personal communication.}

Genomic analysis is already an important field of research, and a rapidly
growing industry \cite{SL13}. As a result of dramatic advances in sequencing 
technology, we now have datasets of tens of thousands of genomes, and bigger 
ones are on their way (e.g., there is already a 100,000-human-genomes project%
\footnote{\tt https://www.genomicsengland.co.uk/the-100000-genomes-project}). 
Unfortunately, current software based on FM-indexes cannot handle such
massive datasets: they use 2 bits per base at the very least \cite{KS18}.
Even though the FM-index can represent the sequences within their statistical 
entropy \cite{FMMN07}, this measure is insensitive to the repetitiveness of 
those datasets \cite[Lem.~2.6]{KN13}, and thus the FM-indexes would grow 
proportionally to the sizes of the sequences. Using current tools, indexing
a set of 100{,}000 human genomes would require 75 TB of storage at the very 
least, and the index would have to reside in main memory to operate efficiently.
To handle such a challenge 
we need, instead, compressed text indexes {\em whose size is proportional to 
the amount of unique material} in those huge datasets.

\subsection{Related work}

M\"akinen et al.~\cite{MN05,MNSV08,MNSVrecomb09,MNSV09} pioneered the research 
on indexing and searching repetitive collections. They regard the collection 
as a single concatenated text $T[1..n]$ with separator symbols, and note that 
the number $r$ of {\em runs} (i.e., maximal substrings formed by a single 
symbol) in the {\em Burrows-Wheeler Transform} \cite{BW94} of the text is 
relatively very 
low on repetitive texts. Their index, {\em Run-Length FM-Index (RLFM-index)}, 
uses $O(r)$ words and can count the number of occurrences of a pattern 
$P[1..m]$ in time $O(m\log n)$ and even less. However, they are unable to 
locate where those positions are in $T$ unless they add a set of samples 
that require $\Theta(n/s)$ words in order to offer $O(s\log n)$ time
to locate each occurrence. On repetitive texts, either this sampled structure is
orders of magnitude larger than the $O(r)$-size basic index, or the locating
time is extremely high.

Many proposals since then aimed at reducing the locating time by building on
other compression methods that perform well on repetitive texts: indexes based 
on the Lempel-Ziv parse \cite{LZ76} of $T$, with size bounded in terms of
the number $z$ of phrases 
\cite{KN13,gagie2014lz77,NIIBT15,BCGPR15,Nav17,BEGV18,CE18};
indexes based on the smallest context-free grammar (or an approximation thereof)
that generates $T$ and only $T$ \cite{KY00,CLLPPSS05}, with size bounded in 
terms of the size $g$ of the grammar \cite{CNfi10,CNspire12,GGKNP12,Nav18}; and
indexes based on the size $e$ of the smallest automaton (CDAWG) \cite{BBHMCE87}
recognizing the substrings of $T$ 
\cite{BCGPR15,TGFIA17,BCspire17}. 
Table~\ref{tab:related} summarizes the pareto-optimal achievements.
We do not consider in this paper indexes based on other repetitiveness 
measures that only apply in restricted scenarios, such as those based on
Relative Lempel-Ziv \cite{KPZ10,DJSS14,BGGMS14,FGNPS17} or on alignments
\cite{NPCHIMP13,NPLHLMP13}.

\begin{table}[p]
\begin{center}
\begin{tabular}{l|c|c}
Index & Space & {\bf Count time} \\
\hline
Navarro~\cite[Thm.~6]{Nav18}~~
	& ~~$O(z\log(n/z))$~~
	& ~~$O(m\log n + m\log^{2+\epsilon} (n/z))$ \\
\hline
Navarro~\cite[Thm.~5]{Nav18}~~
	& ~~$O(g)$~~
	& ~~$O(m^2+m\log^{2+\epsilon} g)$ \\
\hline
M\"akinen et al.~\cite[Thm.~17]{MNSV09}~~
	& ~~$O(r)$~~
	& ~~$O(m(\frac{\log\sigma}{\log\log r} + (\log\log n)^2))$ \\
{\bf This paper (Lem.~\ref{lem:rlfm})}
	& $O(r)$ 
	& $O(m\log\log_w(\sigma+n/r))$ \\
{\bf This paper (Thm.~\ref{thm:optimal})}
& $O(r\log\log_w (\sigma+n/r))$
& $O(m)$ \\
{\bf This paper (Thm.~\ref{thm:optimal packed})}
& $O(r w\log_\sigma\log_w(\sigma+n/r))$ 
& $O(\lceil m\log(\sigma)/w\rceil)$ \\
\end{tabular}
\end{center}

\begin{center}
\begin{tabular}{l|c|c}
Index & Space & {\bf Locate time} \\
\hline
Kreft and Navarro \cite[Thm.~4.11]{KN13} 
	& $O(z)$ 
	& $O(m^2 h+(m+occ)\log z)$ \\
Gagie et al.~\cite[Thm.~4]{gagie2014lz77} 
	& $O(z\log(n/z))$ 
	& $O(m\log m + occ \log\log n)$ \\
Bille et al.~\cite[Thm.~1]{BEGV18}
& $O(z\log (n/z))$
& ~~$O(m(1+\log^\epsilon z / \log(n/z))+occ(\log^\epsilon z+\log\log n))$ \\
Christiansen and Ettienne~\cite[Thm.~2(3)]{CE18}
& $O(z\log (n/z))$
& ~~$O(m+\log^\epsilon z + occ(\log^\epsilon z + \log\log n))$ \\
Christiansen and Ettienne~\cite[Thm.~2(1)]{CE18}
& $O(z\log (n/z)+z\log\log z)$
& ~~$O(m + occ(\log^\epsilon z + \log\log n))$ \\
Bille et al.~\cite[Thm.~1]{BEGV18}
& $O(z\log (n/z)\log\log z)$
& ~~$O(m+occ\log\log n)$ \\
\hline
Claude and Navarro~\cite[Thm.~1]{CNspire12} 
	& $O(g)$ 
	& $O(m^2\log\log_g n + (m+occ) \log g)$ \\
Gagie et al.~\cite[Thm.~4]{GGKNP12}
	& $O(g+z\log\log z)$
	& $O(m^2+(m+occ)\log\log n)$ \\
\hline
M\"akinen et al.~\cite[Thm.~20]{MNSV09} 
	& $O(r+n/s)$ 
	& $O((m+s\cdot occ)(\frac{\log\sigma}{\log\log r} + (\log\log n)^2))$ \\
Belazzougui et al.~\cite[Thm.~3]{BCGPR15}
	& $O(\overline{r}+z)$ 
	& $O(m(\log z + \log\log n) + occ (\log^\epsilon z + \log\log n))$ \\
{\bf This paper (Thm.~\ref{thm:locating})}
	& $O(r)$
	& $O((m+occ)\log\log_w(\sigma+n/r))$ \\
{\bf This paper (Thm.~\ref{thm:optimal})}
	& $O(r\log\log_w (\sigma+n/r))$
	& $O(m+occ)$ \\
{\bf This paper (Thm.~\ref{thm:optimal packed})}
& $O(r w \log_\sigma\log_w(\sigma+n/r))$
& $O(\lceil m\log(\sigma)/w\rceil+occ)$ \\
\hline
Belazzougui and Cunial \cite[Thm.~1]{BCspire17}~~
	& $O(e)$ 
	& $O(m+occ)$ \\
\end{tabular}
\end{center}

\begin{center}
\begin{tabular}{l|c|c}
Structure & Space & {\bf Extract time} \\
\hline
Kreft and Navarro \cite[Thm.~4.11]{KN13} 
	& $O(z)$ 
	& $O(\ell\, h)$ \\
Gagie et al.~\cite[Thm.~2]{BGGKOPT15}   
	& $O(z\log(n/z))$ 
	& $O((1+\ell/\log_\sigma n)\log(n/z))$ \\
\hline
Belazzougui et al.~\cite[Thm.~1]{BPT15} 
	& $O(g)$ 
	& $O(\log n + \ell/\log_\sigma n)$ \\
Belazzougui et al.~\cite[Thm.~2]{BPT15} 
	& $O(g\log^\epsilon n\log(n/g))$ 
	& $O(\log n/\log\log n + \ell/\log_\sigma n)$ \\
\hline
M\"akinen et al.~\cite[Thm.~20]{MNSV09} 
	& $O(r+n/s)$ 
	& ~~$O((s+\ell)(\frac{\log\sigma}{\log\log r} + (\log\log n)^2))$ \\
{\bf This paper (Thm.~\ref{thm:extract})}
	& ~~$O(r\log(n/r))$~~
	& $O(\log(n/r) + \ell\log(\sigma)/w)$ \\
\hline
Belazzougui and Cunial \cite[Thm.~1]{BCspire17}~~
	& $O(e)$ 
	& $O(\log n+\ell)$ \\
\end{tabular}
\end{center}

\begin{center}
\begin{tabular}{l|c|c}
Structure & Space & {\bf Typical suffix tree operation time} \\
\hline
M\"akinen et al.~\cite[Thm.~30]{MNSV09} 
	& $O(r+n/s)$ 
	& ~~$O(s(\frac{\log\sigma}{\log\log r} + (\log\log n)^2))$ \\
{\bf This paper (Thm.~\ref{thm:stree})}
	& ~~$O(r\log(n/r))$~~
	& $O(\log(n/r))$ \\
\hline
Belazzougui and Cunial \cite[Thm.~1]{BC17}~~
	& $O(\overline{e})$ 
	& $O(\log n)$ \\
\end{tabular}
\end{center}
\caption{Previous and our new results on counting, locating, extracting, and
supporting suffix tree functionality. We simplified some formulas with tight 
upper bounds. The variables are the text size $n$, pattern length $m$, number 
of occurrences $occ$ of the pattern, alphabet size $\sigma$, extracted length
$\ell$, Lempel-Ziv parsing size $z$, smallest grammar size $g$, $\BWT$ 
runs $r$, CDAWG size $e$, and machine word length in bits $w$. Variable 
$h\le z$ is the depth of the dependency chain in the Lempel-Ziv parse, and 
$\epsilon>0$ is an arbitrarily small constant. Symbols $\overline{r}$ or 
$\overline{e}$ mean $r$ or $e$ of $T$ plus $r$ or $e$ of its reverse. 
The $z$ of Kreft and Navarro \cite{KN13} refers to the Lempel-Ziv variant that 
does not allow overlaps between sources and targets, but their index actually 
works in either variant.}
\label{tab:related}
\end{table}


There are a few known asymptotic bounds between the repetitiveness measures 
$r$, $z$, $g$, and $e$: $z \le g = O(z\log(n/z))$ \cite{Ryt03,CLLPPSS05,Jez16}
and $e = \Omega(\max(r,z,g))$ \cite{BCGPR15,BC17}. 
Examples of string families are known that show that $r$ is not
comparable with $z$ and $g$ \cite{BCGPR15,Pre16}. Experimental results
\cite{MNSV09,KN13,BCGPR15,CFMPN16}, on the other hand,
suggest that in typical repetitive texts it holds $z < r \approx g \ll e$.

For highly repetitive texts, one hopes to have a compressed index not only
able to count and locate pattern occurrences, but also to {\em replace} the
text with a compressed version that nonetheless can efficiently {\em extract}
any substring $T[i..i+\ell]$. Indexes that, implicitly or not, contain a
replacement of $T$, are called {\em self-indexes}. As can be seen in
Table~\ref{tab:related}, self-indexes with $O(z)$ space require up to $O(z)$
time per extracted character, and none exists within $O(r)$ space. Good
extraction times are instead obtained with $O(g)$, $O(z\log(n/z))$, or $O(e)$
space. A lower bound for grammar-based representations \cite{VY13} shows that 
$\Omega((\log n)^{1-\epsilon}/\log g)$ time, for
any constant $\epsilon>0$, is needed to access one random position within 
$O(\mathrm{poly}(g))$ space. This bound shows that various current techniques
using structures bounded in terms of $g$ or $z$ 
\cite{BLRSRW15,BPT15,GGP15,BGGKOPT15} are nearly optimal (note that 
$g =\Omega(\log n)$, so the space of all these structures is 
$O(\mathrm{poly}(g))$). In an extended article \cite[Thm.~6]{CVY13}, the 
authors give a lower bound in terms of $r$, for binary texts on a RAM machine
of $w=\Theta(\log n)$ bits: $\Omega((\log n)^{1-\epsilon})$ for some constant
$\epsilon$ when using $O(\mathrm{poly}(r\log n))$ space.

In more sophisticated applications, especially in bioinformatics, it is 
desirable to support a more complex set of operations, which constitute a 
full suffix tree functionality \cite{Gus97,Ohl13,MBCT15}. While M\"akinen et 
al.~\cite{MNSV09} offered suffix tree functionality, they had the same problem
of needing $\Theta(n/s)$ space to achieve $O(s\log n)$ time for most suffix tree
operations. Only recently a suffix tree of size $O(\overline{e})$ supports
most operations in time $O(\log n)$ \cite{BCGPR15,BC17}, where $\overline{e}$
refers to the $e$ measure of $T$ plus that of $T$ reversed.

Summarizing Table~\ref{tab:related} and our 
discussion, the situation on repetitive text indexing is as follows.

\begin{enumerate}
\item The RLFM-index is the only structure able to count the occurrences of 
$P$ in $T$ in time $O(m\log n)$. However, it 
does not offer efficient locating within $O(r)$ space.
\item The only structure clearly smaller than the RLFM-index, using $O(z)$
space \cite{KN13}, has unbounded locate time. Structures using about the 
same space, $O(g)$, have an $\Omega(m^2)$ one-time overhead in the locate
time \cite{CNfi10,CNspire12,GGKNP12,Nav18}.
\item Structures offering lower locate times require 
$\Omega(z\log(n/z))$ space \cite{gagie2014lz77,NIIBT15,BEGV18,CE18,Nav18}, 
$\Theta(\overline{r}+z)$ space \cite{BCGPR15} (where $\overline{r}$ is the sum 
of $r$ for $T$ and its reverse), or 
$\Omega(e)$ space \cite{BCGPR15,TGFIA17,BCspire17}.
\item Self-indexes with efficient extraction require $\Omega(z\log(n/z))$ space
\cite{Ryt03,CLLPPSS05,GGP15,BGGKOPT15,BEGV18}, $\Omega(g)$ space 
\cite{BLRSRW15,BPT15}, or $\Omega(e)$ space \cite{TGFIA17,BCspire17}.
\item The only efficient compressed suffix tree requires $\Theta(\overline{e})$
space \cite{BC17}.
\item Only a few of all these indexes have been implemented, as far as we know
\cite{MNSV09,CNfi10,KN13,BCGPR15}.
\end{enumerate}

\subsection{Contributions}

Efficiently locating the occurrences of $P$ in $T$ within $O(r)$ space has been
a bottleneck and an open problem for almost a decade.
In this paper we give the first solution to this problem.
Our precise contributions, largely detailed in
Tables~\ref{tab:related} and \ref{tab:contrib}, are the 
following:

\begin{enumerate}
\item We improve the counting time of the RLFM-index to 
$O(m\log\log_w(\sigma+n/r))$, where $\sigma\le r$ is the alphabet size of $T$,
while retaining the $O(r)$ space.
\item We show how to locate each occurrence in time $O(\log\log_w(n/r))$,
within $O(r)$ space. We reduce the locate time to $O(1)$ per 
occurrence by using slightly more space, $O(r\log\log_w (n/r))$.
\item By using $O(r\log\log_w(\sigma+n/r))$ space, we obtain optimal locate 
time in the general setting, $O(m+occ)$, as well as optimal counting time,
$O(m)$. This had been obtained before only with space
bounds $O(e)$ \cite{BCspire17} or $O(\overline{e})$
\cite{TGFIA17}. 
\item By increasing the space to $O(r w\log_\sigma\log_w(\sigma+n/r))$, 
we obtain optimal locate time, $O(\lceil m\log(\sigma)/w\rceil +occ)$,
and optimal counting time, $O(\lceil m\log(\sigma)/w\rceil)$, in the packed 
setting (i.e., the pattern symbols come packed in blocks of $w/\log\sigma$ 
symbols per word). This had not been achieved so far by any compressed 
index, but only by uncompressed ones \cite{NN17}.
\item We give the first structure built on $\BWT$ runs that 
replaces $T$ while retaining direct access. It extracts any substring of 
length $\ell$ in time $O(\log(n/r)+\ell\log(\sigma)/w)$, using $O(r\log(n/r))$ 
space. As discussed, even the additive penalty is near-optimal
\cite[Thm.~6]{CVY13}. Within the same space, we also obtain optimal locating 
and counting time, as well as accessing subarrays of length $\ell$ of
the suffix array, inverse suffix array, and longest common prefix array of $T$,
in time $O(\log(n/r)+\ell)$.
\item We give the first compressed suffix tree whose space is bounded in 
terms of $r$, $O(r\log(n/r))$ words. It implements most navigation
operations in time $O(\log(n/r))$. There exist only comparable suffix trees
within $O(\overline{e})$ space \cite{BC17}, taking $O(\log n)$ time for most
operations.
\item We provide a proof-of-concept implementation of the most basic index
(the one locating within $O(r)$ space), and show that it outperforms all the
other implemented alternatives by orders of magnitude in space or in time to
locate pattern occurrences.
\end{enumerate}

\begin{table}[t]
\begin{center}
\begin{tabular}{l|c|c}
Functionality & Space (words) & Time \\
\hline
Count $+$ Locate (Thm.~\ref{thm:locating})
	& $O(r)$        & ~~$O(m\log\log_w(\sigma+n/r)+occ\log\log_w(n/r))$ \\
Count $+$ Locate (Lem.~\ref{lemma: general locate})
	& $O(r\log\log_w(n/r))$ & $O(m\log\log_w(\sigma+n/r)+occ)$ \\
Count $+$ Locate (Thm.~\ref{thm:optimal})
	& $O(r\log\log_w(\sigma+n/r))$  & $O(m+occ)$ \\
Count $+$ Locate (Thm.~\ref{thm:optimal packed})
& $O(r w \log_\sigma\log_w(\sigma+n/r))$  & $O(\lceil m\log(\sigma)/w\rceil+occ)$ \\
Extract (Thm.~\ref{thm:extract}) 
& $O(r\log(n/r))$ & $O(\log(n/r)+\ell\log(\sigma)/w)$ \\
Access $\SA$, $\ISA$, $\LCP$ (Thm.~\ref{thm:dsa}--\ref{thm:dlcp})~~ 
	& ~~$O(r\log(n/r))$~~ & $O(\log(n/r)+\ell)$ \\
Count $+$ Locate (Thm.~\ref{thm:optimal 2})
	& $O(r\log(n/r))$  & $O(m+occ)$ \\
Suffix tree (Thm.~\ref{thm:stree})
	& ~~$O(r\log(n/r))$~~ & $O(\log(n/r))$ for most operations \\
\end{tabular}
\end{center}
\caption{Our contributions. For any ``Count $+$ Locate'', we can do only
``Count'' in the time given by setting $occ=0$.}
\label{tab:contrib}
\end{table}

Contribution 1 is a simple update of the RLFM-index \cite{MNSV09} with newer 
data structures for rank and predecessor queries \cite{BN14}.
We present it in Section~\ref{sec:basics}, together with a review of the basic
concepts needed to follow the paper.

Contribution 2 is one of the central parts of the paper, and is obtained in
Section~\ref{sec:locate} in two steps. 
The first uses the fact that we can carry out the classical RLFM-index
counting process for $P$ in a way that we always know the position of one 
occurrence in $T$ \cite{Pre16,PP16}; we give a simpler proof of this fact
in Lemma~\ref{lem:find_one}.
The second shows that, if we know the position in $T$ of one occurrence of
$\BWT$, then we can quickly 
obtain the preceding and following ones with an $O(r)$-size
sampling. This is achieved by using the $\BWT$ runs to induce {\em phrases} in 
$T$ (which are somewhat analogous to the Lempel-Ziv phrases \cite{LZ76}) and 
showing that the positions of occurrences within phrases can be obtained from 
the positions of their preceding phrase start. The time $O(1)$ is obtained
by using an extended sampling. 

For Contributions 3 and 4, we use in Section~\ref{sec:optimal} the fact that the
RLFM-index on a text regarded as a sequence of overlapping metasymbols of 
length $s$ has at most $rs$ runs, so that we can process the pattern by chunks
of $s$ symbols. The optimal packed time is obtained by enlarging the samplings.

In Section~\ref{sec:sa}, Contribution 5 uses an analogue of the Block Tree 
\cite{BGGKOPT15} built on the $\BWT$-induced phrases, which satisfies the 
property that any distinct string has an occurrence overlapping a border 
between phrases. Further, it is shown that direct access to the
suffix array $\SA$, inverse suffix array $\ISA$, and array $\LCP$ of $T$,
can be supported in a similar way because they inherit the same repetitiveness 
properties of the text.

Contribution 6 needs, in addition to accessing those arrays, some sophisticated
operations on the $\LCP$ array \cite{FMN09} that are not well supported by 
Block Trees. In Section~\ref{sec:stree}, we implement suffix trees by showing
that a run-length context-free grammar \cite{NIIBT16} of size $O(r\log(n/r))$ 
can be built on the differential $\LCP$ array, and 
then implement the required operations on it.

The results of Contribution 7 are shown in Section~\ref{sec:experiments}. Our
experimental results show that our simple $O(r)$-space index outperforms the 
alternatives by orders of magnitude when locating the occurrences of a pattern,
while being simultaneously smaller or nearly as small. The only compact
structure outperforming our index, the CDAWG, is an order of magnitude larger.

Further, in Section~\ref{sec:construction} we describe construction algorithms 
for all our data structures, achieving construction spaces bounded in terms
of $r$ for the simpler and most practical structures. We finally conclude in
Section~\ref{sec:conclusion}.

This article is an extension of a conference version presented in {\em SODA
2018} \cite{GNP18}. The extension consists, on the one hand, in a significant 
improvement in Contributions 3 and 4: in Section~\ref{sec:optimal}, optimal 
time locating is now obtained in a much simpler way and in significantly less 
space. Further, optimal time is obtained as well, which is new.
On the other hand, Contribution 6, that is, the 
machinery to support suffix tree functionality in Section~\ref{sec:stree}, is
new. We also present an improved implementation in 
Section~\ref{sec:experiments}, with better experimental results. Finally, the 
construction algorithms in Section~\ref{sec:construction} are new as well.



\section{Basic Concepts}  \label{sec:basics}

A string is a sequence $S[1..\ell] = S[1] S[2] \ldots S[\ell]$, of length
$\ell = |S|$, of symbols (or characters, or letters) chosen from an alphabet 
$[1..\sigma] = \{ 1,2,\ldots,\sigma\}$, that is, $S[i] \in [1..\sigma]$ for all
$1\le i\le\ell$. We use $S[i..j] = S[i] \ldots S[j]$, with $1\le i,j\le\ell$, to
denote a substring of $S$, which is the empty string $\varepsilon$ if $i>j$.
A prefix of $S$ is a substring of the form $S[1..i]$ (also written $S[..i]$)
and a suffix is a substring of the form $S[i..\ell]$ (also written $S[i..]$). 
The juxtaposition of strings and/or symbols represents
their concatenation.

We will consider indexing a {\em text} $T[1..n]$, which is a string over
alphabet $[1..\sigma]$ terminated by the special symbol $\$ = 1$, that is, the
lexicographically smallest one, which appears only at $T[n]=\$$. This makes any
lexicographic comparison between suffixes well defined.

Our computation model is the transdichotomous RAM, with a word of 
$w=\Omega(\log n)$ bits, where all the standard arithmetic and logic
operations can be carried out in constant time. In this article we generally
measure space in words.

\subsection{Suffix Trees and Arrays}

The {\em suffix tree} \cite{Wei73} of $T[1..n]$ is a compacted trie where all
the $n$ suffixes of $T$ have been inserted. By compacted we mean that chains
of degree-1 nodes are collapsed into a single edge that is labeled with the
concatenation of the individual symbols labeling the collapsed edges. The
suffix tree has $n$ leaves and less than $n$ internal nodes. By representing
edge labels with pointers to $T$, the suffix tree uses $O(n)$ space, and can 
be built in $O(n)$ time \cite{Wei73,McC76,Ukk95,FFM00}.

The {\em suffix array} \cite{MM93} of $T[1..n]$ is an array $\SA[1..n]$ storing
a permutation of $[1..n]$ so that, for all $1 \le p < n$, the suffix
$T[\SA[p]..]$ is lexicographically smaller than the suffix
$T[\SA[p+1]..]$. Thus $\SA[p]$ is the starting position in $T$ of the
$p$th smallest suffix of $T$ in lexicographic order. This can be regarded as
an array collecting the leaves of the suffix tree. The suffix array uses
$n$ words and can be built in $O(n)$ time without building the suffix tree
\cite{KSPP05,KA05,KSB06}. 

All the occurrences of a pattern string $P[1..m]$ in $T$ can be easily spotted 
in the suffix tree or array. In the suffix tree, we descend from the root
matching the successive symbols of $P$ with the strings labeling the edges.
If $P$ is in $T$, the symbols of $P$ will be exhausted at a node $v$ or inside
an edge leading to a node $v$; this node is called the {\em locus} of $P$,
and all the $occ$ leaves descending from $v$ are the suffixes starting with $P$,
that is, the starting positions of the occurrences of $P$ in $T$. By using
perfect hashing to store the first characters of the edge labels descending
from each node of $v$, we reach the locus in optimal time $O(m)$ and the space 
is still $O(n)$. If $P$ comes packed using $w/\log\sigma$ symbols per computer 
word, we can descend in time $O(\lceil m \log(\sigma)/w\rceil)$ \cite{NN17},
which is optimal in the packed model. In the suffix array, all the suffixes 
starting with $P$ form a
range $\SA[sp..ep]$, which can be binary searched in time $O(m\log n)$, or
$O(m+\log n)$ with additional structures \cite{MM93}.

The inverse permutation of $\SA$, $\ISA[1..n]$, is called the {\em inverse 
suffix array}, so that $\ISA[i]$ is the lexicographical position of the 
suffix $T[i..]$ among all the suffixes of $T$.

Another important concept related to suffix arrays and trees is the longest
common prefix array. Let $lcp(S,S')$ be the length of the longest common prefix
between two strings $S \not=S'$, that is, $S[1..lcp(S,S')]=S'[1..lcp(S,S')]$ but
$S[lcp(S,S')+1]\not=S'[lcp(S,S')+1]$. Then we define the {\em longest
common prefix array} $\LCP[1..n]$ as $\LCP[1]=0$ and $\LCP[p] =
lcp(T[\SA[p-1]..],T[\SA[p]..])$. The $\LCP$ array uses $n$ words and can be
built in $O(n)$ time \cite{KLAAP01}.

\subsection{Self-indexes}

A {\em self-index} is a data structure built on $T[1..n]$ that provides at 
least the following functionality:
\begin{description}
	\item[Count:] Given a pattern $P[1..m]$, compute the number $occ$ of 
occurrences of $P$ in $T$.
	\item[Locate:] Given a pattern $P[1..m]$, return the $occ$ positions where $P$
occurs in $T$.
	\item[Extract:] Given a range $[i..i+\ell-1]$, return $T[i..i+\ell-1]$.
\end{description}

The last operation allows a self-index to act as a replacement of $T$, that is,
it is not necessary to store $T$ since any desired substring can be extracted
from the self-index. This can be trivially obtained by including a copy of $T$
as a part of the self-index, but it is challenging when the self-index 
must use little space.

In principle, suffix trees and arrays can be regarded as self-indexes that can
count in time $O(m)$ or $O(\lceil m\log(\sigma)/w\rceil)$ (suffix tree, by storing $occ$ in
each node $v$) and $O(m\log n)$ or $O(m+\log n)$ (suffix array, with $occ =
ep-sp+1$), locate each occurrence in $O(1)$ time, and extract
in time $O(\lceil\ell\log(\sigma)/w\rceil)$. However, they use $O(n\log n)$ bits, much more
than the $n\log\sigma$ bits needed to represent $T$ in plain form.
We are interested in {\em compressed self-indexes} \cite{NM07,Nav16}, which use
the space required by a compressed representation 
of $T$ (under some compressibility measure) plus some redundancy (at worst
$o(n\log\sigma)$ bits). We describe later the FM-index, a particular self-index
of interest to us.

\subsection{Burrows-Wheeler Transform} \label{sec:bwt}

The {\em Burrows-Wheeler Transform} of $T[1..n]$, $\BWT[1..n]$ \cite{BW94}, is a
string defined as $\BWT[p] = T[\SA[p]-1]$ if $\SA[p]>1$, and $\BWT[p]=T[n]=\$$ 
if $\SA[p]=1$. That is, $\BWT$ has the same symbols of $T$ in a different order,
and is a reversible transform.

The array $\BWT$ is obtained from $T$ by first building $\SA$, although it can
be built directly, in $O(n)$ time and within $O(n\log\sigma)$ bits of space
\cite{MNN17}. To obtain $T$ from $\BWT$ \cite{BW94}, one considers two arrays,
$L[1..n] = \BWT$ and $F[1..n]$, which contains all the symbols of $L$ (or $T$) 
in ascending order. Alternatively, $F[p]=T[\SA[p]]$, so $F[p]$ follows $L[p]$ 
in $T$. We need a function that maps any $L[p]$ to the position $q$ of that 
same character in $F$. The formula is $\LF(p) = C[c] + \rank[p]$, where
$c=L[p]$, $C[c]$ is the number of occurrences of symbols less than $c$ in $L$,
and $\rank[p]$ is the number of occurrences of symbol $L[p]$ in $L[1..p]$.
A simple $O(n)$-time pass on $L$ suffices to compute arrays $C$ and 
$\rank$ using $O(n\log\sigma)$ bits of space. Once they are computed, we
reconstruct $T[n]=\$$ and $T[n-k] \leftarrow L[\LF^{k-1}(1)]$ for $k=1,\ldots,
n-1$, in $O(n)$ time as well. Note that $\LF$ is a permutation formed by a
single cycle.

\subsection{Compressed Suffix Arrays and FM-indexes} \label{sec:fmindex}

Compressed suffix arrays \cite{NM07} are a particular case of self-indexes 
that simulate
$\SA$ in compressed form. Therefore, they aim to obtain the suffix array
range $[sp..ep]$ of $P$, which is sufficient to count since $P$ then appears
$occ = ep-sp+1$ times in $T$. For locating, they need to access the content of
cells $\SA[sp],\ldots,\SA[ep]$, without having $\SA$ stored.

The FM-index \cite{FM05,FMMN07} is a compressed suffix array that exploits the
relation between the string $L=\BWT$ and the suffix array $\SA$. It 
stores $L$ in compressed form (as it can be easily compressed to the 
high-order empirical entropy of $T$ \cite{Man01}) and adds sublinear-size data 
structures to compute (i) any desired position $L[p]$, (ii) the generalized
{\em rank function} $\rank_c(L,p)$, which is the number of times symbol $c$
appears in $L[1..p]$. Note that these two operations permit, in particular, 
computing $\rank[p] = \rank_{L[p]}(L,p)$, which is called {\em partial rank}.
Therefore, they compute
$$ \LF(p) ~~=~~ C[L[p]] + \rank_{L[p]}(L,p).$$

For counting, the FM-index resorts to {\em backward search}. This procedure
reads $P$ backwards and at any step knows the range $[sp_j,ep_j]$ of $P[j..m]$
in $T$. Initially, we have the range $[sp_{m+1}..ep_{m+1}]=[1..n]$ for
$P[m+1..m]=\varepsilon$. Given the range $[sp_{j+1}..ep_{j+1}]$, one obtains
the range $[sp_j..ep_j]$ from $c=P[j]$ with the operations 
\begin{eqnarray*}
sp_j &=& C[c]+\rank_c(L,sp_{j+1}-1)+1,\\
ep_j &=& C[c]+\rank_c(L,ep_{j+1}).
\end{eqnarray*}
Thus the range $[sp..ep]=[sp_1..ep_1]$ is obtained with $O(m)$ computations
of $\rank$, which dominates the counting complexity.

For locating, the FM-index (and most compressed suffix arrays) stores sampled
values of $\SA$ at regularly spaced text positions, say multiples of $s$. 
Thus, to retrieve $\SA[p]$, we find the smallest $k$ for which 
$\SA[\LF^k(p)]$ is sampled, and then the answer is $\SA[p] = \SA[\LF^k(p)]+k$. 
This is because
function $\LF$ virtually traverses the text backwards, that is, it drives us
from $L[p]$, which precedes suffix $\SA[p]$, to its position $F[q]$, where the
suffix $\SA[q]$ starts with $L[p]$, that is, $\SA[q]=\SA[p]-1$: 
$$\SA[\LF(p)] ~=~ \SA[p]-1.$$
Since it is guaranteed that $k < s$, each occurrence
is located with $s$ accesses to $L$ and computations of $\LF$, and the extra 
space for the sampling is $O((n\log n)/s)$ bits, or $O(n/s)$ words.

For extracting, a similar sampling is used on $\ISA$, that is, we sample the
positions of $\ISA$ that are multiples of $s$. To extract $T[i..i+\ell-1]$ we
find the smallest multiple of $s$ in $[i+\ell..n]$, $j=s \cdot \lceil
(i+\ell)/s \rceil$, and extract $T[i..j]$. Since $\ISA[j]=p$ is sampled, we
know that $T[j-1] = L[p]$, $T[j-2] = L[\LF(p)]$, and so on. In total we
require at most $\ell+s$ accesses to $L$ and computations of $\LF$ to extract
$T[i..i+\ell-1]$. The extra space is also $O(n/s)$ words.

For example, using a representation \cite{BN14} that accesses $L$ and computes 
partial ranks in constant time (so $\LF$ is computed in $O(1)$ time), and
computes $\rank$ in the optimal $O(\log\log_w \sigma)$ time, an FM-index can
count in time $O(m\log\log_w \sigma)$, locate each occurrence in $O(s)$ time,
and extract $\ell$ symbols of $T$ in time $O(s+\ell)$, by using $O(n/s)$ space
on top of the empirical entropy of $T$ \cite{BN14}. There exist even faster
variants \cite{BN13}, but they do not rely on backward search.

\subsection{Run-Length FM-index} \label{sec:rlfmi}

One of the sources of the compressibility of $\BWT$ is that symbols are
clustered into $r \le n$ {\em runs}, which are maximal substrings formed by the
same symbol. M\"akinen and Navarro \cite{MN05} proved a (relatively weak)
bound on $r$ in terms of the high-order empirical entropy of $T$ and, more
importantly, designed an FM-index variant that uses $O(r)$ words of space,
called {\em Run-Length FM-index} or {\em RLFM-index}. They later experimented
with several variants of the RLFM-index, where the one called RLFM+ 
\cite[Thm.~17]{MNSV09} corresponds to the original RLFM-index \cite{MN05}.

The structure stores the {\em run heads}, that is, the first positions of the
runs in $\BWT$, in a data structure $E = \{ 1 \} \cup \{ 1<p\le n, 
\BWT[p] \not= \BWT[p-1]\}$ that supports predecessor searches. Each element 
$e \in E$ has associated the value $e.v = |\{ e' \in E, e' \le e\}|$, which is
its position in a string $L'[1..r]$ that stores the run symbols.
Another array, $D[0..r]$, stores the cumulative lengths of the runs after stably
sorting them lexicographically by their symbols (with $D[0]=0$). Let array
$C'[1..\sigma]$ count the number of {\em runs} of symbols smaller than $c$ in
$L$. One can then simulate
$$\rank_c(L,p) ~=~ D[C'[c]+\rank_c(L',q.v-1)] 
+[\textrm{if}~L'[q.v] = c ~\textrm{then}~ p-q+1~\textrm{else}~0],
$$
where $q=pred(E,p)$, at the cost of a predecessor search ($pred$) in $E$ and 
a $\rank$ on $L'$. By using up-to-date data structures, the counting 
performance of the RLFM-index can be stated as follows.

\begin{lemma} \label{lem:rlfm}
The Run-Length FM-index of a text $T[1..n]$ whose $\BWT$ has $r$ runs
can occupy $O(r)$ words and count the number of occurrences of a pattern 
$P[1..m]$ in time $O(m\log\log_w(\sigma + n/r))$. It also computes $\LF$ and
access to any $\BWT[p]$ in time $O(\log\log_w(n/r))$.
\end{lemma}
\begin{proof}
We use the RLFM+ \cite[Thm.~17]{MNSV09}, using the structure of Belazzougui 
and Navarro \cite[Thm.~10]{BN14} for the sequence $L'$ (with constant access
time) and the predecessor data structure described by Belazzougui and Navarro 
\cite[Thm.~14]{BN14} to implement $E$ (instead of the bitvector used in the
original RLFM+).
The RLFM+ also implements $D$ with a bitvector, but we use a plain array.
The sum of both operation times is $O(\log\log_w \sigma + \log\log_w(n/r))$, 
which can be written as $O(\log\log_w (\sigma+n/r))$. To access $\BWT[p]=L[p]
=L'[pred(E,p).v]$ we only need a predecessor search on $E$, which takes time 
$O(\log\log_w(n/r))$, and a constant-time access to $L'$.
Finally, we compute $\LF$ faster than a general rank query, as we only 
need the partial rank query 
$$\rank_{L[p]}(L,p) ~=~ D[C'[L'[q.v]]+\rank_{L'[q.v]}(L',q.v)-1] + (p-q+1),$$
which is correct since $L[p]=L'[q.v]$. The operation $\rank_{L'[q.v]}(L',q.v)$
can be supported in constant time using $O(r)$ space, by just recording all 
the answers, and therefore the time for $\LF$ on $L$ is also dominated by the 
predecessor search on $E$ (to compute $q$), of $O(\log\log_w(n/r))$ time.
\qed
\end{proof}

We will generally assume that $\sigma$ is the {\em effective} alphabet of $T$,
that is, the $\sigma$ symbols appear in $T$. This implies that $\sigma \le r
\le n$. If this is not the case, we can map $T$ to an effective alphabet
$[1..\sigma']$ before indexing it. A mapping of $\sigma' \le r$ words then
stores the actual symbols when extracting a substring of $T$ is necessary. For
searches, we have to map the $m$ positions of $P$ to the effective alphabet.
By storing a perfect hash or a deterministic dictionary \cite{Ruz08} of 
$O(\sigma')=O(r)$ words, we map each symbol of $P$ in constant time. On the
other hand, the results on packed symbols only make sense if $\sigma$ is
small, and thus no alphabet mapping is necessary. Overall,
we can safely use the assumption $\sigma \le r$ without affecting any of our
results, including construction time and space.

To provide locating and extracting functionality, M\"akinen et al.~\cite{MNSV09}
use the sampling mechanism we described for the FM-index. Therefore, although
they can efficiently count within $O(r)$ space, they need a much larger
$O(n/s)$ space to support these operations in time proportional to $s$.
Despite various efforts \cite{MNSV09}, this has been a bottleneck in theory and
in practice since then.

\subsection{Compressed Suffix Trees}

Suffix trees provide a much more complete functionality than self-indexes,
and are used to solve complex problems especially in bioinformatic applications
\cite{Gus97,Ohl13,MBCT15}. A {\em compressed suffix tree} is regarded as an
enhancement of a compressed suffix array (which, in a sense, represents only
the leaves of the suffix tree). Such a compressed representation must be able
to simulate the operations on the classical suffix tree (see 
Table~\ref{tab:streeops} later in the article), while using little
space on top of the compressed suffix array. The first such compressed suffix
tree \cite{Sad07} used $O(n)$ extra bits, and there are several variants using 
$o(n)$ extra bits \cite{FMN09,Fis10,RNO11,GO13,ACN13}.

Instead, there are no compressed suffix trees using $O(r\,\mathrm{polylog}(n))$ 
space. An extension
of the RLFM-index \cite{MNSV09} still needs $O(n/s)$ space to carry out
most of the suffix tree operations in time $O(s\log n)$. Some variants that
are designed for repetitive text collections \cite{ACN13,NO16} are heuristic
and do not offer worst-case guarantees. Only recently a compressed suffix tree
was presented \cite{BC17} that uses $O(\overline{e})$ space and carries out
operations in $O(\log n)$ time.


\section{Locating Occurrences} \label{sec:locate}

In this section we show that, if the $\BWT$ of a text $T[1..n]$ has $r$ runs, 
then we can have an index using $O(r)$ space that not only efficiently finds the
interval $\SA[sp..ep]$ of the occurrences of a pattern $P[1..m]$ (as was
already known in the literature, see Section~\ref{sec:rlfmi}) but that can 
locate
each such occurrence in time $O(\log\log_w(n/r))$ on a RAM machine of $w$ bits.
Further, the time per occurrence becomes constant if the space is raised to
$O(r\log\log_w(n/r))$.

We start with Lemma~\ref{lem:find_one}, which shows that the typical backward
search process can be enhanced so that we always know the position of one of
the values in $\SA[sp..ep]$. We give a simplification of the previous proof
\cite{Pre16,PP16}. Lemma~\ref{lem:find_neighbours} then shows how
to efficiently obtain the two neighboring cells of $\SA$ if we know the value
of one. This allows us to extend the first known cell in both directions, until
obtaining the whole interval $\SA[sp..ep]$. 
Theorem~\ref{thm:locating} summarizes the main result of this section.

Later, Lemma~\ref{lemma: general locate} shows how this process can be 
accelerated by using more space. We extend the idea in Lemma~\ref{lemma: lcp}, 
obtaining $\LCP$ values in the same
way we obtain $\SA$ values. While not of immediate use for locating,
this result is useful later in the article and also has independent interest.

\begin{definition}\label{def: sampled position}
        We say that a text character $T[i]$ is \emph{sampled} if and only if 
$T[i]$ is the first or last character in its $\BWT$ run. That is, 
$T[\SA[1]-1] = T[n-1]$ is sampled and, if $p>1$ and $\BWT[p] \not= \BWT[p-1]$, 
then $T[\SA[p-1]-1]$ and $T[\SA[p]-1]$ are sampled. In general,
$T[i]$ is {\em $s$-sampled} if it is at distance at most $s$ from a $\BWT$ run
border, where sampled characters are at distance 1.
\end{definition}

\begin{lemma}
	\label{lem:find_one}
	We can store $O (r)$ words such that, given $P [1..m]$, in time
$O (m \log \log_w (\sigma + n/r))$ we can compute the interval $\SA[sp,ep]$ 
of the occurrences of $P$ in $T$, and also return the position $p$ and content
$\SA[p]$ of at least one cell in the interval $[sp,ep]$.
\end{lemma}

\begin{proof}
We store a RLFM-index and predecessor structures $R_c$ storing the position in 
$\BWT$ of all the sampled characters equal to $c$, for each $c \in [1..\sigma]$. Each element $p \in R_c$ is associated with its corresponding text position, that is, we store pairs $\langle p,\SA[p]-1 \rangle$ sorted by their first component. These structures take a total of $O(r)$ words. 

The interval of characters immediately preceding occurrences of the empty 
string is the entire $\BWT[1..n]$, which clearly includes $P[m]$ as the last 
character in some run (unless $P$ does not occur in $T$). It follows that we 
find an occurrence of $P[m]$ in predecessor time by querying $pred(R_{P[m]},n)$.

Assume we have found the interval $\BWT[sp,ep]$ containing the 
characters immediately preceding all the occurrences of some (possibly empty) 
suffix $P[j + 1..m]$ of $P$, and we know the position and content of some 
cell $\SA[p]$ in the corresponding interval, $sp \le p \le ep$. Since 
$\SA [\LF (p)] = \SA [p] - 1$, if $\BWT [p] = P [j]$ then, after the next 
application of $\LF$-mapping, we still know the position and value of some cell 
$\SA[p']$ corresponding to the interval $\BWT[sp',ep']$ for $P [j..m]$, namely 
$p' = \LF(p)$ and $\SA[p'] = \SA[p]-1$.

On the other hand, if $\BWT [p] \neq P [j]$ but $P$ still occurs somewhere in 
$T$ (i.e., $sp' \le ep'$), then there is at least one $P[j]$ and one non-$P[j]$
in $\BWT[sp,ep]$, and therefore the interval intersects an extreme of a run 
of copies of $P[j]$, thus holding a sampled character. Then, a predecessor query $pred(R_{P[j]},ep)$ gives us
the desired pair $\langle p', \SA[p']-1 \rangle$ with $sp \leq p' \leq ep$ and
$\BWT[p']=P[j]$.

Therefore, by induction, when we have computed the $\BWT$ interval for $P$, we
know the position and content of at least one cell in the corresponding
interval in $\SA$. 

To obtain the desired time bounds, we concatenate all the universes of the $R_c$
structures into a single one of size $\sigma n$, and use a single structure 
$R$ on that universe: each $\langle p,\SA[p-1] \rangle \in R_c$ becomes $\langle (c-1)n+p, \SA[p]-1\rangle$ in $R$, and a search 
$pred(R_c,q)$ becomes $pred(R,(c-1)n+q)-(c-1)n$. Since $R$ contains $2r$ elements on
a universe of size $\sigma n$, we can have predecessor searches in time
$O(\log\log_w (n\sigma/r))$ and $O(r)$ space \cite[Thm.~14]{BN14}. This is the
same $O(\log\log_w(\sigma + n/r))$ time we obtained in Lemma~\ref{lem:rlfm} to
carry out the normal backward search operations on the RLFM-index.
\qed
\end{proof}

Lemma~\ref{lem:find_one} gives us a toehold in the suffix array, and we show
in this section that a toehold is all we need.  We first show that, 
given the position and contents of one cell of the suffix array $\SA$ of a
text $T$, we can compute the contents of the neighbouring cells in $O (\log
\log_w(n/r))$ time.  It follows that, once we have counted the occurrences of a
pattern in $T$, we can locate all the occurrences in $O (\log \log_w(n/r))$ time each.

\begin{definition}(\cite{KMP09}) \label{def:phi}
Let permutation $\phi$ be defined as $\phi(i)=\SA[\ISA[i]-1]$ if
$\ISA[i]>1$ and $\phi(i)=\SA[n]$ otherwise.
\end{definition}

That is, given a text position $i=\SA[p]$ pointed from suffix array position
$p$, $\phi(i)=\SA[\ISA[\SA[p]]-1]= \SA[p-1]$ gives the value of the preceding
suffix array cell. Similarly, $\phi^{-1}(i)=\SA[p+1]$.

\begin{definition} \label{def:phrase}
We parse $T$ into {\em phrases} such that $T[i]$ is the first character in
a phrase if and only if $T[i]$ is sampled.
\end{definition}

\begin{lemma}
	\label{lem:find_neighbours}
	We can store $O (r)$ words such that functions $\phi$ and $\phi^{-1}$
are evaluated in $O (\log \log_w (n/r))$ time.
\end{lemma}

\begin{proof}
We store an $O(r)$-space predecessor data structure $P^\pm$ with 
$O(\log \log_w (n/r))$ query time \cite[Thm.~14]{BN14} for the starting phrase
positions $i$ of $T$ (i.e., the sampled text positions).
We also store, associated with such values $i \in P^\pm$,
the positions in $T$ next to the characters immediately preceding and following 
the corresponding position $\BWT[q]$, that is, 
$N[i] = \langle \SA[q-1],\SA[q+1] \rangle$ for $i=\SA[q]-1$.
	
	Suppose we know $\SA [p] = k + 1$ and want to know $\SA [p - 1]$ and $\SA [p + 1]$.  This is equivalent to knowing the position $\BWT[p]=T[k]$ and wanting to know the positions in $T$ of $\BWT [p - 1]$ and $\BWT [p + 1]$.  To compute these positions, we find in $P^\pm$ the position $i$ in $T$ of the first character of the phrase containing $T [k]$, take the associated positions $N[i]=\langle x,y\rangle$, and return 
$\SA[p-1] = x+k-i$ and $\SA[p+1] =y+k-i$.

	To see why this works, let $\SA[p-1] = j+1$ and $\SA[p+1]=l+1$,
that is, $j$ and $l$ are the positions in $T$ of $\BWT[p-1]=T[j]$ and 
$\BWT[p+1]=T[l]$. Note that, for all $0 \le t < k-i$, $T[k-t]$ is not the 
first nor the last character of a run in $\BWT$. Thus, by definition of 
$\LF$, $\LF^t(p-1)$, $\LF^t(p)$, and $\LF^t(p+1)$, that is, the $\BWT$
positions of $T[j-t]$, $T[k-t]$, and $T[l-t]$, are
contiguous and within a single run, thus $T[j - t] = T[k - t] = T[l - t]$.
Therefore, for $t=k-i-1$, $T[j-(k-i-1)] = T[i+1] = T[l-(k-i+1)]$ are 
contiguous in $\BWT$, and thus a further $\LF$ step yields that
$\BWT[q]=T[i]$ is immediately preceded and followed by 
$\BWT[q-1]=T[j-(k-i)]$ and $\BWT[q+1]=T[l-(k-i)]$. 
That is, $N[i]=\langle \SA[q-1],\SA[q+1]\rangle = 
\langle j-(k-i)+1,l-(k-i)+1\rangle$
and our answer is correct.
Figure~\ref{fig:locate} illustrates the proof.
\qed
\end{proof}

\begin{figure}[t]
\begin{center}
\includegraphics[width=7cm]{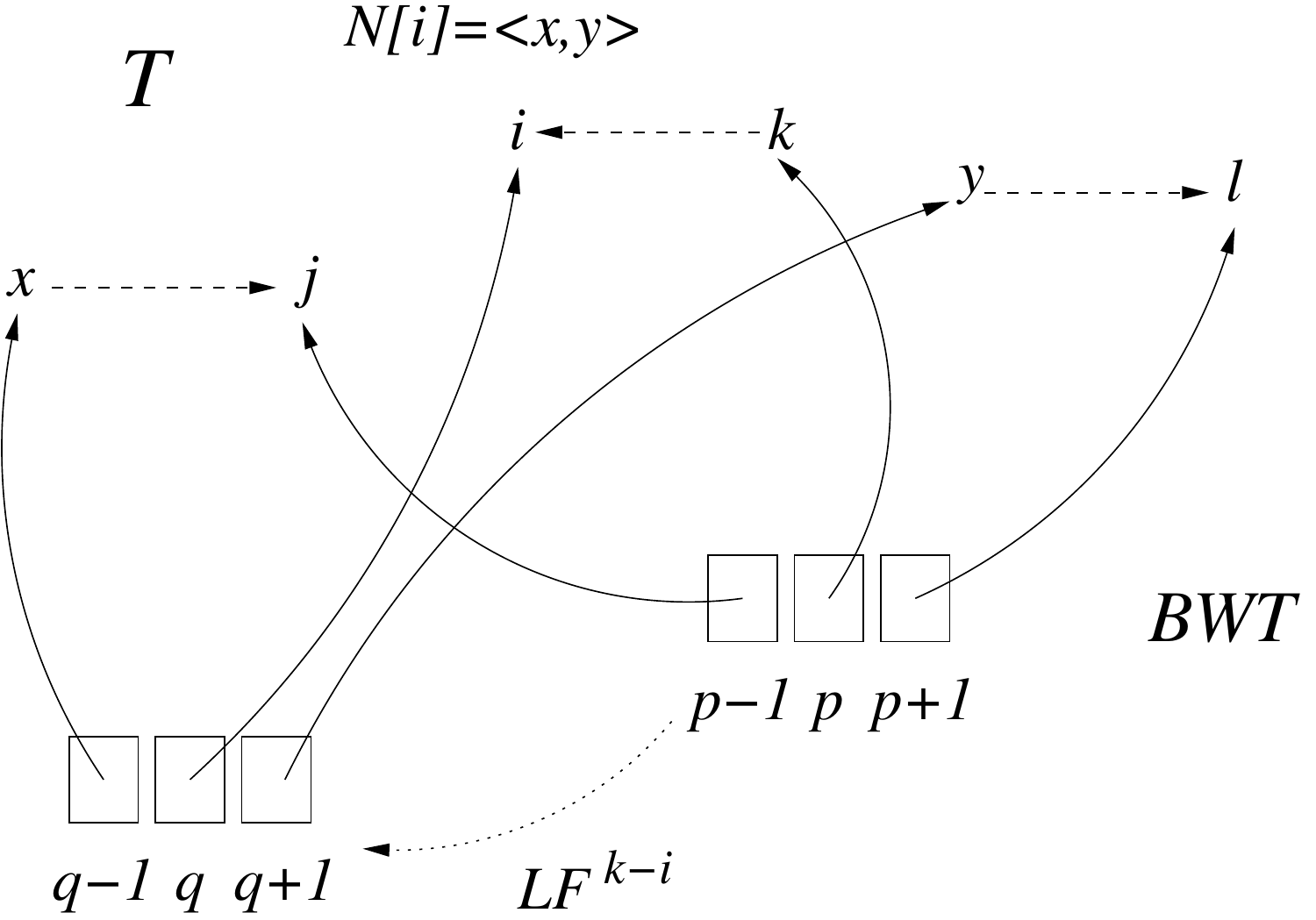}
\end{center}
\caption{Illustration of Lemma~\ref{lem:find_neighbours}. Since $\BWT[p]=T[k]$
and $i$ is the predecessor of $k$, the cells $p-1$, $p$, and $p+1$ would
travel together through consecutive applications of $\LF$, reaching the
positions $N[i]=\langle x,y\rangle$ after $k-i$ steps. Thus it must be that
$\BWT[p-1] = T[x+k-i]$ and $\BWT[p+1] = T[y+k-i]$.}
 \label{fig:locate}
\end{figure}

We then obtain the main result of this section.

\begin{theorem}
	\label{thm:locating}
	We can store a text $T [1..n]$, over alphabet $[1..\sigma]$, in $O(r)$ words, where $r$
	is the number of runs in the $\BWT$ of $T$, such that later, given a pattern
	$P [1..m]$, we can count the occurrences of $P$ in $T$ in $O (m \log \log_w
	(\sigma + n/r))$ time and (after counting) report their $occ$ locations in
	overall time $O (occ \cdot \log \log_w (n/r))$.
\end{theorem}

\subsection{Larger and faster}

The following lemma shows that the above technique  can be generalized. The result is a space-time tradeoff allowing us to list each occurrence in constant time at the expense of a slight increase in space usage. This will be useful later
in the article, in particular to obtain optimal-time locating.

\begin{lemma}\label{lemma: general locate}
	Let $s>0$. We can store a data structure of $O(rs)$ words such that, given
	$\SA[p]$, we can compute $\SA[p - j]$ and $\SA[p + j]$ for $j=1,\dots, s'$ and any $s'\leq s$, in
	$O( \log\log_w(n/r) + s')$ time.
\end{lemma}
\begin{proof}
	Consider all $\BWT$ positions $q_1 < \dots < q_t$ of $s$-sampled characters, and let $W[1..t]$ be an array such that $W[k]$ is the text position
	corresponding to $q_k$, for $k=1,\dots, t$. Now let $q^+_1 < \dots <
	q^+_{t^+}$ be the $\BWT$ positions having a run border at most $s$ positions
	after them, and $q^-_1 < \dots < q^-_{t^-}$ be the $\BWT$ positions having a
	run border at most $s$ positions before them. We store the text positions
	corresponding to $q^+_1 < \dots < q^+_{t^+}$ and $q^-_1 < \dots < q^-_{t^-}$
	in two predecessor structures $P^+$ and $P^-$, respectively, of size $O(rs)$.
	We store, for each $i \in P^+ \cup P^-$, its position $f(i)$ in $W$, that
	is, $W[f(i)]=i$. 
	
	To answer queries given $\SA[p]$, we first compute its $P^+$-predecessor
	$i<\SA[p]$ in $O( \log\log_w(n/r))$ time, and retrieve $f(i)$. Then, it holds
	that $\SA[p + j] = W[f(i)+j] + (\SA[p]-i)$, for $j=0,\dots, s$. Computing $\SA[p - j]$ is symmetric (just use $P^-$ instead of $P^+$).
	
	To see why this procedure is correct, consider the range $\SA[p..p+s]$. 
	We distinguish two cases.
	
	(i) $\BWT[p..p+s]$ contains at least two distinct characters. Then,
	$\SA[p]-1 \in P^+$ (because $p$ is followed by a run break at most $s$
	positions away), and is therefore the immediate predecessor of $\SA[p]$. 
	Moreover, all $\BWT$ positions $[p..p+s]$ are in $q_1, \dots, q_t$ (since
	they are at distance at most $s$ from a run break), and their corresponding text positions
	are therefore contained in a contiguous range of $W$ (i.e., $W[f(\SA[p]-1)..
	f(\SA[p]-1)+s]$). The claim follows.
	
	(ii) $\BWT[p..p+s]$ contains a single character; we say it is unary. 
	Then $\SA[p]-1 \notin P^+$, since there are no run breaks in $\BWT[p..p+s]$. 
	Moreover, by the $\LF$ formula, the $\LF$ mapping applied on the unary range
	$\BWT[p..p+s]$ gives a contiguous range $\BWT[\LF(p)..\LF(p+s)] = \BWT[\LF(p)..\LF(p)+s]$. Note that 
	this corresponds to a parallel backward step on text positions 
	$\SA[p] \rightarrow \SA[p] -1 , \dots,
	\SA[p+s] \rightarrow \SA[p+s]-1$. We iterate the application of $\LF$ until we
	end up in a range $\BWT[\LF^\delta(p)..\LF^\delta(p+s)]$ that is not unary.
	Then, $\SA[\LF^\delta(p)]-1$ is the immediate predecessor of
	$\SA[p]$ in $P^+$, and $\delta+1$ is their distance. This means that with a single predecessor query on $P^+$ we
	``skip'' all the unary $\BWT$ ranges $\BWT[\LF^k(p)..\LF^k(p+s)]$ for
	$k=1,\dots,\delta-1$ and, as in case (i), retrieve the contiguous range in $W$
	containing the values $\SA[p]-\delta, \dots, \SA[p+s]-\delta$, and add $\delta$ to obtain the desired $\SA$ values.
	\qed
\end{proof}	

\subsection{Accessing $\LCP$} \label{sec:lcp}

Lemma~\ref{lemma: general locate} can be further extended to entries of the 
$\LCP$ array, which we will use later in the article. Given $\SA[p]$,
we compute $\LCP[p]$ and its adjacent entries (note that we do not need to 
know $p$, but just $\SA[p]$). For $s=1$ this is known as the {\em permuted
$\LCP$ (PLCP)} array \cite{Sad07}. Our result can indeed be seen as an 
extension of a {\em PLCP} representation by Fischer et al.~\cite{FMN09}.
In Section~\ref{sec:dlcp} we use different structures that enable the
classical access, that is, compute $\LCP[p]$ from $p$, not $\SA[p]$.

\begin{lemma}\label{lemma: lcp}
	Let $s>0$. We can store a data structure of $O(rs)$ words such that,
	given  $\SA[p]$, we can compute $\LCP[p -
j+1]$ and $\LCP[p + j]$, for $j=1,\dots, s'$ and any $s' \le s$, in 
$O( \log\log_w(n/r) + s')$ time.
\end{lemma}
\begin{proof}
	The proof follows closely that of Lemma \ref{lemma: general locate},
except that now we sample $\LCP$ entries corresponding to suffixes 
{\em following} $s$-sampled $\BWT$ positions.
	Let us define $q_1 < \dots < q_t$, $q^+_1 < \dots < q^+_{t^+}$, and
	$q^-_1 < \dots < q^-_{t^-}$, as well as the predecessor structures
	$P^+$ and $P^-$, exactly as in the proof of Lemma~\ref{lemma: general
locate}.
	We store $\LCP'[1..t] = \LCP[q_1], \dots, \LCP[q_t]$.
	We also store, for each $i \in P^+ \cup P^-$, its corresponding position $f(i)$ in $\LCP'$, that
	is, $\LCP'[f(i)] = \LCP[\ISA[i+1]]$. 
	
	To answer queries given $\SA[p]$, we first compute its $P^+$-predecessor
	$i<\SA[p]$ in $O( \log\log_w(n/r))$ time, and retrieve $f(i)$. Then, it holds
	that $\LCP[p + j] = \LCP'[f(i)+j] - (\SA[p]-i-1)$, for $j=1,\dots, s$. Computing $\LCP[p - j]$ for $j=0,\dots,s-1$ is symmetric (using $P^-$ instead of $P^+$).
	
	To see why this procedure is correct, consider the range $\SA[p..p+s]$. 
	We distinguish two cases.
	
	(i) $\BWT[p..p+s]$ contains at least two distinct characters. Then,
	as in case (i) of Lemma~\ref{lemma: general locate},
	$\SA[p]-1 \in P^+$ 
	and is therefore the immediate predecessor $i=\SA[p]-1$ of $\SA[p]$. 
	Moreover, all $\BWT$ positions $[p..p+s]$ are in $q_1, \dots, q_t$, and
	therefore values $\LCP[p..p+s]$ are explicitly stored in a contiguous
range in $\LCP'$ (i.e., $\LCP'[f(i)..f(i)+s]$). Note that $\SA[p]-i=1$, so $\LCP'[f(i)+j] - (\SA[p]-i-1) = \LCP'[f(i)+j]$ for $j=0,\dots,s$. The claim follows.
	
	(ii) $\BWT[p..p+s]$ contains a single character, so it is unary. 
	Then we reason exactly as in case (ii) of Lemma~\ref{lemma: general locate} to 
	define $\delta$ so that
	$i'=\SA[\LF^\delta(p)]-1$ is the immediate predecessor of
	$\SA[p]$ in $P^+$ 
	and, as in case (i)
of this proof, 
	retrieve the contiguous range $\LCP'[f(i')..f(i')+s]$
	containing the values $\LCP[\LF^\delta(p)..\LF^\delta(p+s)]$. Since
the skipped $\BWT$ ranges are unary, it is not hard to see that
$\LCP[\LF^\delta(p+j)] = \LCP[p+j] + \delta$ for $j=1, \dots, s$ (note that we
do not include $j=0$ since we cannot exclude that, for some $k<\delta$,
$LF^k(p)$ is the first position in its run). From the equality $\delta = \SA[p] - i' - 1 =
\SA[p]-\SA[\LF^\delta(p)]$ (that is, $\delta$ is the distance between $\SA[p]$ and its predecessor minus one or, equivalently, the number of $\LF$ steps virtually performed), we then compute $\LCP[p+j] = \LCP'[f(i')+j]-\delta$ for $j=1, \dots, s$.  \qed
	
\end{proof}	

As a simplification that does not change our asymptotic bounds (but that we
consider in the implementation), note that it is sufficient to sample only the 
last (or the first) characters of $\BWT$ runs. In this case, our toehold in
Lemma~\ref{lem:find_one} will be the last cell $\SA[ep]$ of our current range 
$\SA[sp..ep]$: if $\BWT[ep]=P[j]$, then the next toehold is $ep'$ and its 
position is $\SA[ep]-1$. Otherwise, there must be a run end (i.e., a sampled 
position) in $\SA[sp..ep]$, which we find with $pred(R_{P[j]},ep)$, and this
stores $\SA[ep']$. As a consequence, we only need to store $N[i]=\SA[q-1]$ in 
Lemma~\ref{lem:find_neighbours} and just $P^-$ in 
Lemmas~\ref{lemma: general locate} and \ref{lemma: lcp}, thus reducing the
space for sampling. This was noted simultaneously by several authors after our 
conference paper \cite{GNP18} and published independently \cite{BGI18}.
For this paper, our definition is better suited as the sampling holds
crucial properties --- see the next section.


\section{Counting and Locating in Optimal Time}\label{sec:optimal}

In this section we show how to obtain optimal counting and locating time in 
the unpacked --- $O(m)$ and $O(m+occ)$ --- and packed --- 
$O(\lceil m\log(\sigma)/w\rceil)$ and $O(\lceil m\log(\sigma)/w\rceil+occ)$ 
--- scenarios, by
using $O(r\log\log_w(\sigma+n/r))$ and $O(r w \log_\sigma\log_w(\sigma+n/r))$ 
space, respectively.
To improve upon the times of Theorem \ref{thm:locating} we process $P$ by
chunks of $s$ symbols on a text $T^*$ formed by chunks, too.

\subsection{An RLFM-index on chunks}

Given an integer $s \ge 1$, let us define texts $T^k[1..\lceil n/s\rceil]$ 
for $k=0,\ldots, s-1$, so that $T^k[i] = T[k+(i-1)s+1..k+is]$, where we assume
$T$ is padded with up to $2(s-1)$ copies of \$, as needed.
That is, $T^k$ is $T$ devoid of its first $k$ symbols and then seen as a
sequence of metasymbols formed by $s$ original symbols. We then define a 
new text $T^* = T^0\,T^1\cdots T^{s-1}$. The text $T^*$ has length $n^* = 
s \cdot \lceil n/s \rceil < n+s$ and its alphabet is of size at most $\sigma^s$.
Assume for now that $\sigma^s$ is negligible; we consider it soon.

We say that a suffix in $T^*$ corresponds to the suffix of $T$ from where it
was extracted.

\begin{definition} \label{def:corrsuffix}
Suffix $T^*[i^*..n^*]$ corresponds to suffix $T[i..n]$ iff the concatenations
of the symbols forming the metasymbols in $T^*[i^*..n^*]$ is equal to the 
suffix $T[i..n]$, if we compare them up to the first occurrence of $\$$.
\end{definition}

The next observation specifies the algebraic transformation between the 
positions in $T^*$ and $T$.

\begin{observation} \label{obs:corrsuffix}
Suffix $T^*[i^*..n^*]$ corresponds to suffix $T[i..n]$ iff
$i = ((i^*-1)\!\!\mod \lceil n/s\rceil ) \cdot s + \lceil i^*/\lceil
n/s\rceil \rceil$.
\end{observation}

The key property we exploit is that corresponding suffixes of $T$ and $T^*$
have the same lexicographic rank.

\begin{lemma} \label{lem:lexsuffix}
For any suffixes $T^*[i^*..n^*]$ and $T^*[j^*..n^*]$ corresponding to
$T[i..n]$ and $T[j..n]$, respectively, it holds that
$T^*[i^*..n^*] \le T^*[j^*..n^*]$ iff $T[i..n] \le T[j..n]$.
\end{lemma}
\begin{proof}
Consider any $i^* \not= j^*$, otherwise the result is trivial because $i=j$.
We proceed by induction on $n^*-i^*$. If this is zero, then $T[i^*..n^*]
= T[n^*] = T^{s-1}[\lceil n/s\rceil] = T[s-1+(\lceil n/s\rceil-1)s+1 ..
s-1+\lceil n/s\rceil s] = \$^s$ is always $\le T[j^*..n^*]$ for any $j^*$.
Further, by Observation~\ref{obs:corrsuffix}, $i = \lceil n/s\rceil \cdot s$ is the
rightmost suffix of $T$ (extended with \$s), formed by all \$s, 
and thus it is $\le T[j..n]$ for any $j$.

Now, given a general pair $T^*[i^*..n^*]$ and $T^*[j^*..n^*]$, consider the
first metasymbols $T^*[i^*]$ and $T^*[j^*]$. If they are different, then the
comparison depends on which of them is lexicographically smaller. Similarly,
since $T^*[i^*]=T[i..i+s-1]$ and $T^*[j^*]=T[j..j+s-1]$, the comparison of the 
suffixes $T[i..n]$ and $T[j..n]$ depends on which is smaller between the 
substrings $T[i..i+s-1]$ and $T[j..j+s-1]$. Since the metasymbols
$T^*[i^*]$ and $T^*[j^*]$ are ordered lexicographically, the
outcome of the comparison is the same. If, instead, $T^*[i^*]=T^*[j^*]$,
then also $T[i..i+s-1]=T[j+s-1]$. The comparison in $T^*$ is then decided by
the suffixes $T^*[i^*+1..n^*]$ and $T^*[j^*+1..n^*]$, and in $T$ by the
suffixes $T[i+s..n]$ and $T[j+s..n]$. By Observation~\ref{obs:corrsuffix}, the
suffixes $T^*[i^*+1..n^*]$ and $T^*[j^*+1..n^*]$ almost always correspond to 
$T[i+s..n]$ and $T[j+s..n]$, and then by the inductive hypothesis the result
of the comparisons is the same. The case where $T^*[i^*+1..n^*]$ or
$T^*[j^*+1..n^*]$ do not correspond to $T[i+s..n]$ or $T[j+s..n]$ arises when
$i^*$ or $j^*$ are a multiple of $\lceil n/s \rceil$, but in this case they
correspond to some $T^k[\lceil n/s\rceil]$, which contains at least one \$.
Since $i^* \not= j^*$, the number of \$s must be distinct, and then the
metasymbols cannot be equal.
\qed
\end{proof}

An important consequence of Lemma~\ref{lem:lexsuffix} is that the suffix
arrays $\SA^*$ and $\SA$ of $T^*$ and $T$, respectively, list the corresponding
suffixes in the same order (the positions of the corresponding suffixes in
$T^*$ and $T$ differ, though).
Thus we can find suffix array ranges in $\SA$ via searches on $\SA^*$. More
precisely, we can use the RLFM-index of $T^*$ instead of that of $T$. 
The following result is the key to bound the space usage of our structure.

\begin{lemma} \label{lem:rstar}
If the BWT of $T$ has $r$ runs, then the BWT of $T^*$ has $r^* = O(rs)$ runs.
\end{lemma}
\begin{proof}
Kempa \cite[see before Thm.~3.3]{Kem19} shows that the number of {\em $s$-runs}
in the BWT of $T$, that is, the number of maximal runs of equal substrings of 
length $s$ preceding the suffixes in lexicographic order, is at most $s\cdot r$.
Since $\SA$ and $\SA^*$ list the corresponding suffixes in the same order, the 
number of $s$-runs in $T$ essentially corresponds to the number of runs in 
$T^*$, formed by
the length-$s$ metasymbols preceding the same suffixes. The only exceptions are
the $s$ metasymbols that precede some metasymbol $T^k[1]$ in $T^*$. Other $O(s)$
runs can appear because we have padded $T$ with $O(s)$ copies of \$, and thus
$T$ has $O(s)$ further suffixes. Still, the result is in $O(rs)$.
\qed
\end{proof}

\subsection{Mapping the alphabet} \label{sec:mapping}

The alphabet size of $T^*$ is $\sigma^s$, which can be large.
Depending on $\sigma$ and $s$, we could even be unable to handle the
metasymbols in constant time. Note, however, that the {\em effective} alphabet
of $T^*$ must be $\sigma^* \le r^* = O(rs)$, which will always be in 
$o(n\log^2 n)$ for the moderate values of $s$ we will use. Thus we can always
manage metasymbols in $[1..\sigma^*]$ in constant time. 
We use a compact trie 
of height $s$ to convert the existing substrings of length $s$ of $T$ into 
numbers in $[1..\sigma^*]$, respecting the lexicographic order. The trie uses 
perfect hashing to find the desired child in constant time, and the strings
labeling the edges are represented as pointers to an area where we store all the
distinct substrings of length $s$ in $T$. We now show that this area is of 
length $O(rs)$.

\begin{definition} \label{def:primary}
        We say that a text substring $T[i..j]$ is \emph{primary} iff
it contains at least one sampled character (see
Definition~\ref{def: sampled position}).
\end{definition}

\begin{lemma}\label{lemma:primocc}
        Every text substring $T[i..j]$  has a primary occurrence $T[i'..j'] = T[i..j]$.
\end{lemma}
\begin{proof}
        We prove the lemma by induction on $j-i$. If $j-i=0$, then $T[i..j]$ is a single character, and every character has a sampled occurrence $i'$ in the text.
        Now let $j-i > 0$. By the inductive hypothesis, $T[i+1..j]$ has a primary occurrence $T[i'+1..j']$. If $T[i]=T[i']$, then $T[i'..j']$ is a primary 
occurrence of $T[i..j]$. Assume then that $T[i]\neq T[i']$.
Let $[sp,ep]$ be the $\BWT$ range of $T[i+1..j]$. 
        Then there are two distinct symbols in $\BWT[sp,ep]$ and thus there must be a run of $T[i]$'s ending or beginning in $\BWT[sp,ep]$, say at position $sp \leq q \leq ep$. Thus it holds that $\BWT[q] = T[i]$ and the text position $i''=\SA[q]-1$ is sampled. We then have a primary occurrence
$T[i''..j''] = T[i..j]$.
\qed
\end{proof}

\begin{lemma}\label{lemma:distinct kmers}
There are at most $2rs$ distinct $s$-mers in the text, for any $1\leq s \leq n$.
\end{lemma}
\begin{proof}
From Lemma \ref{lemma:primocc}, every distinct $s$-mer
appearing in the text has a primary occurrence. It follows that, in order to
count the number of distinct $s$-mers, we can restrict our attention to the
regions of size $2s-1$ overlapping the at most $2r$ sampled positions 
(Definition~\ref{def: sampled position}). Each sampled position overlaps with
$s$ $s$-mers, so the claim easily follows. \qed
\end{proof}

The compact trie then has size $O(rs)$, since it has $\sigma^* \le r^* = O(rs)$
leaves and no unary paths, and the area with the distinct strings is
also of size $O(rs)$. The structure maps any metasymbol to the new alphabet
$[1..\sigma^*]$, by storing the corresponding symbol in each leaf. Each
internal trie node $v$ also stores the first and last symbols of 
$[1..\sigma^*]$ stored at leaves descending from it, $v_{\min}$ and $v_{\max}$.

We then build the RLFM-index of $T^*$ on the mapped 
alphabet $[1..\sigma^*]$, and our structures using $O(\sigma^s)$ space 
become bounded by space $O(\sigma^*) = O(r^*)$. 

\subsection{Counting in optimal time}

Let us start with the base FM-index. Recalling Section~\ref{sec:fmindex}, the
FM-index of $T^*$ consists of an array $C^*[1..\sigma^*]$ and a string 
$L^*[1..n^*]$, where $C^*[c]$ tells the number of times metasymbols less than 
$c$ occur in $T^*$, and where $L^*$ is the BWT of $T^*$, with the symbols 
mapped to $[1..\sigma^*]$.

To use this FM-index, 
we process $P$ by metasymbols too. We define two patterns,
 $P^* \cdot L_P$ and $P^* \cdot R_P$, with
 $P^*[1..m^*] = P[1..s] P[s+1..2s] \cdots P[\lfloor m/s-1 \rfloor \cdot s+1 ..
\lfloor m/s \rfloor \cdot s]$, $L_P = P[\lfloor m/s \rfloor \cdot
s+1..m] \cdot \$^{s- (m \!\!\mod s)}$, and $R_P = P[\lfloor m/s \rfloor \cdot 
s+1..m] \cdot @^{s- (m \!\!\mod s)}$, $@$ being the largest symbol in the alphabet.
That is, $P^* \cdot P_L$ and $P^* \cdot P_R$ are $P$ padded with the smallest
and largest alphabet symbols, respectively, and then regarded as a sequence of 
$\lfloor m/s \rfloor+1$ metasymbols.
This definition and Lemma~\ref{lem:lexsuffix} ensure that the suffixes of $T$ 
starting with $P$ correspond to the suffixes of $T^*$ starting with strings
lexicographically between $P^* \cdot P_L$ and $P^* \cdot P_R$.  

We use the trie to map the symbols of $P^*$ to the alphabet $[1..\sigma^*]$.
If a metasymbol of $P^*$ is not found, it means that $P$ 
does not occur in $T$. To map the symbols $L_P$ and $R_P$, 
we descend by the symbols $P[\lfloor m/s\rfloor \cdot s+1..m]$ and, upon 
reaching trie node $v$, we use the precomputed limits $v_{\min}$ and $v_{\max}$.
Overall, we map $P^*$, $L_P$ and $R_P$ in $O(m)$ time.

We can then apply backward search almost as in Section~\ref{sec:fmindex}, but 
with a twist for the last symbols of $P^* \cdot P_L$ and $P^* \cdot P_R$: We 
start with the range $[sp_{m^*},ep_{m^*}] = [C^*[v_{\min}]+1,C^*[v_{\max}]]$, 
and then carry out $m^*-1$ steps, for $j = m^*-1,\ldots,1$, as follows, with 
$c$ being the mapping of $P^*[j]$:
\begin{eqnarray*}
sp_j &=& C^*[c]+\rank_c(L^*,sp_{j+1}-1)+1,\\
ep_j &=& C^*[c]+\rank_c(L^*,ep_{j+1}).
\end{eqnarray*}

The resulting range, $[sp,ep]=[sp_1,ep_1]$, corresponds to the range of $P$ in
$T$, and is obtained with $2(m^*-1) \le 2m/s$ operations $\rank_c(L,i)$.

A RLFM-index (Section~\ref{sec:rlfmi}) on $T^*$ stores, instead of $C^*$ and 
$L^*$, structures $E$, $L'$, $D$, and $C'$, of total size $O(\sigma^* + r^*)
=O(r^*)$. These simulate the
operation $\rank_c(L,i)$ in the time of a predecessor search on $E$ and 
$\rank$ and access operations on $L'$. These add up to
$O(\log\log_w(\sigma^* + n/r^*))$ time. We can still retain $C^*$ to carry out 
the first step of our twisted backward search on $L_P$ and $R_P$, and then 
switch to the RLFM-index.

\begin{lemma}
Let $T[1..n]$, on alphabet $[1..\sigma]$, have a BWT with $r$ runs, and let
$s=O(\log n)$ be a positive integer. Then there exists a data structure using 
$O(rs)$ space that counts the number of occurrences of any pattern $P[1..m]$ 
in $T$ in $O(m+(m/s)\log\log_w(\sigma+n/r))$. In particular, a structure
using $O(r \log\log_w(\sigma+n/r))$ space counts in time $O(m)$.
\end{lemma}
\begin{proof}
We build the mapping trie, the RLFM-index on $T^*$ using the mapped alphabet, 
and the array
$C^*$ of the FM-index of $T^*$. All these require $O(\sigma^* + r^*) = 
O(r^*)$ space, which is $O(rs)$ by Lemma~\ref{lem:rstar}. To count the number
of occurrences of $P$, we first compute $P^*$, $L_P$, and $R_P$ on the mapped
alphabet with the trie, in time $O(m)$. We then carry out the backward
search, which requires one constant-time step to find $[sp_{m^*},ep_{m^*}]$
and then $2(m^*-1) \le 2m/s$ steps requiring $\rank_c(L,i)$, which is simulated by
the RLFM-index in time $O(\log\log_w(\sigma^* + n^*/r^*))$. Since $\sigma^* \le
\sigma^s$, $n^* \le n+s$, and $r^* \ge r$, we can write the time as 
$O(\log\log_w(\sigma^s+n/r)) \subseteq O(\log s + \log\log_w(\sigma+n/r))$.
The term $O(\log s)$ vanishes when multiplied by $O(m/s)$ because there is 
an $O(m)$ additive term.
\qed
\end{proof}

\subsection{Locating in optimal time} \label{sec:optloc}

To locate in optimal time, we will use the toehold technique of
Lemma~\ref{lem:find_one} on $T^*$ and $P^*$. The only twist is that, when we
look for $L_P$ and $R_P$ in our trie, we must store in the internal trie node
we reach by $P[\lfloor m/s\rfloor \cdot s+1..m]$ the position $p$ in $\SA^*$
and the value $\SA^*[p]$ of some metasymbol starting with that string. From 
then on, we do exactly as in Lemma~\ref{lem:find_one}, so we can recover the 
interval $\SA^*[sp,ep]$ of $P^*$ in $T^*$. Since, by Observation~\ref{obs:corrsuffix},
we can easily convert position $\SA^*[p]$ to the corresponding position 
$\SA[p]$ in $T$, we have the following result.

\begin{lemma} \label{lem:find1star}
We can store $O(rs)$ words such that, given $P[1..m]$, in time
$O(m+(m/s)\log\log_w(\sigma+n/r))$ we can compute the interval $\SA[sp,ep]$
of the occurrences of $P$ in $T$, and also return the position $p$ and content
$\SA[p]$ of at least one cell in the interval $[sp,ep]$.
\end{lemma}

We now use the structures of Lemma~\ref{lemma: general locate} on the original
text $T$ and with the same value of $s$. Thus, once we obtain some value
$\SA[p]$ within the interval, we return the occurrences in $\SA[sp..ep]$
by chunks of $s$ symbols, in time $O(s+\log\log_w(n/r))$. We then have the
following result.

\begin{theorem} \label{thm:optimal}
Let $s>0$. We can store a text $T[1..n]$, over alphabet $[1..\sigma]$, in
$O(rs)$ words, where $r$ is the number of runs in the BWT of $T$, such that
later, given a pattern $P[1..m]$, we can count the occurrences of $P$ in $T$
in $O(m+(m/s)\log\log_w(\sigma+n/r))$ time and (after counting) report their
$occ$ locations in overall time $O((1+\log\log_w(n/r)/s)\cdot occ)$. In
particular, if $s=\log\log_w(\sigma+n/r)$, the structure uses 
$O(r\log\log_w(\sigma+n/r))$ space, counts in time $O(m)$, and locates in
time $O(m+occ)$.
\end{theorem}

\subsection{RAM-optimal counting and locating}

In order to obtain RAM-optimal time, that is, replacing $m$ by 
$\lceil m \log(\sigma)/w \rceil$ in the counting and locating times, we
can simply use Theorem~\ref{thm:optimal} with $s = (w/\log\sigma) \cdot
\log\log_w(\sigma+n/r) = w \log_\sigma \log_w (\sigma+n/r)$. There is,
however, a remaining $O(m)$ time coming from traversing the trie in order
to obtain the mapped alphabet symbols of $P^*$, $P_L$, and $P_R$.

We then replace our trie by a more sophisticated structure, which is described
by Navarro and Nekrich \cite[Sec.~2]{NN17}, built on the $O(rs)$ distinct
strings of length $s$. Let $d=\lfloor w/\log\sigma\rfloor$. The structure is
like our compact trie but it also stores,
at selected nodes, perfect hash tables that allow 
descending by $d$ symbols in $O(1)$ time. This is sufficient to find the locus 
of a string of length $s$ in $O(\lceil s/d\rceil) = 
O(\lceil s\log(\sigma)/w\rceil)$ time, except for the last $s \!\!\mod d$ 
symbols. For those, the structure also stores weak prefix search (wps) 
structures \cite{BBPV18} on the selected nodes, which allow descending by up 
to $d-1$ symbols in constant time.

The wps structures, however, may fail if the string has no locus, so we must
include a verification step. Such verification is done in RAM-optimal time by
storing the strings of length $2s-1$ extracted around sampled text positions
in packed form, in our memory area associated with the edges. The space of
the data structure is $O(1)$ words per compact trie node, so in our case it is 
$O(rs)$. We then map $P^*$, $P_L$, and $P_R$, in time 
$O(\lceil m\log(\sigma)/w\rceil)$.

\begin{theorem} \label{thm:optimal packed}
We can store a text $T[1..n]$, over alphabet $[1..\sigma]$, in
$O(r w \log_\sigma \log_w(\sigma+n/r))$ words, where $r$ is the number of runs 
in the BWT of $T$, such that later, given a pattern $P[1..m]$, we can count the
occurrences of $P$ in $T$ in $O(\lceil m\log(\sigma)/w \rceil)$ time and
(after counting) report their $occ$ locations in overall time $O(occ)$. 
\end{theorem}


\section{Accessing the Text, the Suffix Array, and Related Structures} \label{sec:sa}

In this section we show how we can provide direct access to the text $T$, 
the suffix array $\SA$, its inverse $\ISA$, and the longest common prefix 
array $\LCP$. The latter
enable functionalities that go beyond the basic counting, locating, and
extracting that are required for self-indexes, and will be used to 
enable a full-fledged compressed suffix tree in Section~\ref{sec:stree}.


We introduce a representation of $T$ that uses $O(r\log(n/r))$ space 
and can retrieve any substring of length $\ell$ in time 
$O(\log(n/r)+\ell\log(\sigma)/w)$. The second term is optimal in the packed
setting and, as explained in the Introduction, the $O(\log(n/r))$ additive 
penalty is also near-optimal in general. 

For the other arrays, we exploit the fact that the runs that appear in the 
$\BWT$ of $T$  induce equal substrings
in the {\em differential} suffix array, its inverse, and longest common
prefix arrays, $\DSA$, $\DISA$, and $\DLCP$, where we store the difference
between each cell and the previous one. Therefore, all the solutions will be
variants of the one that extracts substrings of $T$.
Their extraction time will be $O(\log(n/r)+\ell)$.

\subsection{Accessing $T$} \label{sec:extract}

Our structure to extract substrings of $T$ is a variant of Block Trees
\cite{BGGKOPT15} built around Lemma~\ref{lemma:primocc}.

\begin{theorem}\label{thm:extract}
Let $T[1..n]$ be a text over alphabet $[1..\sigma]$.	
We can store a data structure of $O(r\log(n/r))$ words supporting the
extraction of any length-$\ell$ substring of $T$ in 
$O(\log(n/r)+\ell\log(\sigma)/w)$ time. 
\end{theorem}
\begin{proof}
We describe a data structure supporting the extraction of $\alpha = \frac{w\log(n/r)}{\log\sigma}$ packed characters in $O(\log(n/r))$ time. To extract a text substring of length $\ell$ we divide it into $\lceil\ell/\alpha\rceil$ blocks and extract each block with the proposed data structure. Overall, this will take $O((1+\ell/\alpha)\log(n/r)) = O(\log(n/r) + \ell\log(\sigma)/w)$ time. 

Our data structure is stored in $O(\log(n/r))$ levels. For simplicity, we assume that $r$ divides $n$ and that $n/r$ is a power of two. 
The top level (level 0) is special: we divide the text into $r$ blocks $T[1..n/r],T[n/r+1..2n/r],\dots,T[n-n/r+1..n]$ of size $n/r$. For levels $l>0$, we let $s_l = n/(r\cdot 2^{l-1})$ and, for every sampled position $i$, we consider the two non-overlapping blocks of length $s_l$: $X_{l,i}^1 = T[i-s_l..i-1]$ and  $X_{l,i}^2 = T[i..i+s_l-1]$. 
Each such block $X^k_{l,i}$, for $k=1,2$, is composed of two half-blocks, $X^k_{l,i} = X^k_{l,i}[1..s_l/2]\,X^k_{l,i}[s_l/2+1..s_l]$. 
We moreover consider three additional consecutive and non-overlapping half-blocks, starting in the middle of the first, $X^1_{l,i}[1..s_l/2]$, and ending in the middle of the last, $X^2_{l,i}[s_l/2+1..s_l]$, of the 4 half-blocks just described: $T[i-s_l+s_l/4..i-s_l/4-1],\ T[i-s_l/4..i+s_l/4-1]$, and $T[i+s_l/4..i+s_l-s_l/4-1]$.

From Lemma~\ref{lemma:primocc}, blocks at level $l=0$ and each half-block at level $l> 0$ have a primary occurrence covered by blocks at level $l+1$. Such an occurrence can be fully identified by the coordinate $\langle i', \mathit{off}\rangle$, where $i'$ is a sampled position (actually we
store a pointer $\mathit{ptr}$ to the data structure associated with sampled
position $i'$), and $0< \mathit{off} \leq s_{l+1}$ indicates that the occurrence starts at position $i'-s_{l+1}+ \mathit{off}$ of $T$.

Let $l^*$ be the smallest number such that $s_{l^*} < 4\alpha = 
\frac{4w\log(n/r)}{\log\sigma}$. Then $l^*$ is the last level of our structure.
At this level, we explicitly store a packed string with the characters of the blocks. This uses in total $O(r \cdot s_{l^*}\log(\sigma)/w) = O(r\log(n/r))$ words of space. 

All the blocks at level 0 and half-block at levels $0<l<l^*$ store instead the coordinates $\langle i',\mathit{off}\rangle$ of their primary occurrence in the next level. At level $l^*-1$, these coordinates point inside the strings of explicitly stored characters. These pointers also add up to 
$O(r \cdot l^*) = O(r\log(n/r))$ words of space.

Let $S= T[j..j+\alpha-1]$ be the text substring to be extracted. Note that we can assume $n/r \geq \alpha$; otherwise all the text can be stored in plain packed form using $n\log(\sigma)/w < \alpha r\log(\sigma)/w = O(r\log(n/r))$ words and we do not need any data structure. It follows that $S$ either spans two blocks at level 0, or it is contained in a single block. The former case can be solved with two queries of the latter, so we assume, without losing generality, that $S$ is fully contained inside a block at level $0$. To retrieve $S$, we map it down to the next levels (using the stored coordinates of primary occurrences of half-blocks) as a contiguous text substring as long as this is possible, that is, as long as it fits inside a single half-block. 
Note that, thanks to the way half-blocks overlap, this is always possible as long as $\alpha \leq s_l/4$. By definition, then, we arrive in this way precisely at level $l^*$, where characters are stored explicitly and we can return the packed text substring. 
Figure~\ref{fig:extract} illustrates the data structure.
\qed
\end{proof}

\begin{figure}[t]
\begin{center}
\includegraphics[width=9cm]{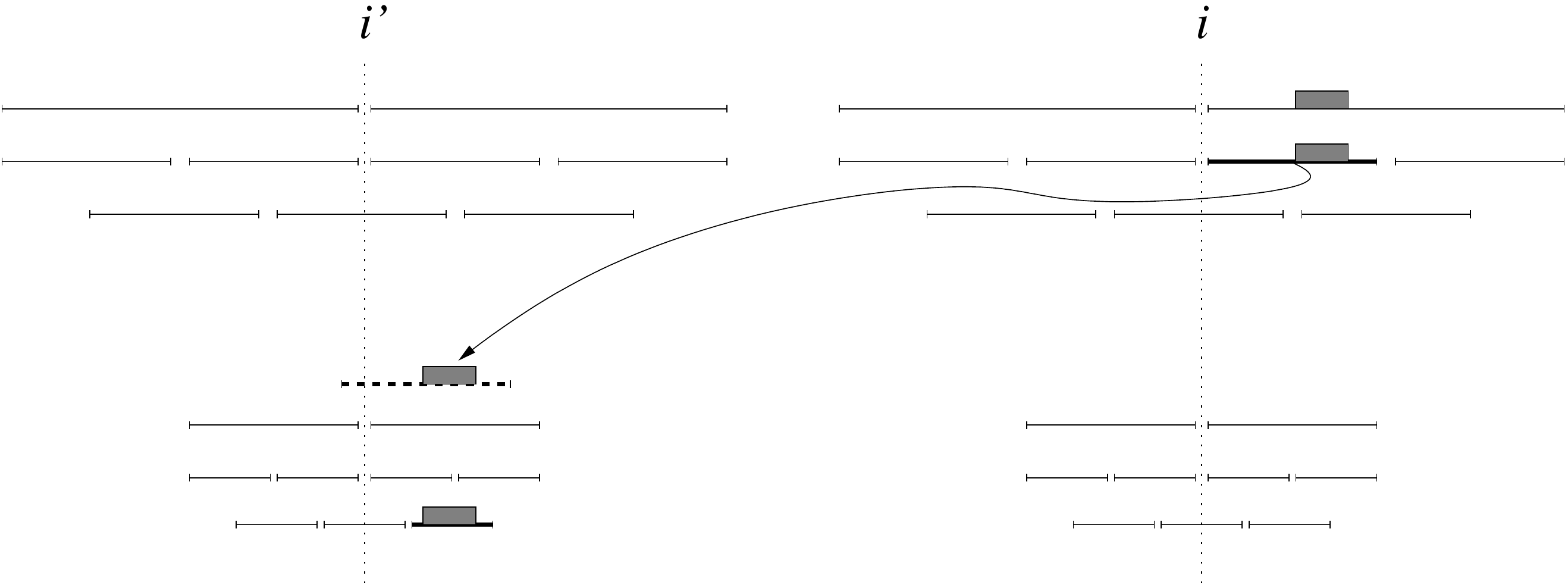}
\end{center}
\caption{Illustration of the proof of Theorem~\ref{thm:extract}. Extracting
the grayed square, we have arrived at a block around sampled position $i$ in
level $l$. Due to its size, the square must be contained in a half-block. This 
half-block (in thick line) has a copy crossing a sampled position $i'$ (we 
show this copy with a dashed line). Thus the extraction task is translated to 
level $l+1$, inside another block of half the length. Since the square is 
still small enough, it must fall inside some half-block of level $l+1$ (also
in thick line). This continues until the last level,
where the symbols are stored directly.}
\label{fig:extract}
\end{figure}

\subsection{Accessing $\SA$} \label{sec:dsa}

Let us define the differential suffix array $\DSA[p] = \SA[p]-\SA[p-1]$ for 
all $p>1$, and $\DSA[1]=\SA[1]$. The next lemmas show that the runs of $\BWT$
induce analogous repeated substrings in $\DSA$.

\begin{lemma} \label{lem:dsa}
Let $[p-1,p]$ be within a $\BWT$ run.
Then $\LF(p-1)=\LF(p)-1$ and $\DSA[\LF(p)]=\DSA[p]$.
\end{lemma}
\begin{proof}
Since $p$ is not the first position in a $\BWT$ run, it holds that
$\BWT[p-1] = \BWT[p]$, and thus $\LF(p-1)=\LF(p)-1$ follows from the 
formula of $\LF$. Therefore, if $q=\LF(p)$, we have $\SA[q]=\SA[p]-1$ and 
$\SA[q-1]=\SA[\LF(p-1)]= \SA[p-1]-1$; therefore $\DSA[q]=\DSA[p]$. 
\qed
\end{proof}


\begin{lemma} \label{lem:dsa2}
Let $[p-1..p+s]$ be within a $\BWT$ run, for some $1 < p \le n$ and
$0 \le s \le n-p$. Then there exists $q \not= p$ such that $\DSA[q..q+s] =
\DSA[p..p+s]$ and $[q-1..q+s]$ contains the first position of a $\BWT$ run.
\end{lemma}
\begin{proof}
By Lemma~\ref{lem:dsa}, it holds that $\DSA[p'..p'+s] = \DSA[p..p+s]$, where
$p' = \LF(p)$. If $\DSA[p'-1..p'+s]$ contains the first position of a
$\BWT$ run, we
are done. Otherwise, we apply Lemma~\ref{lem:dsa} again on $[p'..p'+s]$, and
repeat until we find a range that contains the first position of a run. This
search eventually terminates because there are $r>0$ run beginnings, there are
only $n-s+1$ distinct ranges,
and the sequence of visited ranges, $[\LF^k(p)..\LF^k(p)+s]$, forms a single
cycle; recall Section~\ref{sec:bwt}. Therefore our search will visit all the 
existing ranges before returning to $[p..p+s]$.
\qed
\end{proof}

This means that there exist $2r$ positions in $\DSA$, namely those $[q,q+1]$
where $\BWT[q]$ is the first position of a run, such that any substring
$\DSA[p..p+s]$ has a copy covering some of those $2r$ positions. Note that this 
is the same property of Lemma~\ref{lemma:primocc}, which
enabled efficient access and fingerprinting on $T$. We now exploit it to
access cells in $\SA$ by building a similar structure on $\DSA$.

\begin{theorem} \label{thm:dsa}
Let the $\BWT$ of a text $T[1..n]$ contain $r$ runs. Then
there exists a data structure using $O(r \log(n/r))$ words that
can retrieve any $\ell$ consecutive values of its suffix array $\SA$ in time 
$O(\log(n/r)+\ell)$.
\end{theorem}
\begin{proof}
We describe a data structure supporting the extraction of $\alpha = \log(n/r)$ consecutive 
cells in $O(\log(n/r))$ time. To extract $\ell$ consecutive cells
of $\SA$, we divide it into $\lceil\ell/\alpha\rceil$ blocks and extract each 
block independently. This yields the promised time complexity.

Our structure is stored in $O(\log(n/r))$ levels. As before, let us assume
that $r$ divides $n$ and that $n/r$ is a power of two.
At the top level ($l=0$), we divide $\DSA$ into $r$ blocks 
$\DSA[1..n/r],\DSA[n/r+1..2n/r],\dots,\DSA[n-n/r+1..n]$ of size $n/r$. For 
levels $l>0$, we let $s_l = n/(r\cdot 2^{l-1})$ and, for every position $q$ 
that starts a run in $\BWT$, we consider the two non-overlapping blocks of 
length $s_l$: $X_{l,q}^1 = \DSA[q-s_l+1..q]$ and
$X_{l,q}^2=\DSA[q+1..q+s_l]$.%
\footnote{Note that this symmetrically covers both positions $q$ and $q+1$; in 
Theorem~\ref{thm:extract}, one extra unnecessary position is covered with
$X_{l,q}^1$, for simplicity.}
Each such block $X^k_{l,q}$, for $k=1,2$, is composed of two half-blocks, 
$X^k_{l,q} = X^k_{l,q}[1..s_l/2]\,X^k_{l,q}[s_l/2+1..s_l]$.
We moreover consider three additional consecutive and non-overlapping 
half-blocks, starting in the middle of the first, $X^1_{l,q}[1..s_l/2]$, and 
ending in the middle of the last, $X^2_{l,q}[s_l/2+1..s_l]$, of the 4 
half-blocks just described: $\DSA[q-s_l+s_l/4+1..q-s_l/4],\ 
\DSA[q-s_l/4+1..q+s_l/4]$, and $\DSA[q+s_l/4+1..q+s_l-s_l/4]$.

From Lemma~\ref{lem:dsa2}, blocks at level $l=0$ and each half-block at level 
$l> 0$ have an occurrence covered by blocks at level $l+1$. Let the half-block
$X$ of level $l$ (blocks at level $0$ are analogous) have an occurrence 
containing position $q^* \in \{ q,q+1\}$, where $q$ starts a run in
$\BWT$. Then we store the pointer $\langle q^*,\mathit{off},\delta\rangle$ 
associated with $X$, where $0< \mathit{off} \leq s_{l+1}$ indicates that the 
occurrence of $X$ starts at position $q^*-s_{l+1}+ \mathit{off}$ of $\DSA$, and
$\delta = \SA[q^*-s_{l+1}] - \SA[q^*-s_{l+1}+ \mathit{off}-1]$. (We also store
the pointer to the data structure of the half-block of level $l+1$ containing
the position $q^*$.)

Additionally, every level-0 block $X'=\DSA[q'+1..q'+s_l]$ stores the 
value $S(X') = \SA[q']$ (assume $\SA[0]=0$ throughout), and every 
half-block $X' = \DSA[q'+1..q'+s_{l+1}/2]$ corresponding to the area 
$X_{l+1,q}^1 X_{l+1,q}^2 = \DSA[q-s_{l+1}+1..q+s_{l+1}]$ stores the value 
$\Delta(X') = \SA[q'] - \SA[q-s_{l+1}]$.

Let $l^*$ be the smallest number such that $s_{l^*} < 4\alpha = 4\log(n/r)$.
Then $l^*$ is the last level of our structure. At this level, we explicitly 
store the sequence of $\DSA$ cells of the areas $X_{l^*,q}^1 X_{l^*,q}^2$,
for each $q$ starting a run in $\BWT$. This uses in total 
$O(r \cdot s_{l^*}) = O(r\log(n/r))$ words of space. The pointers stored for
the $O(r)$ blocks at previous levels also add up to $O(r\log(n/r))$ words.

Let $S= \SA[p..p+\alpha-1]$ be the sequence of cells to be extracted. This range
either spans two blocks at level 0, or it is contained in a single block. In
the former case, we decompose it into two queries that are fully contained 
inside a block at level $0$. To retrieve a range contained in a single block
or half-block, we map it down to the next levels using the pointers from blocks
and half-blocks, as a contiguous sequence as long as it fits inside a single 
half-block. This is always possible as long as $\alpha\leq s_l/4$.
By definition, then, we arrive in this way precisely to level $l^*$, where 
the symbols of $\DSA$ are stored explicitly and we can return the sequence.

We need, however, the contents of $\SA[p..p+\alpha-1]$, not of
$\DSA[p..p+\alpha-1]$. To obtain the former from the latter, we need only the 
value of $\SA[p]$. During the traversal, we will maintain a value $f$ with the 
invariant that, whenever the original position $\DSA[p]$ has been mapped to a
position $X[p']$ in the current block $X$, then it holds that $\SA[p] = f +
X[1] + \ldots + X[p']$. This invariant must be maintained when we use
pointers, where the original $\DSA$ values in a block $X$ are obtained 
from a copy that appears somewhere else in $\DSA$.

The invariant is initially valid by setting $f$ to the $S(X)$ value 
associated with the level-0 block $X$ that contains $\SA[p]$. When we follow 
a pointer $\langle q,\mathit{offs},\delta \rangle$ and choose $X'$ from the 7
half-blocks that cover the target, we update $f \leftarrow f+\delta+\Delta(X')$.
When we arrive at a block $X$ at level $l^*$, we scan $O(\alpha)$ symbols 
until reaching the first value of the desired position $X[p']$. The values 
$X[1],\ldots,X[p']$ scanned are also summed to $f$. 
At the end, we have that $\SA[p]=f$.
See Figure~\ref{fig:sarray}.
\qed
\end{proof}

\begin{figure}[t]
\begin{center}
\includegraphics[width=8cm]{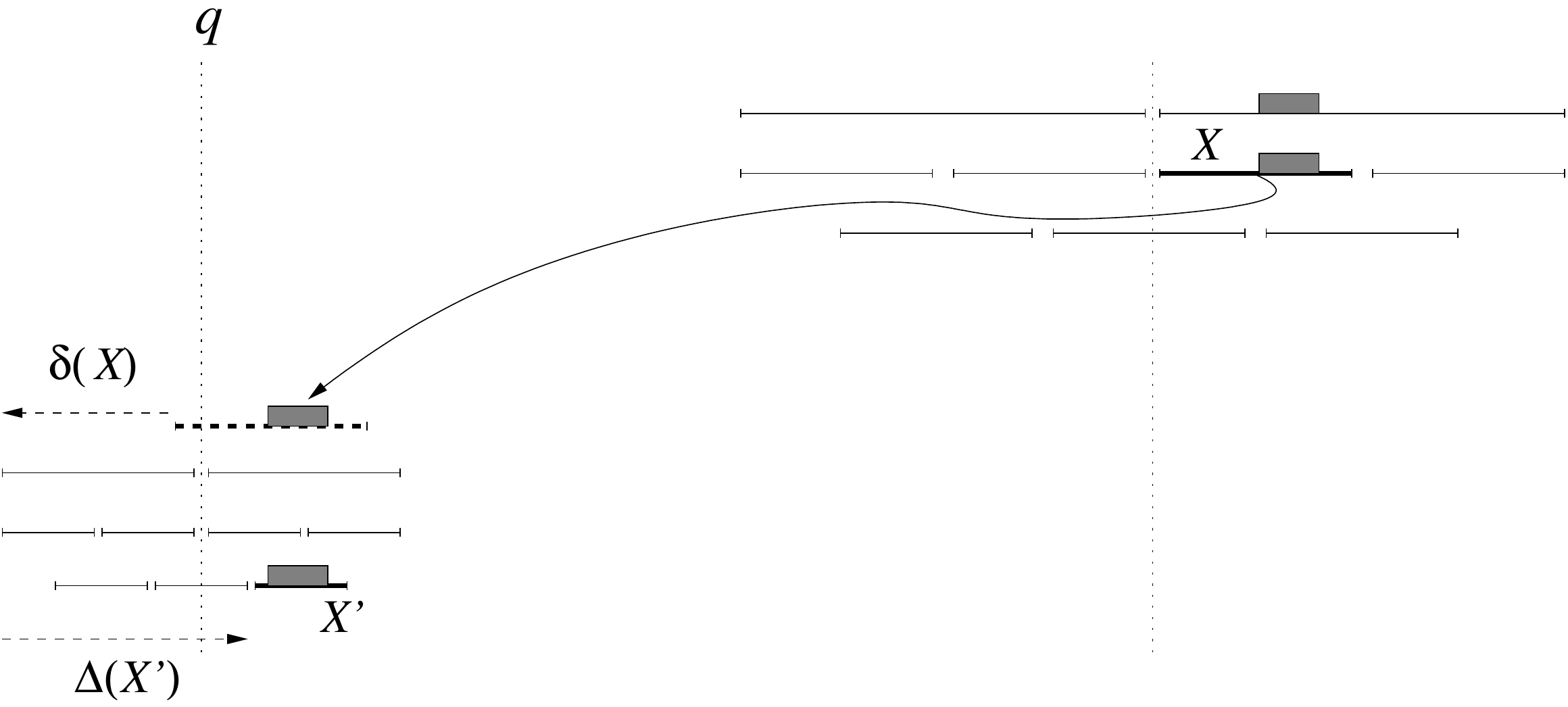}
\end{center}
\caption{Illustration of Theorem~\ref{thm:dsa}. The area to extract (a gray
square) is inside the thick half-block ($X$), which points inside another area
around position $q$ in the next level. The sum of $\DSA$ over the offset
from the beginning of the area to the mapped block (in thick dashed line)
is stored at $\delta(X)$, in negative (hence the direction of the arrow).
The squared area is mapped to a smaller half-block, $X'$, which records in
$\Delta(X')$ the sum of $\DSA$ between the beginning of the area and $X'$
(see the other dashed arrow). By adding $\delta(X)+\Delta(X')$, we map from
the first thick block to the second.}
\label{fig:sarray}
\end{figure}

\subsection{Accessing $\ISA$ and $\LCP$} \label{sec:isa}

A similar method can be used to access inverse suffix array cells, $\ISA[i]$.
Let us define $\DISA[i] = \ISA[i]-\ISA[i-1]$ for all $i>1$, and 
$\DISA[1]=\ISA[1]$. The role of the runs in $\SA$ will now be played by the
phrases in $\ISA$, which will be defined analogously as in the proof of
Lemma~\ref{lem:find_neighbours}: Phrases in $\ISA$ start at the positions
$\SA[p]$ such that a new run starts in $\BWT[p]$ (here, last positions of runs 
do not start phrases). Instead
of $\LF$, we use the cycle $\phi(i)$ of Definition~\ref{def:phi}.
We make use of the following lemmas.

\begin{lemma} \label{lem:disa}
Let $[i-1..i]$ be within a phrase of $\ISA$. Then it holds that 
$\phi(i-1)=\phi(i)-1$ and $\DISA[i]=\DISA[\phi(i)]$.
\end{lemma}
\begin{proof}
Consider the pair of positions $T[i-1..i]$ within a phrase. Let them be pointed 
from $\SA[p]=i$ and $\SA[q]=i-1$, therefore $\ISA[i] = p$, $\ISA[i-1]=q$,
and $\LF(p)=q$. Now, since $i$ is not a phrase 
beginning, $p$ is not the first position in a $\BWT$ run. Therefore,
$\BWT[p-1]=\BWT[p]$, from which
it follows that $\LF(p-1)=\LF(p)-1=q-1$. Now let $\SA[p-1]=j$, that is, 
$j=\phi(i)$. Then $\phi(i-1)=\SA[\ISA[i-1]-1]=\SA[q-1]=\SA[\LF(p-1)]=\SA[p-1]-1
=j-1=\phi(i)-1$.
It also follows that $\DISA[i]=p-q=\DISA[j]=
\DISA[\phi(i)]$. 
\qed
\end{proof}

\begin{lemma} \label{lem:disa2}
Let $[i-1..i+s]$ be within a phrase of $\DISA$, for some $1 < i \le n$ and
$0 \le s \le n-i$. Then there exists $j \not= i$ such that $\DISA[j..j+s] =
\DISA[i..i+s]$ and $[j-1..j+s]$ contains the first position of a phrase.
\end{lemma}
\begin{proof}
By Lemma~\ref{lem:disa}, it holds that $\DISA[i'..i'+s] = \DISA[i..i+s]$, where
$i' = \phi(i)$. If $\DISA[i'-1..i'+s]$ contains the first position of a
phrase, we
are done. Otherwise, we apply Lemma~\ref{lem:disa} again on $[i'..i'+s]$, and
repeat until we find a range that contains the first position of a phrase. 
Just as in Lemma~\ref{lem:dsa}, this search eventually terminates because 
$\phi$ is a permutation with a single cycle.
\qed
\end{proof}

We can then use on $\DISA$ exactly the same data structure we defined to 
access $\SA$ in Theorem~\ref{thm:dsa}, and obtain a similar result for $\ISA$.

\begin{theorem} \label{thm:disa}
Let the $\BWT$ of a text $T[1..n]$ contain $r$ runs. Then there exists a data 
structure using $O(r \log(n/r))$ words that can retrieve any $\ell$ 
consecutive values of its inverse suffix array $\ISA$ in time 
$O(\log(n/r)+\ell)$.
\end{theorem}

Finally, by combining Theorem~\ref{thm:dsa} and Lemma~\ref{lemma: lcp}, we
also obtain access to array $\LCP$ without knowing the corresponding text
positions.

\begin{theorem} \label{thm:dlcp}
Let the $\BWT$ of a text $T[1..n]$ contain $r$ runs. Then there exists a data 
structure using $O(r \log(n/r))$ words that can retrieve any $\ell$ 
consecutive values of its longest common prefix array $\LCP$ in time 
$O(\log(n/r)+\ell)$.
\end{theorem}
\begin{proof}
Build the structure of Theorem~\ref{thm:dsa}, as well as the one of
Lemma~\ref{lemma: lcp} with $s=\log(n/r)$. Then, to retrieve $\LCP[p..p+s'-1]$
for any $0 \le s' \le s$, we first compute $\SA[p]$ in time $O(\log(n/r))$ 
using Theorem~\ref{thm:dsa} and then, given $\SA[p]$, we compute 
$\LCP[p..p+s'-1]$ using Lemma~\ref{lemma: lcp} in time $O(\log\log_w(n/r)+s')$.
Adding both times gives $O(\log(n/r))$.

To retrieve an arbitrary sequence of cells $\LCP[p..p+\ell-1]$, we use the
method above by chunks of $s$ cells, plus a possibly smaller final chunk.
As we use $\lceil \ell/s \rceil$ chunks, the total time is $O(\log(n/r)+\ell)$.
\qed
\end{proof}

\subsection{Optimal counting and locating in $O(r\log(n/r))$ space}

The $O(r\log(n/r))$ space we need for accessing $T$ is not comparable 
with the $O(r\log\log_w(\sigma+n/r))$ space we need for optimal counting and 
locating. The latter is in general more attractive, because the former is
better whenever $r = \omega (n/\log_w \sigma)$, which means that the text
is not very compressible. Anyway, we show how to obtain optimal counting and
locating within space $O(r\log(n/r))$.

By the discussion above, we only have to care about the case $r \ge n/\log n$.
In such a case, it holds that $r\log(n/r) \ge (n \log\log n)/\log n$,%
\footnote{Since $r\log(n/r)$ grows with $r$ up to $r=n/e$, we obtain the lower
bound by evaluating it at $r = n/\log n$.} and thus we are allowed to use
$\Theta(n\log\log n)$ bits of space.
We can then make use of a result of Belazzougui and Navarro
\cite[Lem.~6]{BN13}. They show how we can enrich the $O(n)$-bit compressed
suffix tree of Sadakane \cite{Sad07} so that, using $O(n(\log t_\SA +
\log\log\sigma))$ bits, one can find the interval $\SA[sp..ep]$ of $P$ in time
$O(m+t_\SA)$ plus the time to extract a substring of length $m$ from $T$.%
\footnote{The $O(n\log\log\sigma)$ bits of the space are not explicit in their
lemma, but are required in their Section 5, which is used to prove their Lemma 6.}
Since we provide $t_\SA=O(\log(n/r))$ in Theorem~\ref{thm:dsa} and extraction
time $O(\log(n/r)+m\log(\sigma)/w)$ in Theorem~\ref{thm:extract}, this
arrangement uses $O(n(\log\log(n/r)+\log\log\sigma)) \subseteq O(n\log\log n)$ bits,
and it supports counting in time $O(m+\log(n/r))$. 

Once we know the interval, apart from counting, we can use 
Theorem~\ref{thm:dsa} to obtain $\SA[p]$ for any $sp \le p \le ep$ in time
$O(\log(n/r))$, and then use the structure of Lemma~\ref{lemma: general locate}
with $s=\log(n/r)$ to extract packs of $s' \le s$ consecutive $\SA$ entries in 
time $O(\log\log_w(n/r) + s') \subseteq O(\log(n/r)+s)$. Overall, we can locate
the $occ$ occurrences of $P$ in time $O(m+\log(n/r)+occ)$.

Finally, to remove the $O(\log(n/r))$ term in the times, we must
speed up the searches for patterns shorter than $\log(n/r)$. We index them using a compact trie as that of Section~\ref{sec:mapping}. We store in each explicit trie node (i) the number of occurrences of the corresponding string, to support counting, and (ii) a position $p$ where it occurs in $\SA$, the value $\SA[p]$, and the result of the predecessor queries on $P^+$ and $P^-$, as required for locating in Lemma~\ref{lemma: general locate}, so that we can retrieve any number $s' \le s$ of consecutive entries of $\SA$ in time $O(s')$. By Lemma~\ref{lemma:distinct kmers}, the size of the trie and of the text substrings explicitly stored to support path compression is $O(r\log(n/r))$. 

\begin{theorem} \label{thm:optimal 2}
We can store a text $T[1..n]$, over alphabet $[1..\sigma]$, in
$O(r \log(n/r))$ words, where $r$ is the number of runs
in the BWT of $T$, such that later, given a pattern $P[1..m]$, we can count the
occurrences of $P$ in $T$ in $O(m)$ time and
(after counting) report their $occ$ locations in overall time $O(occ)$.
\end{theorem}

\section{A Run-Length Compressed Suffix Tree} \label{sec:stree}

In this section we show how to implement a compressed suffix tree within
$O(r\log(n/r))$ words, which solves a large set of navigation operations
in time $O(\log(n/r))$. The only exceptions are going to a child by some letter
and performing level ancestor queries, which may cost as much as 
$O(\log(n/r)\log n)$. The first compressed suffix tree for repetitive 
collections was built on runs \cite{MNSV09}, but just like the self-index, it 
needed $\Theta(n/s)$ space to obtain $O(s\log n)$ time in key
operations like accessing $\SA$. Other compressed suffix trees for repetitive 
collections appeared later \cite{ACN13,NO16,FGNPS17}, but they do not offer 
formal space guarantees (see later). A recent one, instead,
uses $O(\overline{e})$ words and supports a number of operations in time 
typically $O(\log n)$ \cite{BC17}. The two space measures are not comparable.

\subsection{Compressed Suffix Trees without Storing the Tree}

Fischer et al.~\cite{FMN09} showed that a rather complete suffix tree
functionality including all the operations in Table~\ref{tab:stree}
can be efficiently supported by a representation where suffix tree nodes $v$
are identified with the suffix array intervals $\SA[v_l..v_r]$ they cover. 
Their representation builds on the following primitives:
\begin{enumerate}
\item Access to arrays $\SA$ and $\ISA$, in time we call $t_\SA$. 
\item Access to array $\LCP$, in time we call $t_\LCP$.
\item Three special queries on $\LCP$:
\begin{enumerate}
	\item Range Minimum Query, $$\RMQ(i,j)=\arg\min_{i\le k\le j} \LCP[k],$$
	choosing the leftmost one upon ties, in time we call $t_\RMQ$.
	\item Previous/Next Smaller Value queries,
\begin{eqnarray*}
\PSV(p)&=&\max (\{q<p,\LCP[q]<\LCP[p]\} \cup \{0\}), \\
\NSV(p)&=&\min (\{q>p,\LCP[q]<\LCP[p]\} \cup \{n+1\}),
\end{eqnarray*}
		in time we call $t_\SV$.
\end{enumerate}
\end{enumerate}

\begin{table}[t]
\begin{center}
\begin{tabular}{ll}
Operation & Description \\
\hline
{\em Root}() & Suffix tree root.\\
{\em Locate}($v$) & Text position $i$ of leaf $v$.\\
{\em Ancestor}($v,w$) & Whether $v$ is an ancestor of $w$.\\
{\em SDepth}($v$) & String depth for internal nodes, i.e., length of string represented by $v$. \\
{\em TDepth}($v$) & Tree depth, i.e., depth of tree node $v$. \\
{\em Count}($v$) & Number of leaves in the subtree of $v$. \\
{\em Parent}($v$) & Parent of $v$. \\
{\em FChild}($v$) & First child of $v$.\\
{\em NSibling}($v$) & Next sibling of $v$.\\
{\em SLink}($v$) & Suffix-link, i.e., if $v$ represents $a\cdot\alpha$ then
the node that represents $\alpha$, for $a\in [1..\sigma]$.\\
{\em WLink}($v,a$) & Weiner-link, i.e., if $v$ represents $\alpha$ then
the node that represents $a \cdot \alpha$.\\
{\em SLink}$^i$($v$) & Iterated suffix-link.\\
{\em LCA}($v,w$) & Lowest common ancestor of $v$ and $w$.\\
{\em Child}($v,a$) & Child of $v$ by letter $a$.\\
{\em Letter}($v,i$) & The $ith$ letter of the string represented by $v$.\\
{\em LAQ}$_S$($v,d$) & String level ancestor, i.e., the highest ancestor of $v$ with string-depth $\ge d$.\\
{\em LAQ}$_T$($v,d$) & Tree level ancestor, i.e., the ancestor of $v$ with tree-depth $d$.\\
\hline
\end{tabular}
\end{center}
\caption{Suffix tree operations.}
\label{tab:stree}
\end{table}

An interesting finding of Fischer et al.~\cite{FMN09} related to our results 
is that array $\PLCP$, which stores the $\LCP$ values in text order, can be 
stored in $O(r)$ words and accessed efficiently; therefore we can compute 
any $\LCP$ value in time $t_\SA$ (see also Fischer~\cite{Fis10}). 
We obtained a generalization of this property in Section~\ref{sec:lcp}.
They \cite{FMN09} also show how to represent the array $\TDE[1..n]$,
where $\TDE[i]$ is the tree-depth of the lowest common ancestor of the 
$(i-1)$th and $i$th suffix tree leaves (and $\TDE[1]=0$). Fischer et 
al.~\cite{FMN09} represent
its values in text order in an array $\PTDE$, which just like $\PLCP$ can be 
stored in $O(r)$ words and accessed efficiently, thereby giving access to $\TDE$
in time $t_\SA$. They use $\TDE$ to compute operations {\em TDepth} and 
{\em LAQ}$_T$ efficiently.

Abeliuk et al.~\cite{ACN13} show that primitives $\RMQ$, $\PSV$, and $\NSV$
can be implemented using a simplified variant of {\em range min-Max trees
(rmM-trees)} \cite{NS14}, consisting of a perfect binary tree on top of $\LCP$ 
where each node stores the minimum $\LCP$ value in its subtree. 
The three primitives are then computed in logarithmic time. They 
define the extended primitives
\begin{eqnarray*}
\PSV'(p,d)&=&\max (\{q<p,\LCP[q]<d\} \cup \{0\}), \\
\NSV'(p,d)&=&\min (\{q>p,\LCP[q]<d\} \cup \{n+1\}),
\end{eqnarray*}
and compute them in time $t_{\SV'}$, which in their setting is the same
$t_\SV$ of the basic $\PSV$ and $\NSV$ primitives. The extended primitives
are used to simplify some of the operations of Fischer et al.~\cite{FMN09}. 

The resulting time complexities are given in the second column of
Table~\ref{tab:streeops}, where $t_\LF$ is the time to compute function $\LF$ 
or its inverse, or to access a position in $\BWT$. Operation {\em WLink},
not present in Fischer et al.~\cite{FMN09}, is trivially obtained with two
$\LF$-steps. We note that most times 
appear multiplied by $t_\LCP$ in Fischer et al.~\cite{FMN09} because
their $\RMQ$, $\PSV$, and $\NSV$ structures do not store $\LCP$ values inside,
so they need to access the array all the time; this is not the case when we
use rmM-trees. The time of {\em LAQ}$_S$ is due to 
improvements obtained with the extended primitives $\PSV'$ and $\NSV'$ 
\cite{ACN13}.%
\footnote{They also use these primitives for {\em NSibling}, mentioning that
the original formula has a bug. Since we obtain better $t_\RMQ$ than
$t_{\SV'}$ time, we rather prefer to fix the original bug \cite{FMN09}.
The formula fails for the penultimate child of its parent.
To compute the next sibling of $[v_l,v_r]$ with parent $[w_l,w_r]$, the
original formula $[v_r+1,u]$ with $u=\RMQ(v_r+2,w_r)-1$ (used only if
$v_r < w_r-1$) must now be checked as follows: if $u < w_r$ and 
$\LCP[v_r+1] \not= \LCP[u+1]$, then correct it to $u=w_r$.}
The time for {\em Child}($v,a$) is obtained by binary searching 
among the $\sigma$ minima of $\LCP[v_l,v_r]$, and extracting the desired
letter (at position {\em SDepth}$(v)+1$) to compare with $a$.
Each binary search operation can be done with an extended primitive
$\RMQ'(p,q,m)$ that finds the $m$th left-to-right occurrence
of the minimum in a range. This is easily done in $t_{\RMQ'} = t_\RMQ$ time on
a rmM-tree by storing, in addition, the number of times the minimum of each 
node occurs below it \cite{NS14}, but it may be not so easy to do on other
structures. Finally, the complexities of {\em TDepth} and {\em LAQ}$_T$ make 
use of array $\TDE$. While Fischer et al.~\cite{FMN09} use an $\RMQ$ operation 
to compute {\em TDepth}, we note that
{\em TDepth}$(v) = 1 + \max(\TDE[v_l],\TDE[v_r+1])$, because the suffix tree 
has no unary nodes (they used this simpler formula only for leaves).%
\footnote{We observe that {\em LAQ}$_T$ can be solved exactly as {\em LAQ}$_S$,
with the extended $\PSV'/\NSV'$ operations, now defined on the array $\TDE$ 
instead of on $\LCP$. However, an equivalent to 
Lemma~\ref{lem:dlcp} for the differential $\TDE$ array does not hold, and
therefore we cannot use that solution within the desired space bounds.}

\begin{table}[t]
\begin{center}
\begin{tabular}{lcc}
Operation & Generic    & Our \\
          & Complexity & Complexity \\
\hline
{\em Root}() & $1$ & $1$ \\
{\em Locate}($v$) & $t_\SA$ & $\log(n/r)$ \\
{\em Ancestor}($v,w$) & $1$ & $1$ \\
{\em SDepth}($v$) & $t_\RMQ + t_\LCP$ & $\log(n/r)$ \\
{\em TDepth}($v$) & $t_\SA$ & $\log(n/r)$ \\
{\em Count}($v$) &  $1$ & $1$ \\
{\em Parent}($v$) & $t_\LCP + t_\SV$ & $\log(n/r)$ \\
{\em FChild}($v$) & $t_\RMQ$ & $\log(n/r)$ \\
{\em NSibling}($v$) & $t_\LCP + t_\RMQ$ & $\log(n/r)$ \\
{\em SLink}($v$) & $t_\LF + t_\RMQ + t_\SV$ & $\log(n/r)$ \\
{\em WLink}($v$) & $t_\LF$ & $\log\log_w(n/r)$ \\
{\em SLink}$^i$($v$) & $t_\SA + t_\RMQ + t_\SV$ & $\log(n/r)$ \\
{\em LCA}($v,w$) & $t_\RMQ + t_\SV$ & $\log(n/r)$ \\
{\em Child}($v,a$) & $t_\LCP + (t_{\RMQ'} + t_\SA + t_\LF)\log\sigma$ & 
				$~~~~\log(n/r) \log\sigma~~~~$ \\
{\em Letter}($v,i$) & $t_\SA + t_\LF$ & $\log(n/r)$ \\
{\em LAQ}$_S$($v,d$) & $t_{\SV'}$ & $\log(n/r) + \log\log_w r$ \\
{\em LAQ}$_T$($v,d$) & $(t_\RMQ+t_\LCP)\log n$ & $\log(n/r)\log n$ \\
\hline
\end{tabular}
\end{center}
\caption{Complexities of suffix tree operations. {\em Letter}($v,i$) can also 
be solved in time $O(i\cdot t_\LF) = O(i \log\log_w(n/r))$.}
\label{tab:streeops}
\end{table}

An important idea of Abeliuk et al.~\cite{ACN13} is that they represent
$\LCP$ differentially, that is, the array $\DLCP[1..n]$, where 
$\DLCP[i]=\LCP[i]-LCP[i-1]$ if $i>1$ and $\DLCP[1]=\LCP[1]$,
using a {\em context-free grammar (CFG)}.
Further, they store the rmM-tree information in the nonterminals, that is, 
a nonterminal $X$ expanding to a substring $D$ of $\DLCP$ stores the 
(relative) minimum 
$$m(X) ~=~ \min_{0 \le k \le |D|} \sum_{i=1}^k D[i]$$ 
of any $\LCP$ segment having those differential values, and its position 
inside the segment,
$$p(X) ~=~\arg\min_{0 \le k \le |D|} \sum_{i=1}^k D[i].$$
Thus, instead of a perfect rmM-tree, they conceptually use the grammar tree
as an rmM-tree. They show how to adapt the algorithms on the perfect rmM-tree 
to run on the grammar, and thus solve primitives $\RMQ$,
$\PSV'$, and $\NSV'$, in time proportional to the grammar height.

Abeliuk et al.~\cite{ACN13}, and also Fischer et al.~\cite{FMN09}, claim that
the grammar produced by RePair \cite{LM00} is of size $O(r\log(n/r))$. This is 
an incorrect result borrowed from Gonz\'alez et al.~\cite{GN07,GNF14}, where it
was claimed for $\DSA$. The proof fails for a reason we describe in our 
technical report \cite[Sec.~A]{GNP17}.

We now start by showing how to build a grammar of size
$O(r\log(n/r))$ and height $O(\log(n/r))$ for $\DLCP$.
This grammar is of an extended type called {\em run-length context-free grammar
(RLCFG)} \cite{NIIBT15}, which allows rules of the form $X \rightarrow Y^t$
that count as size 1. 
We then show how to implement the operations $\RMQ$ and $\NSV/\PSV$ in time
$O(\log(n/r))$ on the resulting RLCFG, and $\NSV'/\PSV'$ in time
$O(\log(n/r)+\log\log_w r)$. Finally, although we cannot implement
$\RMQ'$ in time below $\Theta(\log n)$, we show how the specific {\em Child} 
operation can be implemented in time $O(\log(n/r)\log\sigma)$.

Note that, although we could represent $\DLCP$ using a Block-Tree-like 
structure as we did in Section~\ref{sec:sa} for $\DSA$ and $\DISA$, we have 
not devised a way to implement the more complex operations we need on $\DLCP$
using such a Block-Tree-like data structure within polylogarithmic time.

Using the results we obtain in this and previous sections, that is, 
$t_\SA = O(\log(n/r))$, $t_\LF = O(\log\log_w(n/r))$, 
$t_\LCP = t_\SA + O(\log\log_w(n/r)) = O(\log(n/r))$,
$t_\RMQ = t_\SV = O(\log(n/r))$, $t_{SV'} = O(\log(n/r)+\log\log_w r)$,
and our specialized algorithm for {\em Child}, we obtain our result.

\begin{theorem} \label{thm:stree}
Let the $\BWT$ of a text $T[1..n]$, over alphabet $[1..\sigma]$, contain $r$ 
runs. Then a compressed suffix
tree on $T$ can be represented using $O(r \log(n/r))$ words, and it
supports the operations with the complexities given in the third column
of Table~\ref{tab:streeops}.
\end{theorem}

\subsection{Representing $\DLCP$ with a run-length grammar} \label{sec:dlcp}

In this section we show that the differential array $\DLCP$ can be represented
by a RLCFG of size $O(r\log(n/r))$. We first prove a lemma analogous to
those of Section~\ref{sec:sa}.

\begin{lemma} \label{lem:dlcp}
Let $[p-2,p]$ be within a $\BWT$ run. Then $\LF(p-1)=\LF(p)-1$ and 
$\DLCP[\LF(p)]=\DLCP[p]$.
\end{lemma}
\begin{proof}
Let $i=\SA[p]$, $j=\SA[p-1]$, and $k=\SA[p-2]$.
Then $\LCP[p]=lcp(T[i..],T[j..])$ and $\LCP[p-1]=lcp(T[j..],T[k..])$.
We know from Lemma~\ref{lem:dsa} that, if $q=\LF(p)$, then
$\LF(p-1)=q-1$ and $\LF(p-2)=q-2$. Also, $\SA[q]=i-1$, $\SA[q-1]=j-1$,
and $\SA[q-2]=k-1$.
Therefore, $\LCP[\LF(p)]=\LCP[q]=lcp(T[\SA[q]..],T[\SA[q-1]..) =
lcp(T[i-1..],T[j-1..])$. Since $p$ is not the first position in a $\BWT$
run, it holds that $T[j-1] = \BWT[p-1] = \BWT[p] = T[i-1]$,
and thus $lcp(T[i-1..],T[j-1..])=1+lcp(T[i..],T[j..])=1+\LCP[p]$.
Similarly, $\LCP[\LF(p)-1]=\LCP[q-1]=lcp(T[\SA[q-1]..],T[\SA[q-2]..) =
lcp(T[j-1..],T[k-1..])$. Since $p-1$ is not the first position in a $\BWT$
run, it holds that $T[k-1] = \BWT[p-2] = \BWT[p-1] = T[j-1]$,
and thus $lcp(T[j-1..],T[k-1..])=1+lcp(T[j..],T[k..])=1+\LCP[p-1]$.
Therefore $\DLCP[q]=\LCP[q]-\LCP[q-1]=(1+\LCP[p])-(1+\LCP[p-1])=\DLCP[p]$.
\qed
\end{proof}

It follows that, if $p_1 < \ldots < p_r$ are the positions that start runs in
$\BWT$, then we can define a {\em bidirectional macro scheme} \cite{SS82}
of size at most $4r+1$ on $\DLCP$. 

\begin{definition}
A {\em bidirectional macro scheme (BMS)} of size $b$ on a sequence $S[1..n]$ 
is a partition $S = S_1 \ldots S_b$ such that each $S_k$ is of length 1 (and
is represented as an explicit symbol) or it appears somewhere else in $S$ (and 
is represented by a pointer to that other occurrence). Let $f(i)$, for $1 \le i
\le n$, be defined arbitrarily if $S[i]$ is an explicit symbol, and $f(i)=
j+i'-1$ if $S[i] = S_k[i']$ is inside some $S_k$ that is represented as a 
pointer to $S[j..j']$. A correct BMS must hold that, for any $i$, there is
a $k \ge 0$ such that $f^k(i)$ is an explicit symbol.
\end{definition}

Note that $f(i)$ maps the position $S[i]$ to the source from which it is to
be obtained. The last condition then ensures that we can recover any symbol 
$S[i]$ by following the chain of copies until finding an explicitly stored 
symbol. Finally, note that all the $f$ values inside a block are consecutive:
if $S_k = S[i..i']$ has a pointer to $S[j..j']$, then 
$f([i..i']) = [j..j']$.

\begin{lemma}
Let $p_1 < \ldots < p_r$ be the positions that start runs in $\BWT$, and assume
$p_0=-2$ and $p_{r+1}=n+1$. Then, the partition formed by (1) all the explicit
symbols $\DLCP[p_i+k]$ for $1 \le i \le r$ and $k \in \{0,1,2\}$, and (2) all 
the nonempty regions $\DLCP[p_i+3..p_{i+1}-1]$ for all $0 \le i \le r$, 
pointing to $\DLCP[\LF(p_i+3)..\LF(p_{i+1}-1)]$, is a BMS.
\end{lemma}
\begin{proof}
By Lemma~\ref{lem:dlcp}, it holds that $\LF(p_i+3+k)=\LF(p_i+3)+k$ and
$\DLCP[p_i+3+k] = \DLCP[\LF(p_i+3)+k]$ for all $0 \le k \le p_{i+1}-p_i-4$, so 
the partition is well defined and the copies are correct. To see that
it is a BMS, it is sufficient to notice that $\LF$ is a permutation with one
cycle on $[1..n]$, and therefore $\LF^k(p)$ will eventually reach an explicit 
symbol, for some $0 \le k < n$.
\qed
\end{proof}

We now make use of the following result.

\begin{lemma}[{\cite[Thm.~1]{GNPlatin18}}] \label{lem:gnp18}
Let $S[1..n]$ have a BMS of size $b$. Then there exists a RLCFG of size 
$O(b\log(n/b))$ that generates $S$.
\end{lemma}

Since $\DLCP$ has a BMS of size at most $4r+1$, the following corollary is 
immediate.

\begin{lemma} \label{lem:dlcp-grammar}
Let the $\BWT$ of $T[1..n]$ have $r$ runs. Then there exists a RLCFG of size
$O(r\log(n/r))$ that generates its differential $\LCP$ array, $\DLCP$.
\end{lemma}

\subsection{Supporting the primitives on the run-length grammar}
\label{sec:rlaccess}

We describe how to compute the primitives $\RMQ$ and $\PSV/\NSV$ on 
the RLCFG of $\DLCP$, in time $t_\RMQ = t_\SV = O(\log(n/r))$. The extended
primitives
$\PSV'/\NSV'$ are solved in time $t_{SV'}=O(\log(n/r)+\log\log_w r)$.
While analogous procedures have been described before on CFGs and trees 
\cite{ACN13,NS14}, the extension to RLCFGs and the particular structure of 
our grammar requires a complete description. 

The RLCFG built in Lemma~\ref{lem:gnp18} \cite{GNPlatin18} is of height 
$O(\log(n/r))$ and has one
initial rule $S \rightarrow X_1 \ldots X_{O(r)}$. The other rules are of the
form $X \rightarrow Y_1 Y_2$ or $X \rightarrow Y^t$ for $t > 2$. All the 
right-hand symbols can be terminals or nonterminals.

The data structure we use is formed by a sequence $\DLCP' = X_1 \ldots X_{O(r)}$
capturing the initial rule of the RLCFG, and an array of the other
$O(r\log(n/r))$ rules. For each nonterminal $X$ expanding to a substring $D$
of $\DLCP$, we store its length $l(X) = |D|$ and its total difference
$d(X) = D[1] + \ldots + D[l(X)]$. For terminals $X$, assume $l(X)=1$ and
$d(X)=X$. We also store a cumulative length array
$L[0]=0$ and $L[x] = L[x-1] + l(X_x)$ that can be binary searched to find
the symbol of $\DLCP'$ that contains any desired position $\DLCP[p]$. To
ensure that this binary search takes time $O(\log(n/r))$ when $r = 
\omega(n/r)$, we can store a sampled
array of positions $S[1..r]$, where $S[t] = x$ if $L[x-1] < t\cdot(n/r) \le
L[x]$ to narrow down the binary search to a range of $O(n/r)$ entries of $L$. 
We also store a cumulative differences array $A[0]=0$ and
$A[x] = A[x-1] + d(X_x)$.

Although we have already provided access to any $\LCP[p]$ in 
Section~\ref{sec:isa}, it is also possible to do it with these structures. 
We first find $x$ by binary searching $L$ for $p$, possibly with the help of 
$S$, and set $f \leftarrow A[x-1]$
and $p \leftarrow p-L[x-1]$. Then we enter recursively into nonterminal $X =
X_x$. If its rule is $X \rightarrow Y_1 Y_2$, we continue by $Y_1$ if $p\le 
l(Y_1)$; otherwise we set $f \leftarrow f+d(Y_1)$, $p \leftarrow p-l(Y_1)$, 
and continue by $Y_2$. If, instead, its rule is $X \rightarrow Y^t$, we compute
$t' = \lceil p/l(Y) \rceil - 1$, set $f \leftarrow f+t'\cdot d(Y)$,
$p \leftarrow p-t' \cdot l(Y)$, and continue by $Y$. When we finally arrive
at a terminal $X$, the answer is $f+d(X)$. All this process takes time
$O(\log(n/r))$, the height of the RLCFG.

\paragraph{\bf Answering $\RMQ$.}
To answer this query, we store a few additional structures. We store an array 
$M$ such that $M[x] = \min_{L[x-1] < k \le L[x]} \LCP[k]$, that is, the minimum
value in the area of $\LCP$ expanded by $X_x = \DLCP'[x]$. We store a 
succinct data structure $\RMQ_M$, which requires just $O(r)$ bits and 
finds the leftmost position of a minimum in any range $M[x..y]$ in 
constant time, without need to access $M$ \cite{FH11}. 
We also store, for each nonterminal $X$, the already defined values $m(X)$ and 
$p(X)$ (for terminals $X$, we can store $m(X)$ and $p(X)$ or compute them on 
the fly).

To compute $\RMQ(p,q)$ on $\LCP$, we first use $L$ and $S$ to determine that 
$\DLCP[p..q]$ contains the expansion of $\DLCP'[x+1..y-1]$, whereas
$\DLCP'[x..y]$ expands to $\DLCP[p'..q']$ with $p' < p \le q < q'$. Thus, 
$\DLCP[p..q]$ partially overlaps $\DLCP'[x]$ and $\DLCP'[y]$ (the overlap 
could be empty). We first obtain in constant time the minimum position of 
the central area, $z = \RMQ_M(x+1,y-1)$, and then the minimum value in that
area is $\LCP[L[z-1]+p(X_z)]$. To complete the query, we must compare this
value with the minima in 
$X_x \langle p-p'+1,l(X_x)\rangle$ and $X_y \langle 1,l(X_y)+q-q' \rangle$,
where $X \langle a,b \rangle$ refers to the substring $D[a..b]$ in the 
expansion $D$ of $X$. A relevant special case in this scheme is that 
$\DLCP[p..q]$ is inside a single symbol $\DLCP'[x]$ expanding to
$\DLCP[p'..q']$, in which case the query boils down to
finding the minimum value in $X_x \langle p-p'+1,l(X_x)+q-q' \rangle$.

Let us disregard the rules $X \rightarrow Y^t$ for a moment.
To find the minimum in $X_w \langle a,b \rangle$, we identify the $k =
O(\log(n/r))$ maximal nodes of the grammar tree that cover the range $[a..b]$ 
in the expansion of $X_w$. Let these nodes be $Y_1, Y_2, \ldots, Y_k$.
We then find the minimum of $m(Y_1)$, $d(Y_1)+m(Y_2)$, $d(Y_1)+d(Y_2)+m(Y_3)$, 
$\ldots$, in $O(k)$ time. Once the minimum is identified at $Y_s$, we obtain
the absolute value by extracting $\LCP[L[w-1]+l(Y_1)+\ldots+l(Y_{s-1})+
p(Y_s)]$.

Our grammar also has rules of the form $X \rightarrow Y^t$, and thus the
maximal coverage $Y_1,\ldots,Y_k$ may include a part of these rules, say
$Y^{t'}$ for some $1 \le t' < t$. We can then compute on the fly
$m(Y^{t'})$ as $m(Y)$ if $d(Y) \ge 0$, and $(t'-1)\cdot d(Y)+m(Y)$ otherwise. 
Similarly, $p(Y^{t'})$ is $p(Y)$ if $d(Y) \ge 0$ and $(t'-1)\cdot l(Y)+p(Y)$ 
otherwise.

Once we have the (up to) three minima from $X_x$, $\DLCP'[x+1..y-1]$, and
$X_y$, the position of the smallest of the three is $\RMQ(p,q)$. 

\paragraph{\bf Answering $\PSV/\NSV$ and $\PSV'/\NSV'$.}

These queries are solved analogously. Let us describe $\NSV'(p,d)$, since 
$\PSV'(p,d)$ is similar. Let $\DLCP[p..]$ be included
in the expansion of $\DLCP'[x..]$, which expands to $\DLCP[p'..]$ (for the
largest possible $p' \le p$), and let us subtract $\LCP[p'-1]=A[x-1]$ from $d$ 
to put it in relative form. We first
consider $X_x \langle p-p'+1,l(X_x) \rangle = X\langle a,b\rangle$, obtaining
the $O(\log(n/r))$ maximal nonterminals $Y_1, Y_2, \ldots, Y_k$ that cover 
$X\langle a,b\rangle$, and find the first $Y_s$ where $d(Y_1)+\ldots+d(Y_{s-1})
+m(Y_s) < d$. Then we subtract $d(Y_1)+\ldots+d(Y_{s-1})$ from $d$, add
$l(Y_1)+\ldots+l(Y_{s-1})$ to $p$, and continue recursively inside 
$Y_s$ to find the precise point where the cumulative differences fall below $d$.

The recursive traversal from $Y_s$ works as follows. If $Y_s \rightarrow Y_1
Y_2$, we first see if $m(Y_1) < d$. If so, we continue recursively on $Y_1$;
otherwise, we subtract $d(Y_1)$ from $d$, add $l(Y_1)$ to $p$, and continue 
recursively on $Y_2$. If, instead, the rule is $Y_s \rightarrow Y^t$, we 
proceed as follows. If $d(Y) \ge 0$, then the answer must be in the first copy 
of $Y$, thus we recursively continue on $Y$. If $d(Y) < 0$, instead, we must 
find the copy $t'$ into which we continue. This is the smallest $t'$ such
that $(t'-1) \cdot d(Y) + m(Y) < d$, that is, 
$t' = \max(1,2+\lfloor (d - m(Y))/d(Y) \rfloor)$. Thus we subtract 
$(t'-1)\cdot d(Y)$
from $d$, add $(t'-1)\cdot l(Y)$ to $p$, and continue with $Y$.
Finally, when we arrive at a terminal $X$, it holds that $m(X) < d$ and the
answer to the query is the current value of $p$. All of this process takes time
$O(\log(n/r))$, the height of the grammar.

It might be, however, that we traverse $Y_1, Y_2,\ldots,Y_k$, that is, the whole
$X_x \langle p-p'+1,l(X_x)\rangle$, and still do not find a value below $d$.
We then must find where we fall below (the current value of) $d$ inside
$\DLCP'[x+1..]$. Once this search identifies the leftmost position $\DLCP'[z]$ 
where the answer lies, we complete the search on $X_z \langle 1,l(X_z) \rangle$
as before, for $d \leftarrow d-A[z-1]+A[x]$.

The search problem can be regarded as follows: Given the array
$B[z] = A[z]+m(X_z)$, find the leftmost position $z > x$ such that $B[z] <
A[x]+d$. Navarro and Sadakane \cite[Sec.~5.1]{NS14} show that this query can
be converted into a weighted ancestor query on a tree: given nodes with
weights that decrease toward the root, the query gives a node $v$ and a
weight $w$ and seeks for its nearest ancestor with weight $<w$. In our case,
the tree has $O(r)$ nodes and the weights are $\LCP$ values, in the range
$[0..n-1]$.

Kopelowitz and Lewenstein \cite[Sec.~3.2]{KL07} show how this query can be 
solved in $O(r)$ space and the time of a predecessor query. Those predecessor
queries are done on universes of size $n$ where there can be arbitrarily few
elements. However, we can resort to binary search if there are $O(n/r)$
elements, within the allowed time $O(\log(n/r))$. Therefore, the predecessor
queries have to be implemented only on sets of $\Omega(n/r)$ elements.
By using the structure of Belazzougui and Navarro \cite[Thm.~14]{BN14}, the
predecessor time is $O(\log\log_w r)$. 
Therefore, we obtain time $t_{\SV'} = O(\log(n/r) + \log\log_w r)$.

This time can be reduced to $t_\SV = O(\log(n/r))$ for the simpler primitives
$\PSV/\NSV$ as follows: When $r$ is so large that $\log(n/r) < \log\log n$, 
that is, $r > n/\log n$, the allowed $\Theta(r\log(n/r) w)$ bits
of space are actually $\Omega(n\log\log n)$.
We are then entitled to use $O(n)$ bits of space, within which we can solve
queries $\PSV$ and $\NSV$ in $O(1)$ time \cite[Thm.~3]{FMN09}.

\subsection{Supporting operation $\mathit{Child}$}

Operation $\mathit{Child}(v,a)$ requires us to binary search the $O(\sigma)$
positions where the minimum occurs in $\LCP[v_l+1..v_r]$, and choose the
one that descends by letter $a$. Each check for $a$ takes $O(\log(n/r))$ time,
as explained.

To implement this operation efficiently, we will store for each nonterminal $X$
the number $n(X)$ of times $m(X)$ occurs inside the expansion of $X$. To do
the binary search on $\LCP[p..q]$ (with $p=v_l+1$ and $q=v_r$), we first compute
$\RMQ(p,q)$ as in the previous section, and then find the desired occurrence of
its relative version, $\mu = \LCP[\RMQ(p,q)]-\LCP[p-1]$, through 
$X_x \langle p-p'+1,l(X_x)\rangle$, $\DLCP'[x+1..y-1]$, and
$X_y \langle 1,l(X_y)+q-q' \rangle$. The values $p'$, $q'$, $x$, and $y$ are 
those we computed to obtain $\RMQ(p,q)$.

\paragraph{\bf Searching inside a nonterminal.}
To process $X_w \langle a,b \rangle$ we first determine how many 
occurrences of $\mu$ it contains. We start with a counter $c=0$ and scan again 
$Y_1, Y_2, \ldots, Y_k$. For each $Y_s$, if $d(Y_1)+\ldots+d(Y_{s-1})+m(Y_s) 
= \mu$, we add $c \leftarrow c+n(Y_s)$. To process $Y^{t'}$ in constant time
note that, if $\mu$ occurs in $Y^{t'}$, then it occurs only in the first copy 
of $Y$ if $d(Y)>0$, only in the last if $d(Y)<0$, and in every copy if 
$d(Y)=0$. Therefore, $n(Y^{t'}) = n(Y)$ if $d(Y) \not= 0$ and $t' \cdot n(Y)$
if $d(Y)=0$. 

After we compute $c$ in $O(\log(n/r))$ time,
we binary search the $c$ occurrences of $\mu$ in $X_w \langle a,b
\rangle$. For each of the $O(\log c)=O(\log\sigma)$ steps of this binary
search, we must find a specific occurrence of $\mu$, and then compute the 
corresponding letter to compare with $a$ and decide the direction of the search.
As said,
we can compute the corresponding letter in time $O(\log(n/r))$. We now show
how a specific occurrence of $\mu$ is found within the same time complexity.

\paragraph{\bf Finding a specific occurrence inside a nonterminal.}
Assume we want to find the $c'$th occurrence of $\mu$ in 
$X_x \langle p-p'+1,l(X_x)\rangle$. We traverse once again
$Y_1, Y_2, \ldots, Y_k$. For each $Y_s$, if $d(Y_1)+\ldots+d(Y_{s-1})+m(Y_s) 
= \mu$, we subtract $c' \leftarrow c'-n(Y_s)$. When the result is below 1, 
the occurrence is inside $Y_s$. We then add $l(Y_1)+\ldots+l(Y_{s-1})$ to
$p$, subtract $d(Y_1)+\ldots+d(Y_{s-1})$ from $\mu$, restore $c' \leftarrow
c'+n(Y_s)$, and recursively search 
for $\mu$ inside $Y_s$. 

Let $Y_s\rightarrow Y_1 Y_2$. If $m(Y_1) \not= \mu$, 
we continue on $Y_2$ with $p \leftarrow p+l(Y_1)$ and $\mu \leftarrow 
\mu-d(Y_1)$. If $m(Y_1) = \mu$ and $n(Y_1) \ge c'$, we continue on
$Y_1$. Otherwise, we continue on $Y_2$ with $p \leftarrow p+l(Y_1)$, 
$\mu \leftarrow \mu-d(Y_1)$ and $c' \leftarrow c'-n(Y_1)$. 

To process $Y^t$ in the quest for $c'$, we do as follows. If $d(Y)>0$, then
$\mu$ can only occur in the first copy of $Y$. Thus, if $m(Y) \not= \mu$,
we just skip $Y^t$ with $p \leftarrow p+ t\cdot l(Y)$ and 
$\mu \leftarrow \mu - t \cdot d(Y)$. If $m(Y)=\mu$, we
see if $n(Y) \ge c'$. If so, then we enter into $Y$; otherwise we skip
$Y^t$ with $p \leftarrow p+t\cdot l(Y)$, 
$\mu \leftarrow \mu - t \cdot d(Y)$ and $c' \leftarrow c'-n(Y)$.
The case where $d(Y) < 0$ is similar, except that when we enter into $Y$,
it is the last one of $Y^t$, and thus we set $p \leftarrow p+(t-1)\cdot l(Y)$
and $\mu \leftarrow \mu-(t-1)\cdot d(Y)$. 
Finally, if $d(Y)=0$, then the minimum of $Y$ appears many times.
If $m(Y) \not= \mu$, we skip $Y^t$ with $p \leftarrow p+t\cdot l(Y)$ and
$\mu \leftarrow \mu - t \cdot d(Y)$.
Otherwise, if $t \cdot n(Y) < c'$, we must also skip $Y^t$, updating $p$ and
$\mu$, and also $c' \leftarrow c' - t \cdot n(Y)$. Otherwise, we must enter into
the $t'$th occurrence of $Y$, where $t' = \lceil c'/n(Y)\rceil$, by continuing
on $Y$ with $p \leftarrow p+ (t'-1)\cdot l(Y)$, 
$\mu \leftarrow \mu - (t'-1)\cdot d(Y)$ and $c' \leftarrow c'-
(t'-1)\cdot n(Y)$.

Therefore, if the desired minimum is in $X_w \langle a,b \rangle$, we spot it
in $O(\log(n/r)\log\sigma)$ time.

\paragraph{\bf Searching the central area.}
If we do not find the desired letter inside 
$X_x \langle p-p'+1,l(X_x)\rangle$ or $X_y \langle 1,l(X_y)+q-q' \rangle$,
we must find it in $\DLCP'[x+1..y-1]$. Here we proceed differently.
From the computation of $\RMQ(p,q)$ we know if there are occurrences
of $\mu$ in $\DLCP'[x+1..y-1]$. If there are, then any minimum in this range
is an occurrence of $\mu$. We binary search those minima by using a 
representation that uses $O(r)$ bits on top of $M$ and 
finds an approximation to the median of the minima in constant time \cite{FH10}
(it might not be the median but its rank is a fraction between $1/16$ and
$15/16$ of the total). 
For each $\DLCP[z]$ that contains some occurrence of $\mu$, we obtain its 
leftmost position $L[z-1]+p(Y_z)$ and determine the associated letter, compare 
it with $a$ and determine if the binary
search on $\DLCP'[x+1..y-1]$ goes left or right. Since there are $O(\sigma)$
minima in $\DLCP'[x+1..y-1]$, the search also takes $O(\log(n/r)\log\sigma)$
time.

Once we have finally
determined that our letter must occur inside some $X_z$, we process it as
done on $X_w \langle a,b \rangle$ to determine the exact occurrence, if it
exists.


\section{Experimental results} \label{sec:experiments}

We implemented our simplest scheme, that is, Theorem~\ref{thm:locating} using
$O(r)$ space, and compared it with the state of the art.

\subsection{Implementation}

We implemented the simpler version described by Bannai et al.~\cite{BGI18} of
the structure of Theorem \ref{thm:locating} (with $s=1$) using the
\texttt{sdsl} library~\cite{gbmp2014sea}.\footnote{\texttt{https://github.com/simongog/sdsl-lite}} 
For the run-length FM-index, we used the implementation described by Prezza~\cite[Thm. 28]{Pre16} (suffix array sampling excluded), taking $(1+\epsilon)r(\log(n/r)+2) + r\log\sigma$ bits of space (lower-order terms omitted for readability) for any constant $\epsilon > 0$ fixed at construction time and supporting $O(\log(n/r)+\log\sigma)$-time LF mapping. 
In our implementation, we chose $\epsilon = 0.5$.
This structure employs Huffman-compressed wavelet trees (\texttt{sdsl}'s \texttt{wt\_huff}) to represent run heads, as in our experiments they turned out to be comparable in size and faster than Golynski et al.'s structure~\cite{golynski2006rank}, which is implemented in \texttt{sdsl}'s \texttt{wt\_gmr}.

Our \texttt{locate} machinery is implemented as follows. 
We store one gap-encoded bitvector \texttt{First$[1..n]$} marking with a bit set the text positions that are the first in their BWT run (note that \texttt{First$[i]$} refers to \emph{text} position $i$, not BWT position). 
\texttt{First} is implemented using \texttt{sdsl}'s \texttt{sd\_vector}, takes overall $r(\log(n/r)+2)$ bits of space (lower-order terms omitted), and answers queries in $O(\log(n/r))$ time. 
We also store a vector \texttt{FirstToRun$[1..r]$} such that text position $\mathtt{First.select_1(i)}$ belongs to the \texttt{FirstToRun$[i]$}-th BWT run. 
\texttt{FirstToRun} is a packed integer vector stored in $r\log r$ bits. 
Finally, we explicitly store $r$ suffix array samples in a vector  \texttt{Samples$[1..r]$}: \texttt{Samples$[p]$} is the text position corresponding to the last letter in the $p$-th BWT run. 
\texttt{Samples} is also a simple packed vector, stored in $r\log n$ bits of space. 

Let $\SA[sp..ep]$ be the range of our query pattern. 
The run-length FM-index and vector \texttt{Samples} are sufficient to find the range $[sp..ep]$ and locate $SA[ep]$ using the simplified toe-hold 
lemma \cite{BGI18}. 
Moreover, for $0<i<n$ it holds that $\phi(i) = \mathtt{Samples[FirstToRun[First.rank_1(\mathit{i})]-1]} + \Delta$, where $\Delta = i - \mathtt{First.predecessor}(i)$ (assuming for simplicity that the first text position is marked in \texttt{First}; the general case can easily be handled with few more operations).
Note that $\phi$ is evaluated in just $O(\log(n/r))$ time. Notably, this time drops to $O(1)$ in the average case, that is, when bits set in \texttt{First} are uniformly distributed. This is because \texttt{sdsl}'s \texttt{sd\_vector} breaks the bitvector into $r$ equal-sized  buckets and solves queries inside each bucket (which in the average case contains just $O(1)$ bits set).
Occurrences $\SA[ep-1], \SA[ep-2], \dots, \SA[sp]$ are then retrieved as $\phi^k(\SA[ep])$, for $k=1,\dots, ep-sp$.

Overall, our index takes at most $((1+\epsilon)\log(n/r) + 2\log n + \log \sigma + 4+2\epsilon)\,r$ bits of space for any constant $\epsilon>0$ (lower-order terms omitted for readability) and, after counting, locates each pattern occurrence in $O(\log(n/r))$ time. 
Note that the space of our index essentially coincides with the information-theoretic minimum
needed for storing the run-length BWT and $2r$ text positions in plain format (which is $r\log(n/r)+r\log\sigma + 2r\log n$ bits); therefore it is close to the optimum, since our locate strategy requires storing $2r$ text positions. In the following, we refer to our index as \texttt{r-index}; the code is publicly available%
\footnote{\texttt{https://github.com/nicolaprezza/r-index}}.

\subsection{Experimental Setup}

We compared \texttt{r-index} with the state-of-the-art index for each 
compressibility measure: \texttt{lzi}%
\footnote{\texttt{https://github.com/migumar2/uiHRDC} \label{migumar}} 
\cite{KN13,CFMPN16} ($z$), 
\texttt{slp}\footref{migumar} \cite{CNfi10,CFMPN16} ($g$), 
\texttt{rlcsa}\footnote{\texttt{https://github.com/adamnovak/rlcsa}}
\cite{MNSVrecomb09,MNSV09} ($r$), and
\texttt{cdawg}\footnote{\texttt{https://github.com/mathieuraffinot/locate-cdawg}}
\cite{BCGPR15} ($e$). 
We also included \texttt{hyb}%
\footnote{\texttt{https://github.com/hferrada/HydridSelfIndex}}
\cite{FGHP13,FKP18}, which combines a Lempel-Ziv index with an FM-index,
with parameter $M=8$, which is optimal for our experiment.
We tested \texttt{rlcsa} using three suffix array sample rates per dataset: the rate $X$ resulting in the same size for \texttt{rlcsa} and \texttt{r-index}, plus rates $X/2$ and $X/4$.

We measured memory usage and \texttt{locate} times per occurrence of all indexes on 1000 patterns of length 8 extracted from four repetitive datasets, which are also published with our implementation: 
\begin{description}
\item[\texttt{DNA:}] an artificial dataset of 629{,}145 copies of a DNA sequence of length 1000 (Human genome) where each character was mutated with probability $10^{-3}$;
\item[\texttt{boost:}] a dataset consisting of concatenated versions of the \texttt{GitHub}'s \texttt{boost} library;
\item[\texttt{einstein:}] a dataset consisting of concatenated versions of \texttt{Wikipedia}'s English \texttt{Einstein} page;
\item[\texttt{world\_leaders:}] a collection of all pdf files of CIA World Leaders from 2003 to 2009 downloaded from the \texttt{Pizza\&Chili} corpus.
\end{description}

Table~\ref{tab:datasets} shows the main characteristics of the datasets: the
length $n$, the alphabet size $\sigma$, the number of runs $r$ in their BWT, 
the number $z$ of LZ77 phrases%
\footnote{Using code requested to the authors of an efficient LZ77 parser 
\cite{KKP13}.}, 
and the size of $g$ the grammar generated by Repair%
\footnote{Using the ``balanced'' version offered at 
\texttt{http://www.dcc.uchile.cl/gnavarro/repair.tgz}}. 
Note the varying degrees of repetitiveness: \texttt{boost} is the most
repetitive dataset, followed by \texttt{DNA} and \texttt{einstein}, which are
similar, and followed by the least repetitive one, \texttt{world\_leaders}.
It can be seen that $g \ge z$ by a factor of 1.3--2.8 and $r \ge g$ by a factor
of 1.0--1.8. Therefore, we could expect in general that the indexes based on 
grammars or on Lempel-Ziv parsing are smaller than \texttt{r-index}, but as we
see soon, the differences are not that large.

\begin{table}[t]
\begin{center}
\begin{tabular}{l@{~}|@{~}r@{~}|@{~}r@{~}|@{~}r@{~}|@{~}r@{~}|@{~}r}
Dataset & $n$~~~~~~~ & $\sigma$\,\,\, & $r$~~~~~~~~~~~ & $z$~~~~~~~~~ & $g$~~~~~~~~~~ \\
\hline
\texttt{DNA}            & 629{,}140{,}006 &  10 & 1{,}287{,}508 (0.065) & 551{,}237 (0.028) & 727{,}671 (0.037) \\
\texttt{boost}          & 629{,}145{,}600 &  96 &      62{,}025 (0.003) &  22{,}747 (0.001) &  63{,}480 (0.003) \\
\texttt{einstein}       & 629{,}145{,}600 & 194 &     958{,}671 (0.049) & 292{,}117 (0.015) & 631{,}239 (0.032) \\
\texttt{world\_leaders} &  46{,}968{,}181 &  89 &     573{,}487 (0.391) & 175{,}740 (0.120) & 507{,}525 (0.346) \\
\end{tabular}
\end{center}
\caption{The main characteristics of our dataset. 
The numbers in parentheses are rough approximations to the 
bits/symbol achievable by the associated compressors by using one 4-byte
integer per run, phrase, or right-hand-side grammar symbol.}
\label{tab:datasets}
\end{table}

Memory usage (Resident Set Size, RSS) was measured using \texttt{/usr/bin/time} between index loading time and query time. This choice was motivated by the fact that, due to the datasets' high repetitiveness, the number $occ$ of pattern occurrences was very large. This impacts sharply  on the working space of indexes such as \texttt{lzi} and \texttt{slp}, which report the occurrences in a recursive fashion. When considering this extra space, these indexes always use more space than the \texttt{r-index}, but we prefer to emphasize the relation between the index sizes and their associated compressibility measure.
The only existing implementation of \texttt{cdawg} works only on DNA files, so we tested it only on the \texttt{DNA} dataset.

\subsection{Results}

Figure \ref{fig:results} summarizes the results of our experiments. On
all datasets, the time per occurrence of \texttt{r-index} is 100--300
nanoseconds per occurrence, outperforming all the indexes based on Lempel-Ziv
or grammars by a factor of 10 to 100. These indexes are generally smaller, 
using 45\%--95\% (\texttt{lzi}), 80\%--105\% (\texttt{slp}), and 45\%--100\% 
(\texttt{hyb}) of the space of \texttt{r-index}, at the expense of being orders
of magnitude slower, as said: 20--100 (\texttt{lzi}), 8--50 (\texttt{slp}), and
7--11 (\texttt{hyb}) times. Further, \texttt{r-index} dominates all practical 
space-time tradeoffs of \texttt{rlcsa}: using the same space, \texttt{rlcsa} 
is 20--500 times slower than \texttt{rindex}, and letting it use 1.7--4.4 times
the space of \texttt{r-index}, it is still 5--100 times slower. The regular 
sampling mechanism of the FM-index is then completely outperformed. Finally,
\texttt{cdawg} is almost twice as fast as \texttt{r-index}, but it is 60 
times larger (indeed, larger than a classical FM-index), which leaves it out
of the competition on ``small'' indexes. 

\begin{figure}[t]
\includegraphics[width=\textwidth]{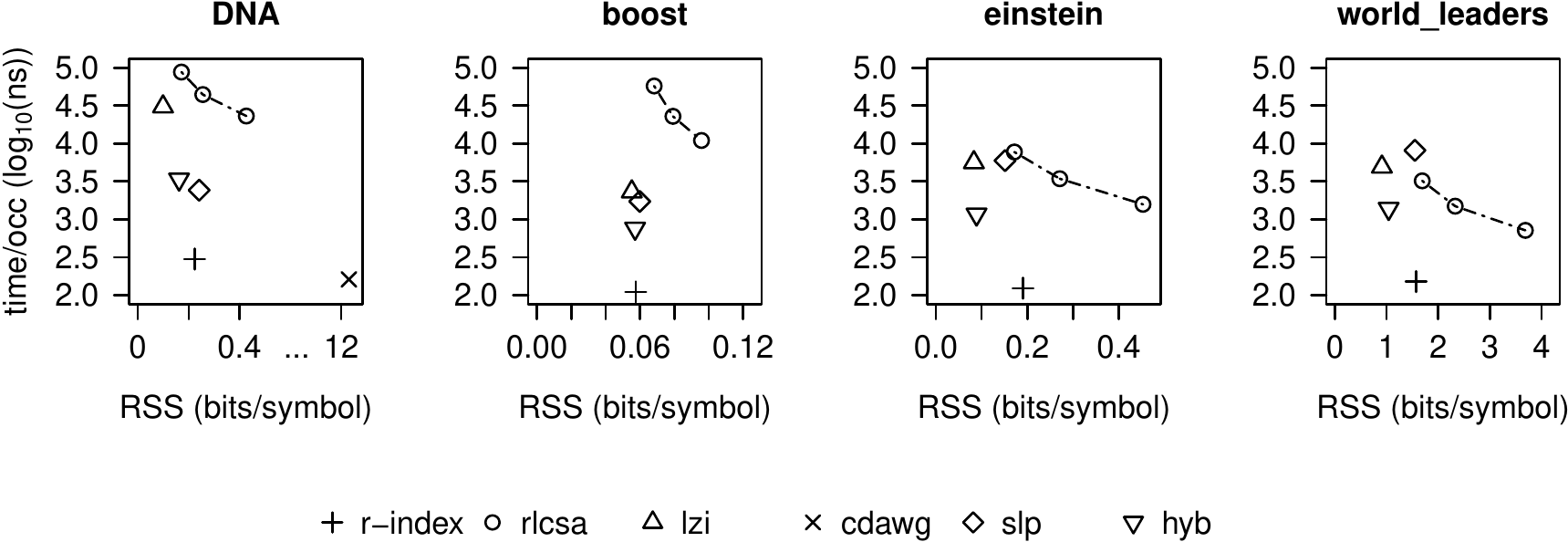}
\caption{Locate time per occurrence and working space (in bits per symbol) of the indexes. The $y$-scale measures nanoseconds per occurrence reported and is 
logarithmic.}
\label{fig:results}
\end{figure}

Comparing with the bits per symbol of Table~\ref{tab:datasets}, we note that
the space of \texttt{r-index} is 2--4 words per run, whereas \texttt{lzi} and
\texttt{hyb} use 3--6 words per Lempel-Ziv phrase and \texttt{slp} uses 4--6
words per symbol on the right-hand-side of a rule. The low space 
per run of \texttt{r-index} compared to the indexes based on $z$ or $g$ shrink
the space gap one could expect from comparing the measures $r$, $z$, and $g$.

\subsection{Scalability}

We finish with an experiment showing the space performance of the indexes on a 
real collection of Influenza nucleotide sequences from NCBI%
\footnote{\texttt{ftp://ftp.ncbi.nih.gov/genomes/INFLUENZA/influenza.fna.gz}, the
description is in the parent directory.}. 
It is formed by 641{,}444 sequences, of total size 0.95 GB after removing the
headers and newlines. We built the indexes on 100 prefixes of the dataset, 
whose sizes 
increased evenly from 1\% to 100\% of the sequences. As the prefixes grew, 
they became more repetitive; we measured how the bits per symbol used by the 
indexes decreased accordingly. As a repetition-insensitive variant, we also 
include a classical succinct FM-index ({\tt fm-index}), with a typical sampling
rate of $\lceil \lg n\rceil$ positions for locating, plain bitvectors for the
wavelet trees and for marking the sampled $\SA$ positions, and a $\rank$ 
implementation using 1.25 bits per input bit.

\begin{figure}[t]
\begin{center}
\includegraphics[width=0.7\textwidth]{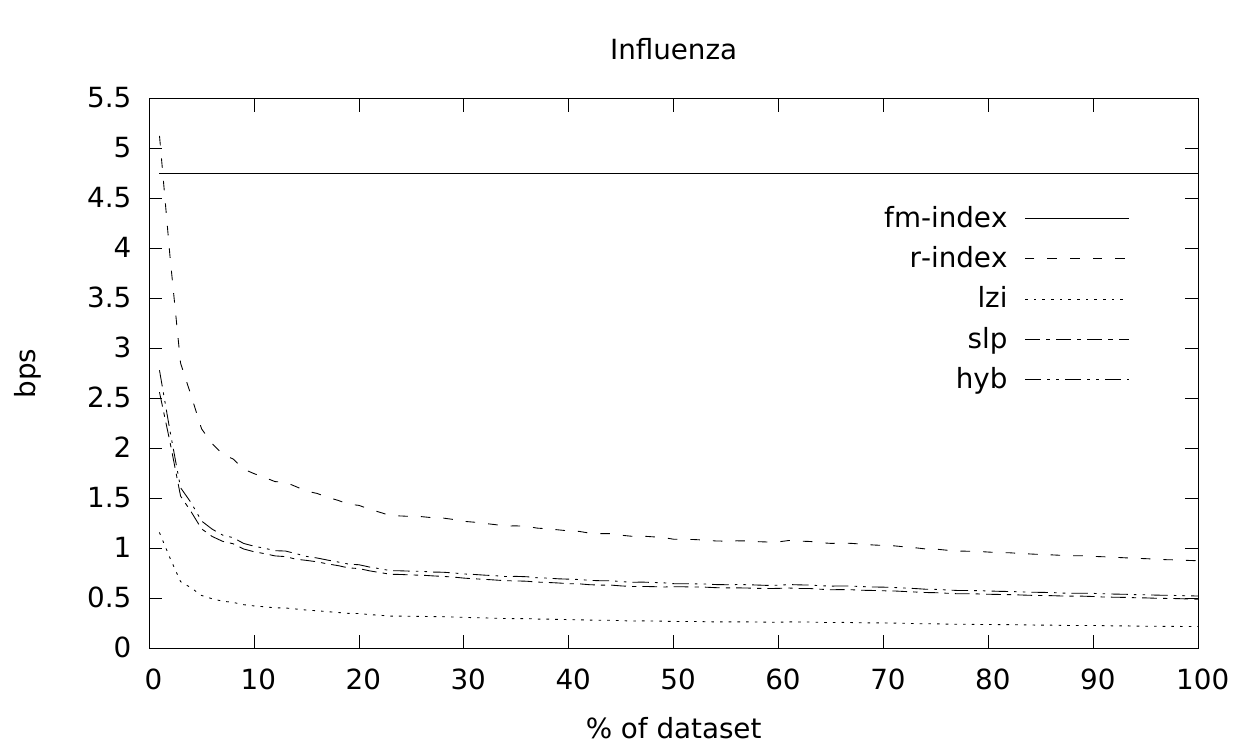}
\end{center}
\vspace*{-5mm}
\caption{Index sizes (in bits per symbol, bps) for increasing prefixes of a 
repetitive collection of genomic data.}
\label{fig:influenza}
\end{figure}

Figure~\ref{fig:influenza} shows the evolution of the index sizes. As we add
more and more similar sequences, all the indexes (except the FM-index) decrease
in relative size (bps), as expected. On the complete collection, {\tt fm-index} 
still uses 4.75 bits per symbol (bps), whereas {\tt r-index} has decreased to 
0.88 bps (about 2.4 words per run), {\tt hyb} to 0.52 bps (about 5.5 words per 
phrase, 60\% of {\tt r-index}), {\tt slp} to 0.49 bps (about 1.9 words per 
symbol, 56\% of {\tt r-index}), and {\tt lzi} to 0.22 bps (about 
2.3 words per phrase, 25\% of {\tt r-index}). We remind that, in exchange, 
{\tt r-index} is 10--100 times faster than those indexes, and that it uses 
18\% of the space of the classic {\tt fm-index} (a factor that decreases
as the collection grows).

This collection, where the repetitiveness is not as high as in the previous 
datasets (in fact, it is close to that of {\tt world\_leaders}), shows that 
$r$ (and thus {\tt r-index}) is more sensitive than $g$ and $z$ to the decrease
in repetitiveness. In particular, $g$ and $z$ are always $O(n/\log_\sigma n)$,
and thus the related indexes always use $O(n\log\sigma)$ bits. Instead, $r$ can
be as large as $n$ \cite{Pre16}, so in the worst case {\tt r-index} can use 
$\Theta(n\log n)$ bits. Note, in particular, that the other indexes are below 
the 2 bps of the raw data after processing just 3\% of the collection; {\tt
r-index} breaks this barrier only after 8\%.


\section{\textcolor{red}{Construction}} \label{sec:construction}

In this section we analyze the working space and time required to build all 
our data structures. Table~\ref{tab:constr} summarizes the results. The
working space does not count the space needed to read the text in online form,
right-to-left. Times are worst-case unless otherwise stated. Expected cases
hold with high probability (w.h.p.), which means over $1-1/n^c$ for any fixed
constant $c$.

\begin{table}[t]
\begin{tabular}{l@{~}|@{~}l@{~}|@{~}l}
Structure & Construction time & Construction space \\
\hline
Basic counting and locating (Lem.~\ref{lem:rlfm}) 
	& $O(n\log r)$ 	& $O(r)$ \\
~~~~~ or	& $O(n)$ 	& $O(n)$ \\
\hline
Fast locating (Thm.~\ref{thm:locating} $+$ Lem.~\ref{lemma: general locate}, $s=\log\log_w(n/r)$) 
	& $O(n(\log r + \log\log_w(n/r)))$ 	& $O(r\log\log_w(n/r))$  \\
~~~~~ or	& $O(n)$ 				& $O(n)$ \\
\hline
Optimal counting and locating (Thm.~\ref{thm:optimal}) 
	& $O(n(\log r + \log\log_w(n/r)))$ 	& $O(r\log\log_w(\sigma+n/r))$  \\
~~~~~ or	& $O(n+r(\log\log \sigma)^3)$ 	& $O(n)$ \\
\hline
RAM-optimal counting and locating (Thm.~\ref{thm:optimal packed}) 
	& $O(n(\log r + \log\log_w(n/r)))$ & $O(rw\log_\sigma\log_w(\sigma+n/r))$  \\
	& $~~~~+ O(rw^{1+\epsilon})$ exp. & \\
~~~~~ or 	& $O(n+rw^{1+\epsilon})$ exp.	& $O(n)$ \\
\hline
Text substrings (Thm.~\ref{thm:extract}) 
	& $O(n(\log r + \log\log_w(n/r)))$ 	& $O(r\log(n/r))$ \\
~~~~~ or	& $O(n)$ 				& $O(n)$ \\
\hline
Accessing $\SA$, $\ISA$, and $\LCP$ (Thm.~\ref{thm:dsa}, \ref{thm:disa} \& \ref{thm:dlcp}) 
	& $O(n(\log r+\log\log_w(n/r)))$ 	& $O(r\log(n/r))$ \\
~~~~~ or	& $O(n)$ 	& $O(n)$ \\
\hline
Optimal counting/locating, $O(r\log(n/r))$ space (Thm.~\ref{thm:optimal 2}) 
	& $O(n+r(\log\log\sigma)^2)$ exp.	& $O(r\log(n/r)\log n/\log\log n)$  \\
\hline
Suffix tree (Thm.~\ref{thm:stree}) 
				& $O(n+r\log\log_w r)$	& $O(n)$ \\
~~~~~ without operation {\em LAQ}$_S$ 	& $O(n)$ 		& $O(n)$ \\
~~~~~ with $O(Sort(n))$ I/Os &
			        $O(n(\log r+\log\log_w(n/r)))$        &
$O(B+r\log(n/r))$ \\
~~~~~ with $O(n/B+\log(n/r))$ I/Os, no {\em TDepth} \& {\em LAQ}$_T$ &
			        $O(n(\log r+\log\log_w(n/r)))$        &
$O(B+r\log(n/r))$ \\
\hline
\end{tabular}
\vspace{5mm}
\caption{Construction time and space for our different data structures, for
any constant $\epsilon>0$. All the expected times (``exp.'') hold w.h.p.\ as 
well. Variable $B$ is the block size in the external memory model, where
$Sort(n)$ denotes the I/O complexity of sorting $n$ integers.}
\label{tab:constr}
\end{table}

\subsection{Dictionaries and predecessor structures}
\label{sec:dictconstr}

A dictionary mapping $t$ keys from a universe of size $u$ to an interval
$[1..O(t)]$ can be 
implemented as a perfect hash function using $O(t)$ space and searching
in constant worst-case time. Such a function can be built in $O(t)$ space
and expected time \cite{FKS84}. A construction that takes $O(t)$ time w.h.p.\ 
\cite{Wil00} starts with a distributor hash function that maps the keys to an 
array of buckets $B[1..t]$. Since the largest bucket contains 
$O(\log t/\log\log t)$ keys w.h.p., we can build a fusion tree \cite{FW93} 
on each bucket, which requires linear space and construction time, and 
constant query time.

If we are interested in deterministic 
construction time, we can resort to the so-called deterministic 
dictionaries, which use $O(t)$ space and can be built in time 
$O(t(\log\log t)^2)$ \cite{Ruz08}.

A {\em minimum perfect hash function (mphf)} maps the keys to the range 
$[1..t]$. This is trivial using $O(t)$ space (we just store the mapped value),
but it is also possible to store a mphf within $O(t)$ bits, building it in 
$O(t)$ expected time and $O(t)$ space \cite{BBD09}. Such expected time holds 
w.h.p.\ as well if they use a distributor function towards $t'=O(t/\log t)$
buckets. For each bucket $B_i$, $i \in [1..t']$, they show that w.h.p.\ 
$O(\log t)$ trials are sufficient to find a perfect hash function $\sigma(i)$ 
for $B_i$, adding up to $O(t)$ time w.h.p. Further, the indexes $\sigma(i)$ 
found distribute geometrically (say, with a constant parameter $p$), 
and the construction also fails if their sum exceeds $\lambda\cdot t/p$ for 
some constant $\lambda$ of our choice. The probability of that event is
exponentially decreasing with $t'$ for any $\lambda>1$ \cite{Jan17}.

A {\em monotone mphf (mmphf)}, in addition, preserves the order of the keys.
A mmphf can be stored in $O(t\log\log u)$ bits while answering in constant 
time. Its construction time and space is as for a mphf
 \cite[Sec.~3]{BBPV09} (see also \cite[Sec.~3]{BBPV11}).
Therefore, all the expected cases we mention related to building perfect hash
functions of any sort hold w.h.p.\ as well. Alternatively, their construction
time can turn into worst-case w.h.p.\ of being correct. 

Our predecessor structure \cite[Thm.~14]{BN14} requires $O(t)$ words and
answers in time $O(\log\log_w (u/t))$. Its reduced-space version
\cite[Sec.~A.1 and A.2]{BN14}, using $O(t\log(u/t))$ bits, does 
not use hashing. It is a structure of $O(\log\log_w(u/t))$ layers, each 
containing a bitvector of $O(t)$ bits. Its total worst-case construction time 
is $O(t\log\log_w (u/t))$, and requires $O(t)$ space.

Finally, note that if we can use $O(u)$ bits, then we can build a constant-time
predecessor structure in $O(u)$ time, by means of rank queries on a bitvector.

\subsection{Our basic structure} \label{sec:constrbasic}

The basic structures of Section~\ref{sec:rlfmi} can be built in $O(r)$ 
space. We start by using an $O(r)$-space construction of the run-length encoded
$\BWT$ that scans $T$ once, right to left, in $O(n\log r)$ time \cite{Pre16}
(see also Ohno et al.~\cite{OTIS17} and Kempa \cite{Kem19}). The text $T$ is 
not needed anymore from now on.

We then build the predecessor structure $E$ that enables the $\LF$-steps in 
time $O(\log\log_w(n/r))$. The construction takes $O(r\log\log_w(n/r))$ time
and $O(r)$ space. The positions $p$ that start or end $\BWT$ runs are easily
collected in $O(r)$ time from the run-length encoded $\BWT$. 

The structures to compute $\rank$ on $L'$ in time $O(\log\log_w\sigma)$ 
\cite{BN14} also use predecessor structures. These are organized in $r/\sigma$ 
chunks of size $\sigma$. Each chunk has $\sigma$ lists of positions in 
$[1..\sigma]$ of lengths $\ell_1,\ldots,\ell_\sigma$, which add up to $\sigma$. 
The predecessor structure for the $i$th list is then built over 
a sample of $\ell_i/\log_w \sigma$ elements, in time
$O((\ell_i/\log_w \sigma)\log\log_w(\sigma\log_w(\sigma)/\ell_i))$. Adding over 
all the lists, we obtain $O((\sigma/\log_w\sigma)\log\log_w\sigma) 
\subseteq O(\sigma)$. The total construction time of this structure is 
then $O(r)$.

In total, the basic structures can be built in $O(n\log r)$ time and $O(r)$
space. Of course, if we can use $O(n)$ construction space, then we easily obtain
$O(n)$ construction time, by building the suffix array in linear time and then 
computing the structures from it. In this case the predecessor structure $E$ is 
implemented as a bitvector, as explained, and $\LF$ operates in constant time.

\subsection{Fast locating} \label{sec:constrlocate}

Structure $E$ collects the starts of runs. In Section~\ref{sec:locate}
we build two extended versions that collect starts and ends of runs. The
first is a predecessor structure $R$ (Lemma~\ref{lem:find_one}), which 
organizes the $O(r)$ run starts and ends separated by their character, on a 
universe of size $\sigma n$. The second uses two predecessor structures
(Lemmas~\ref{lem:find_neighbours} and \ref{lemma: general locate}), called
$P^+$ and $P^-$ in Lemma~\ref{lemma: general locate}, which contain the $\BWT$ 
positions at distance at most $s$ from run borders.

To build both structures, we simulate a backward traversal of $T$ (using 
$\LF$-steps from the position of the symbol \$) to collect the text positions of
all the run starts and ends (for $R$), or all the elements at distance at most
$s$ from a run start or end (for $P^+$ and $P^-$). We use predecessor and
successor queries on $E$ (the latter are implemented without increasing the 
space of the predecessor structure) and accesses to $L'$ to 
determine whether the current text position must be stored, and where. The 
traversal alone takes time $O(n\log\log_w(n/r))$ for the $\LF$-steps.

The predecessor structure $R$ is built in $O(r)$ space and 
$O(r\log\log(\sigma n/r)) \subseteq O(n \log\log\sigma) \subseteq
O(n\log r)$ (since $\sigma \le r$) time. 
The structures $P^+$ and $P^-$ contain $O(rs)$ elements in a universe of
size $n$, and thus are built in $O(rs)$ space and time
$O(rs\log\log(n/(rs))) \subseteq O(n)$ (we index up to $rs$ elements but
never more than $n$).

Overall, the structure of Theorem~\ref{thm:locating}, enhanced as in
Lemma~\ref{lemma: general locate}, can be built in $O(rs)$ 
space and $O(n\log r + n\log\log_w(n/r))$ time. If we can use
$O(n)$ space for the construction, then the $\LF$-steps can be implemented in 
constant time and the traversal requires $O(n)$ time. In this case, the
structures $R$, $P^+$ and $P^-$ can also be built in $O(n)$ time, since the 
predecessor searches can be implemented with bitvectors. 

For the structure of Lemma~\ref{lemma: lcp} we follow the same procedure,
building the structures $P^+$ and $P^-$. The classical algorithm to build the
base $\LCP$ array \cite{KLAAP01} uses $O(n)$ time and space. Within this
space we can also build the predecessor structures in $O(n)$ time, as before.
Note that this structure is not needed for Theorem~\ref{thm:locating}, but in 
later structures. Using those, we will obtain a construction of $\LCP$ using 
less space (see Section~\ref{sec:constrsa}).

\subsection{Optimal counting and locating} \label{sec:constropt}

The first step of this construction is to build the compact trie that contains 
all the distinct substrings of length $s$ of $T$. All these lie around sampled 
text positions, so we can simulate a backward traversal of $T$ using $E$ and 
$L'$, as before, while maintaining a window of the last $s$ symbols seen. 
Whenever we hit a run start or end in $L$, we collect the next $s-1$ 
symbols as well, forming a substring of length $2s-1$, and from there we restart
the process, remembering the last $s$ symbols seen.\footnote{If we hit other
run starts or ends when collecting the $s-1$ additional symbols, we form a 
single longer text area including both text samples; we omit the details.}
This traversal costs $O(n\log\log_w(n/r))$ time as before.

The memory area where the edges of the compact trie will point is simply the
concatenation of all the areas of length $2s-1$ we obtained. We now collect the
$s$ substrings of length $s$ from each of these areas, and radix-sort the 
$O(rs)$ resulting strings of length $s$, in time $O(rs^2)$. After the strings
are sorted, if we remove duplicates (getting $\sigma^*$ distinct strings) and 
compute the longest common prefix of the consecutive strings, we easily build 
the compact trie structure in a single $O(\sigma^*)$-time pass. We then assign 
consecutive mapped values to the $\sigma^*$ leaves and also assign the values
$v_{\min}$ and $v_{\max}$ to the internal nodes. By recalling the suffix array
and text positions each string comes from, we can also assign the values $p$
and $\SA[p]$ (or $\SA^*[p]$) to the trie nodes.

To finish, we must create the perfect hash functions on the children of each 
trie node. There are $O(rs)$ children in total but each set stores at most 
$\sigma$ children, so the total deterministic time to create the dictionaries 
is $O(rs (\log\log \sigma)^2)$.
In total, we create the compact trie in time $O(n\log\log_w(n/r) + rs^2 +
rs (\log\log \sigma)^2)$ and space $O(rs)$.

The construction of the RLFM-index of $T^*$ can still be done within this 
space, without explicitly generating $T^*$, as follows. For each position 
$L[i]$, the $\BWT$ of $T$, we perform $s$ $\LF$-steps to obtain the metasymbol 
corresponding to $L^*[i]$, which we use to traverse the compact trie in order 
to find the mapped symbol $L^*[i]$. Since the values of $L^*$ are obtained in
increasing order, we can easily compress its runs on the fly, in $O(rs)$ space.
The $\BWT$ of $T^*$ is then obtained in time $O(ns\log\log_w (n/r))$.
We can improve the time by obtaining this $\BWT$ run by run instead of
symbol by symbol: We start from each run $L[x_1,y_1]$ and compute $x_2 = 
\LF(x_1)$. From it, we find the end $y_2$ 
of the run $x_2$ belongs in $L$. The new run is $L[x_2,y_2]$. We repeat the 
process $s$ times until obtaining $L[x_s,y_s]$. The next run of $L^*$ is then
$L^*[x_1,x_1+(y_s-x_s)]$. The computation of $y_k$ from $x_k$ can be done by
finding the predecessor of $x_k$ in $E$ and associating with each element in 
$E$ the length of the run it heads, which is known when building $E$. In this 
way, the cost to compute the $\BWT$ of $T^*$ decreases to 
$O(r^* \log\log_w(n/r)) \subseteq O(n\log\log_w(n/r))$.

From the $\BWT$, the other structures of the RLFM-index of $T^*$ are built as in
Section~\ref{sec:constrbasic}, in time $O(r^* \log\log_w(n/r^*)) \subseteq
O(n)$ and space $O(r^*) \subseteq O(rs)$. The array $C^*$ is also built in 
$O(r^*) \subseteq O(rs)$ time and $O(\sigma^*) \subseteq O(rs)$ space.

To finish, we need to build the structure of Lemma~\ref{lemma: general locate},
which as seen in Section~\ref{sec:constrlocate} is built in $O(rs)$ space and 
$O(n\log\log_w(n/r))$ time. With $s=\log\log_w(\sigma+n/r)$, the total 
construction time is upper bounded
by $O(n(\log\log_w(n/r)+(\log\log\sigma)^2))$ and the construction space by 
$O(rs)$ (we
limit $rs$ by $n$ because we never have more than $n$ symbols in the trie or
runs in $L^*$). When added to the $O(n\log r)$ time to build the $\BWT$ of $T$,
the total simplifies to $O(n(\log r + \log\log_w(n/r)))$ because $\sigma 
\le r$.

If we can use $O(n)$ space for the construction, then the $\LF$-steps can be 
implemented in constant time. We can generate $T^*$ explicitly and use 
linear-time and linear-space suffix array construction algorithms, so all the
structures are built in time $O(n)$. The compact trie can be built by pruning
at depth $s$ the suffix tree of $T$, which is built in $O(n)$ time. We still
need to build the perfect hash functions for the children, in deterministic time
$O(rs (\log\log\sigma)^2)$. When added to $O(n)$, the total simplifies to
$O(n+r(\log\log\sigma)^3)$.

The difference when building the RAM-optimal version is that the compact trie
must be changed by the structure of Navarro and Nekrich \cite[Sec.~2]{NN17}.
In their structure, they jump by $\log_\sigma n$ symbols, whereas we jump by
$w/\log\sigma$ symbols. Their perfect hash functions, involving $O(rs)$ nodes,
can be built in time $O(rs(\log\log(rs))^2)$, whereas their weak prefix search
structures \cite[Thm.~6]{BBPV18} are built in expected time $O(rs w^\epsilon)$
for any constant $\epsilon>0$. For the value of $s$ used in this case, the
time can be written as $O(r w^{1+\epsilon})$. The construction space stays in
$O(rs)$.

\subsection{Access to the text}

The structure of Theorem~\ref{thm:extract} can be built as follows. We 
sample the text positions of starts and ends of $\BWT$ runs. Each sampled
position induces a constant number of half-blocks at each of the $O(\log(n/r))$
levels (there are also $r$ blocks of level 0). For each block or half-block, 
we must find its primary occurrence. We first find all their rightmost $\BWT$
positions with an $\LF$-guided scan of $T$ of time $O(n\log\log_w(n/r))$,
after which we can read each block or half-block backwards in 
$O(\log\log_w(n/r))$ time per symbol. For each of them, we follow the method 
described in Lemma~\ref{lemma:primocc} to find its primary 
occurrence in $O(\log\log_w(\sigma+n/r))$ time per symbol, doing the backward
search as we extract its symbols backwards too. Since at level $l$ there are
$O(r)$ blocks or half-blocks of length $O(n/(r \cdot 2^{l-1}))$, the total
length of all the blocks and half-blocks adds up to $O(n)$, and the total time
to find the primary occurrences is $O(n\log\log_w(\sigma+n/r))$.

We also need to fill in the text at the leaves of the structure. This can
be done with an additional traversal of the $\BWT$, filling in the $T$
values (read from $L'$) 
at the required positions whenever we reach them in the traversal. 
The extra time for this operation is $O(n\log\log_w(n/r))$ (we use predecessor
and successor queries on $E$ to
determine when our $\BWT$ position is close enough to a sample so that the
current symbol of $L'$ must be recorded in the leaf associated with the 
sample).

Therefore, the structure of Theorem~\ref{thm:extract} is built in
$O(n\log\log_w(\sigma+n/r))$ time and $O(r\log(n/r))$ working space, once
the basic structure of Lemma~\ref{lem:rlfm} is built. 

In case of having $O(n)$ space for construction, we can replace predecessor
structures with $\rank$ queries on bitvectors of $n$ bits, but we still have
the $O(\log\log_w \sigma)$ time for $\rank$ on $L'$. Thus the total time is
$O(n\log\log_w \sigma)$. Although this is the most intuitive construction, we
will slightly improve it in Section~\ref{sec:constrsa}.

\subsection{Suffix array access and byproducts} \label{sec:constrsa}

The other structures of Section~\ref{sec:sa} give access to cells of the suffix array
($\SA$), its inverse ($\ISA$), and the longest common prefix array ($\LCP$).

The structure of Theorem~\ref{thm:dsa} is analogous to that of
Theorem~\ref{thm:extract}: it has $O(\log(n/r))$ levels and $O(r)$ blocks or
half-blocks of length $s_l = n/(r\cdot 2^{l-1})$ at each level $l$. However,
its domain is the suffix array cells and the way to find a primary occurrence 
of each block is different. At each level,
we start with any interval of length $s_l$ and compute $\LF$ on its left 
extreme. This leads to another interval of length $s_l$. We repeat the process 
until completing the cycle and returning to the initial interval. Along the way,
we collect all the intervals that correspond to blocks or half-blocks of this
level. Each time the current interval contains or immediately
follows a sampled $\BWT$ position in $E$, we make it the primary occurrence of 
all the blocks or half-blocks collected so far (all those must coincide with
the content of the current block or half-block), and reinitialize an empty set
of collected blocks. This process takes $O(n\log\log_w(n/r))$ time for a fixed
level. We can perform a single traversal for all the levels simultaneously, 
storing all the blocks in a dictionary using the left extreme as their search 
key. As we traverse the $\BWT$, we collect the blocks of all lengths starting
at the current position $p$. Further, we find the successor of $p-1$ in $E$ to 
determine the minimum length of the blocks that cover or follow the nearest 
sampled position, and all the sufficiently long collected blocks find their
primary occurrence starting at $p$. The queries on $E$ also amount to
$O(n\log\log_w(n/r))$ time.

This multi-level process requires a dictionary of all the $O(r\log(n/r))$ 
blocks and half-blocks. If we implement it as a predecessor structure, it takes
$O(r\log(n/r))$ space, it is constructed in $O(r\log(n/r)\log\log_w(n/r))$ 
time, and answers the $O(n)$ queries in time $O(n\log\log_w(n/r))$. The 
collected segments can be stored separated by length, and the $O(\log(n/r))$
active lengths be marked in a small bitvector, where we find the nonempty
sets over some length in constant time.

We also need to fill the $\DSA$ cells of the leaves of the structure. This can
be done with an additional traversal of the $\BWT$, filling in the $\SA$
values at the required positions whenever we reach them in the traversal. We
can then easily convert $\SA$ to $\DSA$ values in the leaves. This does not
add extra time or space, asymptotically.

The construction of the structures of Theorem~\ref{thm:disa} is analogous. This
time, the domain of the blocks and half-blocks are the text positions and,
instead of traversing with $\LF$, we must use $\phi$. This corresponds to
traversing the $\BWT$ right to left, keeping track of the
corresponding position in $T$. We can maintain the text position using
our basic structure of Lemma~\ref{lem:find_neighbours}. Then, if the current
text position is $i$, we can use the predecessor structures on $T$ to find the
first sampled position following $i-1$, to determine which collected blocks
have found their primary occurrence. We can similarly fill the
required values $\DISA$ by traversing the $\BWT$ right-to-left and writing
the appropriate $\ISA$ values. Therefore, we can build the structures within 
the same cost as Theorem~\ref{thm:dsa}.

In both cases, if we have $O(n)$ space available for construction, we can
build the structures in $O(n)$ time, since $\LF$ can be computed in constant
time and all the dictionaries and predecessor structures can be implemented 
with bitvectors.
We can also use these ideas to obtain a slightly faster construction for
the structures of Theorem~\ref{thm:extract}, which extract substrings of $T$.

\begin{lemma} \label{lem:T}
Let $T[i-1..i]$ be within a phrase. Then it holds that
$\phi(i-1)=\phi(i)-1$ and $T[i-1]=T[\phi(i)-1]$.
\end{lemma}
\begin{proof}
The fact that $\phi(i-1)=\phi(i)-1$ is already proved in Lemma~\ref{lem:disa}.
From that proof it also follows that $T[i-1] = BWT[p] = BWT[p-1] = T[j-1] =
T[\phi(i)-1]$.
\qed
\end{proof}

\begin{lemma} \label{lem:T2}
Let $T[i-1..i+s]$ be within a phrase, for some $1 < i \le n$ and
$0 \le s \le n-i$. Then there exists $j \not= i$ such that $T[j-1..j+s-1] =
T[i-1..i+s-1]$ and $[j-1..j+s]$ contains the first position of a phrase.
\end{lemma}
\begin{proof}
The proof is analogous to that of Lemma~\ref{lem:disa2}.
By Lemma~\ref{lem:T}, it holds that $T[i'-1..i'+s-1] = T[i-1..i+s-1]$, where
$i' = \phi(i)$. If $T[i'-1..i'+s]$ contains the first position of a
phrase, we
are done. Otherwise, we apply Lemma~\ref{lem:T} again on $[i'-1..i'+s]$, and
repeat until we find a range that contains the first position of a phrase.
This search eventually terminates because $\phi$ is a permutation with a 
single cycle.
\qed
\end{proof}

We can then find the primary occurrences for all the blocks in 
Theorem~\ref{thm:extract} analogously as for
$\DISA$ (Theorem~\ref{thm:disa}). We traverse $T$ with $\phi$ (i.e., we 
traverse the $\BWT$ right to left, using Lemma~\ref{lem:find_neighbours} to
compute $\phi$ each time). This time we index the blocks and half-blocks using
their right extreme, collecting all those that end at the current position $i$
of $T$. Then, at each position $i$, we use the predecessor structures on $T$ 
to find the nearest sampled position preceding $i+1$, to determine which 
collected blocks and half-blocks
have found their primary occurrence. We can similarly fill the
required values of $T$ with a final traversal of $\BWT$, accessing $L'$.
Therefore, we can build these structures within the same cost of 
Theorem~\ref{thm:disa}.

Finally, the construction for $\LCP$ access in Theorem~\ref{thm:dlcp} is a 
direct combination of Theorem~\ref{thm:dsa} (i.e., $\SA$) and 
Lemma~\ref{lemma: lcp} (i.e., {\em PLCP} extension, with $s=\log(n/r)$). In
Section~\ref{sec:constrlocate} we saw how to build the latter in $O(n)$ time
and space. Within $O(n)$ space, we can also build the structure of
Theorem~\ref{thm:dlcp} in $O(n)$ time. We can, however, build the structure 
of Lemma~\ref{lemma: lcp} within $O(r\log(n/r)+rs)$ space if we first build 
$\SA$, $\ISA$, and the extraction structure. The classical linear-time 
algorithm \cite{KLAAP01} is
as follows: we compare $T[\SA[2]..]$ with $T[\SA[1]..]$ until they differ;
the number $\ell$ of matching symbols is $\LCP[2]$. Now we jump to compute
$\LCP[\Psi(2)]$, where $\Psi(p)=\ISA[(\SA[p]\!\!\mod n)+1]$ is the inverse of 
$\LF$
\cite{GV06}. Note that $\LCP[\Psi(2)]=lcp(T[\SA[\Psi(2)]..],T[\SA[\Psi(2)-1]..])
=lcp(T[\SA[2]+1..],T[\SA[\Psi(2)-1]..])$ and, if $\ell>0$, this is at least 
$\ell-1$ because $T[\SA[2]+1..]$ already shares the first $\ell-1$ symbols 
with some lexicographically smaller suffix,
$T[\SA[1]+1..]$. Thus the comparison starts from the position $\ell$ onwards:
$\LCP[\Psi(2)]=\ell-1+lcp(T[\SA[\Psi(2)]+\ell-1..],T[\SA[\Psi(2)-1]+\ell-1..])$.
This process continues until the cycle $\Psi$ visits all the positions of 
$\LCP$.

We can simulate this algorithm, traversing the whole virtual array $\LCP[1..n]$ 
but writing only the $O(rs)$ cells that are at distance $s$ from a run border. 
We first build $P^+$, $P^-$, and $\LCP'$ as for Lemma~\ref{lemma: lcp}. We then
traverse $T$ backwards virtually, using $\LF$, in time $O(n\log\log_w(n/r))$, 
spotting the positions in $P^\pm = P^+ \cup P^-$. Say we find $p \in P^\pm$ and
the previous $p' \in P^\pm$ was found $d$ steps ago. This means that 
$p' = \Psi^d(p)$ 
is the next relevant suffix after $p$ along the $\LCP$ algorithm. We store 
$next[f(p)] = \langle p',d\rangle$, where $next$ is a table aligned with
$\LCP'$. Once this pass is complete, we simulate the algorithm starting at the 
last relevant $p$ value we found: we compute $\LCP[p] = \ell$ and store 
$\LCP'[f(p)]=\ell$. Then, if $next[f(p)] = \langle p',d\rangle$, we set
$p = p'$ and $\ell = \max(1,\ell-d)$. Along the process, we carry out $O(rs)$ 
string comparisons for a total of $O(n)$ symbols. Each string comparison takes
time $O(\log(n/r))$ in order to compute $\ISA$. We extract the desired 
substrings of $T$ by chunks of $\log(n/r)$ symbols, so that comparing $\ell$ 
symbols costs $O(\ell + \log(n/r))$. Overall, the traversal takes time 
$O(n+rs\log(n/r))$,
plus the $O(n\log\log_w(n/r))$ time to compute $next$. Added to the 
$O(n\log r + n\log\log_w(n/r))$ time needed in Section~\ref{sec:constrlocate} 
to build the sampling structures, we have a total time of 
$O(n\log r + n\log\log_w(n/r) + rs\log(n/r))$, within $O(r\log(n/r)+rs)$ space. 
For $s=\log(n/r)$, as required in Theorem~\ref{thm:dlcp}, the space is
$O(r\log(n/r))$ and the time is in $O(n\log r + n\log\log_w(n/r))$, because
$O(rs\log(n/r)) = O(r\log^2(n/r)) \subseteq O(n)$.

Finally, to obtain optimal counting and locating in space $O(r\log(n/r))$, we
only need to care about the case $r \ge n/\log n$, so the allowed space 
becomes $\Omega(n\log\log n)$ bits. In this case we use an $O(n)$-bit compressed
suffix tree enriched with the structures of Belazzougui and Navarro 
\cite[Lem.~6]{BN13}. This requires, essentially, the suffix tree topology 
represented with parentheses, edge lengths (capped to $O(\log\log n)$ bits), 
and mmphfs on the first letters of the edges towards the nodes' children. The 
parentheses and edge lengths are obtained directly left-to-right, with a 
sequential pass over $\LCP$ \cite{KLAAP01,Sad07}. If we have $O(n\log n)$ bits
for the construction (which can be written as $O(r\log(n/r)\log n/\log\log n)$
space), the first letters are obtained directly from the suffix array
and the text, all in $O(n)$ time. The construction of the mmphfs on (overall) 
$O(n)$ elements can be done in $O(n)$ expected time. We need, in addition, the
structures to extract substrings of $T$ and entries of $\SA$, and a compact
trie on the distinct strings of length $\log(n/r)$ in $T$. With $O(n)$ space,
the other structures of Section~\ref{sec:constropt} are built in
$O(n+r(\log\log\sigma)^2)$ expected time.

\subsection{Suffix tree}
\label{sec:streeconstr}

The suffix tree needs the compressed representations of $\SA$, $\ISA$, and
$\LCP$. While these can be built in $O(r\log(n/r))$ space, the suffix tree
construction will be dominated by the $O(n)$ space needed
to build the RLCFG on $\DLCP$ in Lemmas~\ref{lem:gnp18} and 
\ref{lem:dlcp-grammar}. Therefore, we build $\SA$, $\ISA$, and $\LCP$ in
$O(n)$ time and space. 

Starting from the plain array $\DLCP[1..n]$,
the RLCFG is built in $O(\log(n/r))$ passes of the 
$O(n)$-time algorithm of Je\.z \cite{Jez15}. This includes identifying the
repeated pairs, which can also be done in $O(n)$ time via radix sort.
The total time is also $O(n)$, because the lengths of the strings compressed 
in the consecutive passes decrease exponentially. 

All the fields $l$, $d$, $p$, $m$, $n$, etc.\ stored for the nonterminals are 
easily computed in $O(r\log(n/r)) \subseteq O(n)$ time, that
is, $O(1)$ per nonterminal.  The arrays $L$, $A$, and $M$ are computed in 
$O(r)$ time and space.  The structure RMQ$_M$ is built in $O(r)$ time and bits 
\cite{FH11}. Finally, the structures used for solving $\PSV'$ and $\NSV'$ 
queries on $\DLCP'$ (construction of the tree for the weighted level-ancestor
queries \cite{FH11}, supporting the queries themselves \cite{KL07}, and the 
simplification
for $\PSV/\NSV$ \cite{FMN09}), as well as the approximate median of the minima 
\cite{FH10}, are built in $O(r)$ time and space, as shown by their authors. 

The construction of the predecessor data structures for the weighted
level-ancestor queries requires creating several structures with $O(r)$
elements in total, on universes of size $n$, having at least $n/r$ elements in
each structure. The total construction time is then $O(r\log\log_w r)$.
Note that this predecessor structure is needed only for $\PSV'/\NSV'$, not for
$\PSV/\NSV$, and thus it can be omitted unless we need the operation
{\em LAQ}$_S$.

In addition, the suffix tree requires the construction of the compressed
representation of $\PTDE$ \cite{FMN09}. This is easily done in $O(n)$ space 
and time by traversing a classical suffix tree.

We note that, with $O(n/B+\log(n/r))$ I/Os (where $B$ is the external memory
block size), we can build most of the suffix tree in main memory space $O(B
+ r\log(n/r))$. The main bottleneck is the 
algorithm of Je\.z \cite{Jez15}. The algorithm starts with two sequential passes
on $\DLCP$, first identifying runs of equal cells (to collapse them into one
symbol using a rule of the form $X \rightarrow Y^t$) and second collecting all
the distinct pairs of consecutive symbols (to create some rules of the form $X
\rightarrow YZ$). Both kinds of rules will add up to $O(r)$ per pass, so the 
distinct pairs can be stored in a balanced tree in main memory using $O(r)$ 
space. Once the pairs to replace are defined (in $O(r)$ time), the algorithm
traverses the text once again, doing the replacements. The new array is of
length at most $(3/4)n$; repeating this process $O(\log(n/r))$ times will
yield an array of size $O(r)$, and then we can finish. By streaming the
successively smaller versions of the array to external memory, we obtain the
promised I/Os and main memory space. The computation time is dominated by the
cost of building the structures $\SA$, $\ISA$, and $\LCP$ in $O(r\log(n/r))$ 
space: $O(n(\log r + \log\log_w(n/r))$. The balanced tree operations add
another $O(n\log r)$ time to this complexity.

The other obstacle is the construction of $\PTDE$. This can be done in
$O(Sort(n))$ I/Os, $O(n)$ computation, and $O(r)$ additional main memory space 
by emulating the linear-time algorithm to build the suffix tree topology from 
the $\LCP$
array \cite{KLAAP01}. This algorithm traverses $\LCP$ left to right, and
maintains a stack of the internal nodes in the rightmost path of the suffix 
tree known up to now, each with its string depth (the stack is easily 
maintained on disk with $O(n/B)$ I/Os). Each new $\LCP[p]$ cell 
represents a new suffix tree leaf. For each such leaf, we pop nodes from the 
stack until we find a node whose string depth is $\le \LCP[p]$. The sequence
of stack sizes is the array $\TDE$. We write those $\TDE$ entries to disk
as they are generated, left to right, in the format $\langle \TDE[p],\SA[p]
\rangle$. Once this array is generated on disk, we sort it by the second
component, and then the sequence of first components is array $\PTDE$. This
array is then read from disk left to right, as we simultaneously fill the
run-length compressed bitvector that represents it in $O(r)$ space 
\cite{FMN09}. The left-to-right traversal of $\LCP$ and $\SA$ is done in
$O(n)$ time by accessing their compressed representation by chunks of 
$\log(n/r)$ cells, using
Theorem~\ref{thm:dsa} and Lemma~\ref{lemma: lcp} with $s=\log(n/r)$.

\section{Conclusions}
\label{sec:conclusion}

We have closed the long-standing problem of efficiently locating the occurrences
of a pattern in a text using an index whose space is bounded by the number of
equal-letter runs in the Burrows-Wheeler transform (BWT) of the text. 
The $occ$ occurrences of a pattern $P[1..m]$ in a text $T[1..n]$ over alphabet 
$[1..\sigma]$ whose BWT has $r$ runs can be counted in time
$O(m\log\log_w(\sigma+n/r))$ and then located in 
$O(occ\log\log_w(n/r))$ time, on a $w$-bit RAM machine, using an $O(r)$-space 
index. Using space $O(r\log\log_w(\sigma+n/r))$, the counting and
locating times are reduced to $O(m)$ and $O(occ)$, respectively, which is
optimal in the general setting. Further, using 
$O(rw\log_\sigma\log_w(\sigma+n/r))$ space we can also obtain optimal 
time in the packed setting, replacing $O(m)$ by $O(\lceil m\log(\sigma)/w\rceil)$ in the counting time. Our findings also include $O(r\log(n/r))$-space structures 
to access consecutive entries of the text, suffix array, inverse suffix array, 
and longest common prefix array, in optimal time plus a per-query penalty of
$O(\log(n/r))$. We upgraded those
structures to a full-fledged compressed suffix tree working in $O(r\log(n/r))$ 
space and carrying out most navigation operations in time $O(\log(n/r))$. All
the structures can be built in times ranging from $O(n)$ worst-case to
$O(nw^{1+\epsilon})$ expected time and $O(n)$ space, and many can be built 
within the same asymptotic space of the final solution plus a single pass over 
the text.

The number of runs in the BWT is an important measure of the compressibility
of highly repetitive text collections, which can be compressed 
by orders of magnitude by exploiting the repetitiveness. While the 
first index of this type \cite{MNSVrecomb09,MNSV09}
managed to exploit the BWT runs, it was not able to locate 
occurrences efficiently. This gave rise to many other indexes based on other
measures, like the size of a Lempel-Ziv parse \cite{LZ76}, the size of a 
context-free grammar \cite{KY00}, the size of the smallest compact automaton
recognizing the text substrings \cite{BBHMCE87},
etc. While the complexities are not always comparable \cite{GNPlatin18},
the experimental results show that our proof-of-concept implementation 
outperforms all the space-efficient alternatives by one or two orders 
of magnitude in locating time.

This work triggered several other lines of research. From the idea of cutting
the text into phrases defined by BWT run ends, we showed that a run-length 
context-free grammar (RLCFG) of size $O(r\log(n/r))$ can be built on the text 
by using locally consistent parsing \cite{Jez15}. 
This was generalized to a RLCFG built on top of any bidirectional macro scheme 
(BMS) \cite{SS82}, which allowed us to prove bounds on the Lempel-Ziv 
approximation to the optimal BMS, as well as several other related bounds 
between compressibility measures \cite{GNPlatin18,NP18.2}. Also, the idea that 
at least one occurrence of any text substring must cross a phrase boundary led 
Kempa and Prezza \cite{KP18} to the concept of {\em string attractor}, a set of
$\gamma$ text positions with such a property. They prove that string attractors
subsume all the other measures of repetitiveness (i.e., $\gamma\le\min(r,z,g)$),
and design universal data structures of size $O(\gamma\log(n/\gamma))$ for 
accessing the compressed text, analogous to ours. Navarro and Prezza then 
extend these ideas to the 
first self-index on attractors \cite{NP18}, of size $O(\gamma\log(n/\gamma))
\subseteq O(r\log(n/r))$, yet they do not obtain our optimal query times.

On the other hand, some questions remain open, in particular regarding 
the operations that can be supported within $O(r)$ space. We have shown that 
this space is not only sufficient to represent the text, but also to 
efficiently count and locate pattern occurrences. We required, however,
$O(r\log(n/r))$ space to provide random access to the text. This raises
the question of whether efficient random access is possible within $O(r)$ space.
For example, recalling Table~\ref{tab:related}, random access in sublinear time
is possible within $O(g)$ space ($g$ being the size of the smallest grammar) 
but it has only been achieved in $O(z\log(n/z))$ space ($z \le g$ being the 
size of the Lempel-Ziv parse); recall that $r$ is incomparable with $g$ and $z$.
On the other hand, random access is possible within $O(\gamma\log(n/\gamma))$ 
space for any attractor of size $\gamma$, as explained.
A more specific question, but still intriguing, is whether we can
provide random access to the suffix array of the text in $O(r)$ space: note 
that we can return the cells {\em that result from a pattern search} within
this space, but accessing an arbitrary cell requires $O(r\log(n/r))$ space,
and this translates to the size required by a suffix tree. On the other hand,
it seems unlikely that one can provide suffix array or tree functionality 
within space related to $g$, $z$, or $\gamma$, since these measures are not 
related to the structure of the suffix array: this is likely to be a specific 
advantage of measure $r$.

Finally, we are working on converting our index into an actual software tool 
for handling large repetitive text collections, and in particular integrating
it into widely used bioinformatic software. This entails some further 
algorithmic challenges. One is to devise practical algorithms for building the 
BWT of very large repetitive datasets within space bounded by the 
repetitiveness. While our results in Section~\ref{sec:construction} are at a
theoretical stage, recent work by Boucher et al.~\cite{BGKM18} may be relevant.
Offering efficient techniques to insert new sequences in an existing
index are also important in a practical context; there is also some progress
in this direction~\cite{BGI18}. Another important aspect is, as explained in 
Section~\ref{sec:experiments}, making the
index less sensitive to lower repetitiveness scenarios, as it could be the
case of indexing short sequences (e.g., sets of reads) or metagenomic
collections. We are working on a hybrid with the classic FM-index to handle in
different ways the areas with higher and lower repetitiveness. Finally,
extending our index to enable full suffix tree functionality will require, 
despite our theoretical achievements in Section~\ref{sec:stree}, a significant
amount of algorithm engineering to obtain good practical space figures.

\section*{Acknowledgements}

We thank Ben Langmead for his pointers to the prevalence of the FM-index in
bioinformatics software, and Meg Gagie for checking our grammar.

\bibliographystyle{plain}
\bibliography{paper}

\end{document}